\documentclass[11pt]{article}

\usepackage[margin=1in]{geometry}

\usepackage{amssymb,amsmath,amsthm,amsfonts}
\usepackage{mathtools}
\usepackage{enumitem}
\usepackage[numbers,comma,sort&compress]{natbib}
\usepackage{authblk}
\usepackage{graphicx}
\usepackage[font=small]{caption}
\usepackage[labelformat=simple]{subcaption}
\usepackage{float}
\usepackage[ruled,vlined,linesnumbered]{algorithm2e}
\usepackage{footnote}
\usepackage{xcolor}

\usepackage[hidelinks]{hyperref}

\newtheorem{fact}{Fact}[section]

\newtheorem{theorem}{Theorem}
\newtheorem{definition}{Definition}
\newtheorem{lemma}{Lemma}
\newtheorem{proposition}{Proposition}

\newtheorem{corollary}{Corollary}
\newtheorem{remark}{Remark}

\numberwithin{equation}{section}
\numberwithin{theorem}{section}
\numberwithin{definition}{section}
\numberwithin{lemma}{section}
\numberwithin{proposition}{section}
\numberwithin{example}{section}
\numberwithin{corollary}{section}
\numberwithin{remark}{section}

\newcommand{\eqn}[1]{(\ref{eqn:#1})}
\newcommand{\thm}[1]{\hyperref[thm:#1]{Theorem~\ref*{thm:#1}}}
\newcommand{\cor}[1]{\hyperref[cor:#1]{Corollary~\ref*{cor:#1}}}
\newcommand{\defn}[1]{\hyperref[defn:#1]{Definition~\ref*{defn:#1}}}
\newcommand{\lem}[1]{\hyperref[lem:#1]{Lemma~\ref*{lem:#1}}}
\newcommand{\prop}[1]{\hyperref[prop:#1]{Proposition~\ref*{prop:#1}}}
\newcommand{\fig}[1]{\hyperref[fig:#1]{Figure~\ref*{fig:#1}}}
\newcommand{\tab}[1]{\hyperref[tab:#1]{Table~\ref*{tab:#1}}}
\newcommand{\algo}[1]{\hyperref[algo:#1]{Algorithm~\ref*{algo:#1}}}
\renewcommand{\sec}[1]{\hyperref[sec:#1]{Section~\ref*{sec:#1}}}
\newcommand{\append}[1]{\hyperref[append:#1]{Appendix~\ref*{append:#1}}}
\newcommand{\fac}[1]{\hyperref[fac:#1]{Fact~\ref*{fac:#1}}}

\makesavenoteenv{tabular}

\mathtoolsset{centercolon}

\DeclarePairedDelimiterX\braket[2]{\langle}{\rangle}{#1 \delimsize\vert #2}

\def\>{\rangle}
\def\<{\langle}

\newcommand{\N}{\mathbb{N}}
\newcommand{\Z}{\mathbb{Z}}
\newcommand{\R}{\mathbb{R}}
\newcommand{\C}{\mathbb{C}}

\DeclareMathOperator{\var}{Var}
\DeclareMathOperator{\poi}{Poi}
\DeclareMathOperator{\supp}{Supp}
\newcommand{\E}{\mathbb{E}}

\def\:{\hbox{\bf:}}
\newcommand{\range}[1]{[#1]}

\newcommand{\hd}[1]{\vspace{2mm} \noindent \textbf{#1}}
\def \eps {\epsilon}

\let\oldnl\nl
\newcommand{\nonl}{\renewcommand{\nl}{\let\nl\oldnl}}

\begin{document}

\begin{titlepage}
\clearpage

\title{Quantum query complexity of entropy estimation}

\author{Tongyang Li\thanks{tongyang@cs.umd.edu}\qquad\qquad Xiaodi Wu\thanks{xwu@cs.umd.edu} \\
\small{Department of Computer Science, Institute for Advanced Computer Studies, and \\ Joint Center for Quantum Information and Computer Science, University of Maryland, USA}}

\date{}

\maketitle
\thispagestyle{empty}

\begin{abstract}
Estimation of Shannon and R{\'e}nyi entropies of unknown discrete distributions is a fundamental problem in statistical property testing and an active research topic in both theoretical computer science and information theory. Tight bounds on the number of samples to estimate these entropies have been established in the classical setting, while little is known about their quantum counterparts. In this paper, we give the first quantum algorithms for estimating $\alpha$-R{\'e}nyi entropies (Shannon entropy being 1-Renyi entropy). In particular, we demonstrate a quadratic quantum speedup for Shannon entropy estimation and a generic quantum speedup for $\alpha$-R{\'e}nyi entropy estimation for all $\alpha\geq 0$, including a tight bound for the collision-entropy (2-R{\'e}nyi entropy). We also provide quantum upper bounds for extreme cases such as the Hartley entropy (i.e., the logarithm of the support size of a distribution, corresponding to $\alpha=0$) and the min-entropy case (i.e., $\alpha=+\infty$), as well as the Kullback-Leibler divergence between two distributions. Moreover, we complement our results with quantum lower bounds on $\alpha$-R{\'e}nyi entropy estimation for all $\alpha
\geq 0$.

Our approach is inspired by the pioneering work of Bravyi, Harrow, and Hassidim (BHH)~\cite{bravyi2011quantum} on quantum algorithms for distributional property testing, however, with many new technical ingredients. For Shannon entropy  and 0-R{\'e}nyi entropy estimation, we improve the performance of the BHH framework, especially its error dependence, using Montanaro's approach to estimating the expected output value of a quantum subroutine with bounded variance \cite{montanaro2015quantum} and giving a fine-tuned error analysis. For general $\alpha$-R{\'e}nyi entropy estimation, we further develop a procedure that recursively approximates $\alpha$-R{\'e}nyi entropy for a sequence of $\alpha$s, which is in spirit similar to a cooling schedule in simulated annealing. For special cases such as integer $\alpha\geq 2$  and $\alpha=+\infty$ (i.e., the min-entropy), we reduce the entropy estimation problem to the $\alpha$-distinctness and the $\lceil\log n\rceil$-distinctness problems, respectively. We exploit various techniques to obtain our lower bounds for different ranges of $\alpha$, including reductions to (variants of) existing lower bounds in quantum query complexity as well as the polynomial method inspired by the celebrated quantum lower bound for the collision problem.
\end{abstract}

\end{titlepage}


\section{Introduction}
\label{sec:intro}

\vspace{-2mm}
\hd{Motivations.}
Property testing is a rapidly developing field in theoretical computer science (e.g. see the survey~\cite{TCS-029}). It aims to determine properties of an object with the least number of independent samples of the object. Property testing is a theoretically appealing topic with intimate connections to statistics, learning theory, and algorithm design. One important topic in property testing is to estimate statistical properties of unknown distributions (e.g., \cite{valiant2011testing}), which are fundamental questions in statistics and information theory, given that much of science relies on samples furnished by nature. The Shannon \cite{ShannonEntropy} and R{\'e}nyi \cite{renyi1961measures} entropies are central measures of randomness compressibility. In this paper, we focus on estimating these entropies for an unknown distribution.

Specifically, given a distribution $p$ over a set $X$ of size $n$ (w.l.o.g. let $X=[n]$) where $p_x$ denotes the probability of $x \in X$, the \emph{Shannon entropy} $H(p)$ of this distribution $p$ is defined by
\begin{align}\label{eqn:Shannon-definition}
  H(p):= \sum_{x \in X:\,p_{x}>0} p_x \log \Big(\frac{1}{p_x}\Big).
\end{align}
A natural question is to determine the \emph{sample complexity} (i.e., the necessary number of independent samples from $p$) to estimate $H(p)$, with error $\eps$ and high probability. This problem has been intensively studied in the classical literature. For multiplicative error $\eps$, Batu et al. \cite[Theorem 2]{batu2005complexity} provided the upper bound of $O(n^{(1+o(1))/(1+\epsilon)^{2}}\log n)$, while an almost matching lower bound of $\Omega(n^{(1-o(1))/(1+\epsilon)^{2}})$ was shown by Valiant \cite[Theorem 1.3]{valiant2011testing}. For additive errors, Paninski gave a nonconstructive proof of the existence of sublinear estimators in~\cite{paninski2003estimation,paninski2004estimating}, while an explicit construction using $\Theta(n/\log n)$ samples was shown by Valiant and Valiant in~\cite{valiant2011estimating} when $\epsilon>n^{-0.03}$; for the case $\epsilon\leq n^{-0.03}$, Wu and Yang \cite{wu2016minimax} and Jiao et al. \cite{jiao2015minimax} gave the optimal estimator with $\Theta(\frac{n}{\epsilon\log n}+\frac{(\log n)^{2}}{\epsilon^{2}})$ samples. A sequence of works in information theory~\cite{jiao2015minimax,wu2016minimax,jiao2014maximum} studied the minimax mean-squared error, which becomes $O(1)$ also using $\Theta(n/\log n)$ samples.

One important generalization of Shannon entropy is the \emph{R{\'e}nyi entropy} of order $\alpha>0$, denoted $H_\alpha(p)$, which is defined by
\begin{align}\label{eqn:Renyi-definition}
H_\alpha(p):= \begin{cases}
  \frac{1}{1-\alpha} \log \sum_{x \in X} p_x^\alpha, & \text{when $\alpha \neq 1$.} \\
  \lim_{\alpha \rightarrow 1} H_{\alpha}(p), & \text{when $\alpha=1$.}
\end{cases}
\end{align}
The R{\'e}nyi entropy of order 1 is simply the Shannon entropy, i.e., $H_1(p)=H(p)$. General R{\'e}nyi entropy can be used as a bound on Shannon entropy, making it useful in many applications (e.g., \cite{arikan1996inequality,csiszar1995generalized}). R{\'e}nyi entropy is also of interest in its own right. One prominent example is the R{\'e}nyi entropy of order 2, $H_2(p)$ (also known as the \emph{collision entropy}), which measures the quality of random number generators (e.g.,~\cite{vanOorschot1999}) and key derivation in cryptographic applications (e.g.,~\cite{BBCM95, IZ89}). Motivated by these and other applications, the estimation of R{\'e}nyi entropy has also been actively studied~\cite{acharya2017estimating, jiao2015minimax,jiao2014maximum}. In particular,  Acharya et al.~\cite{acharya2017estimating} have shown almost tight bounds on the classical query complexity of computing R{\'e}nyi entropy. Specifically,  for any non-integer $\alpha>1$, the classical query complexity of $\alpha$-R{\'e}nyi entropy is $\Omega(n^{1-o(1)})$ and $O(n)$. Surprisingly, for any \emph{integer} $\alpha>1$, the classical query complexity is $\Theta(n^{1-1/\alpha})$, i.e., \emph{sublinear} in $n$. When $0\leq\alpha<1$, the classical query complexity is $\Omega(n^{1/\alpha-o(1)})$ and $O(n^{1/\alpha})$, which is always superlinear.

The extreme case ($\alpha\rightarrow\infty$) is known as the \emph{min-entropy}, denoted $H_{\infty}(p)$, which is defined by
\begin{align}\label{eqn:min-entropy-definition}
H_{\infty}(p):=\lim_{\alpha\rightarrow\infty}H_{\alpha}(p)=-\log\max_{i\in\range{n}}p_{i}.
\end{align}
Min-entropy plays an important role in the randomness extraction (e.g.,~\cite{TCS-010}) and characterizes the maximum number of uniform bits that can be extracted from a given distribution. Classically, the query complexity of min-entropy estimation is $\Theta(n/\log n)$, which follows directly from \cite{valiant2011estimating}.

Another extreme case ($\alpha=0$), also known as the \emph{Hartley entropy} \cite{hartley1928transmission}, is the logarithm of the support size of distributions, where the \emph{support} of any distribution $p$ is defined by
\begin{align}\label{eqn:support-definition}
  \supp(p):= |\{x: x\in X,\,p_{x}>0\}|.
\end{align}
It is a natural and fundamental quantity of distributions with various applications (e.g.,~\cite{efron1976estimating,thisted1987did, haas1995sampling,florencio2007large,kroes1999bacterial,paster2001bacterial,hughes2001counting}).
However, estimating the support size is impossible in general because elements with negligible but nonzero probability, which are very unlikely to be sampled, could still contribute to $\supp(p)$. Two related quantities (\emph{support coverage} and \emph{support size}) have hence been considered as alternatives of 0-R{\'e}nyi entropy with roughly $\Theta(n/\log(n))$ complexity. (See details in \sec{support}.)

Besides the entropic measures of a discrete distribution, we also briefly discuss an entropic measure between two distributions, namely the \emph{Kullback-Leibler (KL) divergence}. Given two discrete distributions $p$ and $q$ with cardinality $n$, the KL divergence is defined as
\begin{align}
D_{\textrm{KL}}(p\|q)=\sum_{i\in\range{n}}p_{i}\log\frac{p_{i}}{q_{i}}.
\end{align}
KL divergence is a key measure with many applications in information theory \cite{kullback1997information,csiszar2011information}, data compression \cite{catoni2004statistical}, and learning theory \cite{kingma2013auto}. Classically, under the assumption that $\frac{p_{i}}{q_{i}}\leq f(n)\ \forall\,i\in\range{n}$ for some $f(n)$, $D_{\textrm{KL}}(p\|q)$ can be approximated within constant additive error with high success probability if $\Theta(\frac{n}{\log n})$ samples are taken from $p$ and $\Theta(\frac{nf(n)}{\log n})$ samples are taken from $q$.

\hd{Main question.} In this paper, we study the impact of quantum computation on estimation of general R{\'e}nyi entropies. Specifically, we aim to characterize
\emph{quantum speed-ups for estimating Shannon and R{\'e}nyi entropies}.

Our question aligns with the emerging topic called ``quantum property testing" (see the survey \cite{montanaro2013survey}) and focuses on investigating the quantum advantage in testing classical statistical properties. To the best of our knowledge, the first research paper on distributional quantum property testing is by Bravyi, Harrow, and Hassidim (BHH)~\cite{bravyi2011quantum}, where they discovered quantum speedups for testing uniformity, orthogonality, and statistical difference on unknown distributions. Some of these results were subsequently improved by Chakraborty et al.~\cite{chakraborty2010new}.
Reference \cite{bravyi2011quantum} also claimed that Shannon entropy could be estimated with query complexity $O(\sqrt{n})$, however, without details and explicit error dependence. Indeed, our framework is inspired by \cite{bravyi2011quantum}, but with significantly new ingredients to achieve our results.
There is also a related line of research on spectrum testing or tomography of quantum states~\cite{OW15, OW16, HHJWY16, OW17}. However, these works aim to test properties of general quantum states, while we focus on using quantum algorithms to test properties of classical distributions (i.e., diagonal quantum states)\footnote{Note that one can also leverage the results of~\cite{OW15, OW16, HHJWY16, OW17} to test properties of classical distributions. However, they are less efficient because they deal with a much harder problem involving general quantum states.}.

\hd{Distributions as oracles.} The sampling model in the classical literature assumes that a tester is presented with independent samples from an unknown distribution. One of the contributions of BHH is an alternative model that allows coherent quantum access to unknown distributions. Specifically, BHH models a discrete distribution $p=(p_{i})_{i=1}^{n}$ on $\range{n}$ by an \emph{oracle} $O_{p}\colon \range{S}\rightarrow\range{n}$ for some $S\in\N$. The probability $p_i$ ($i \in [n]$) is proportional to the size of pre-image of $i$ under $O_p$. Namely, an oracle $O_{p}\colon \range{S}\rightarrow\range{n}$ generates $p$ if and only if for all $i\in\range{n}$,
\begin{align}\label{eqn:BHH-Op}
p_{i}=|\{s\in\range{S}: O_{p}(s)=i\}|/S.
\end{align}
(note that we assume $p_i$s to be rational numbers). If one samples $s$ uniformly from $\range{S}$, then the output $O_{p}(s)$ is from distribution $p$. Instead of considering sample complexity---that is, the number of used samples---we consider the \emph{query complexity} in the oracle model that counts the number of oracle uses. Note that a tester interacting with an oracle can potentially be more powerful due to the possibility of learning the internal structure of the oracle as opposed to the sampling model. However, it is shown in~\cite{bravyi2011quantum} that the query complexity of the oracle model and the sample complexity of the sampling model are in fact the same classically.

A significant advantage of the oracle model is that it naturally allows coherent access when extended to the quantum case, where we transform $O_{p}$ into a unitary operator $\hat{O}_{p}$ acting on $\C^{S}\otimes\C^{n+1}$ such that
\begin{align}\label{eqn:BHH-query-model}
\hat{O}_{p}|s\>|0\>=|s\>|O_{p}(s)\>\quad\forall\,s\in\range{S}.
\end{align}
Moreover, this oracle model can also be readily obtained in some algorithmic settings, e.g., when distributions are generated by some classical or quantum sampling procedure. Thus, statistical property testing results in this oracle model can be potentially leveraged in algorithm design.

\hd{Our Results}. Our main contribution is a systematic study of both upper and lower bounds for the \emph{quantum query complexity} of estimation of R{\'e}nyi entropies (including Shannon entropy as a special case). Specifically, we obtain the following quantum speedups for different ranges of $\alpha$.

\begin{theorem}
There are quantum algorithms that approximate $H_\alpha(p)$ of distribution $p$ on $\range{n}$ within an additive error $0<\epsilon\leq O(1)$ with success probability at least 2/3 using\footnote{It should be understood that the success probability $2/3$ can be boosted to close to 1 without much overhead, e.g., see \lem{Renyi-large-4} in \sec{Renyi-key-inequality}.}
\begin{itemize}
  \item $\tilde{O}\big(\frac{\sqrt{n}}{\epsilon^{1.5}}\big)$ quantum queries when $\alpha=0$, i.e., Hartley entropy. See \thm{size-upper-bound}.\footnote{0-R{\'e}nyi entropy estimation is intractable without any assumption, both classically and quantumly. Here, the results are based on the assumption that nonzero probabilities are at least $1/n$. See \sec{support} for more information.}
  \item $\tilde{O}\big(\frac{n^{1/\alpha-1/2}}{\epsilon^{2}}\big)$ quantum queries\footnote{$\tilde{O}$ hides factors that are polynomial in $\log n$ and $\log 1/\epsilon$.} when $0<\alpha<1$.  See \thm{Renyi-small}.
  \item $\tilde{O}\big(\frac{\sqrt{n}}{\epsilon^{2}}\big)$ quantum queries when $\alpha=1$, i.e., Shannon entropy. See \thm{Shannon}.
  \item $\tilde{O}\big(\frac{n^{\nu(1-1/\alpha)}}{\epsilon^{2}}\big)$ quantum queries when $\alpha>1,\alpha\in\N$ for some $\nu<\frac{3}{4}$. See \thm{Renyi-integer}.
  \item $\tilde{O}\big(\frac{n^{1-1/2\alpha}}{\epsilon^{2}}\big)$ quantum queries when $\alpha>1,\alpha\notin\N$. See \thm{Renyi-large}.
  \item $\tilde{O}\big(Q(\big\lceil\frac{16\log n}{\epsilon^{2}}\big\rceil\textsf{-distinctness})\big)$ quantum queries when $\alpha=\infty$, where $Q(\big\lceil\frac{16\log n}{\epsilon^{2}}\big\rceil\textsf{-distinctness})$ is the quantum query complexity of the $\big\lceil\frac{16\log n}{\epsilon^{2}}\big\rceil$-distinctness problem. See \thm{min-entropy}.
\end{itemize}
\end{theorem}

Our quantum testers demonstrate advantages over classical ones for all $0<\alpha<\infty$; in particular, our quantum tester has a quadratic speedup in the case of Shannon entropy. When $\alpha=\infty$, our quantum upper bound depends on the quantum query complexity of the $\lceil\log n\rceil$-distinctness problem, which is open to the best of our knowledge\footnote{Existing quantum algorithms for the $k$-distinctness problem (e.g., \cite{ambainis2007quantum} has query complexity $O(k^{2}n^{k/k+1})$ and \cite{belovs2012learning} has query complexity $O(2^{k^{2}}n^{\nu})$ for some $\nu<3/4$) do not behave well for super-constant $k$s.} and might demonstrate a quantum advantage.

As a corollary, we also obtain quadratic quantum speedup for estimating KL divergence:
\begin{corollary}[see \thm{KL-divergence}]
Assuming $p$ and $q$ satisfies $\frac{p_{i}}{q_{i}}\leq f(n)\ \forall\,i\in\range{n}$ for some function $f:\N\rightarrow \R^{+}$, $D_{\textrm{KL}}(p\|q)$, there is a quantum algorithm that approximates $D_{\textrm{KL}}(p\|q)$ within an additive error $\epsilon>0$ with success probability at least $\frac{2}{3}$ using $\tilde{O}\big(\frac{\sqrt{n}}{\epsilon^{2}}\big)$ quantum queries to $p$ and $\tilde{O}\big(\frac{\sqrt{n}f(n)}{\epsilon^{2}}\big)$ quantum queries to $q$.
\end{corollary}

We also obtain corresponding quantum lower bounds on entropy estimation as follows. We summarize both bounds in~\tab{Renyi} and visualize them in~\fig{Renyi}.

\begin{theorem}[See \thm{lower-bound-rewrite}]\label{thm:lower-bound}
Any quantum algorithm that approximates $H_\alpha(p)$ of distribution $p$ on $\range{n}$ within additive error $\epsilon$ with success probability at least 2/3 must use
\begin{itemize}
   \item $\Omega(n^{\frac{1}{3}}/\epsilon^{\frac{1}{6}})$ quantum queries when $\alpha=0$, assuming $1/n\leq\epsilon\leq 1$.
   \item $\tilde{\Omega}(n^{\frac{1}{7\alpha}-o(1)}/\epsilon^{\frac{2}{7}})$ quantum queries when $0<\alpha<\frac{3}{7}$.
   \item $\Omega(n^{\frac{1}{3}}/\epsilon^{\frac{1}{6}})$ quantum queries when $\frac{3}{7}\leq\alpha\leq 3$, assuming $1/n\leq\epsilon\leq 1$.
   \item $\Omega(n^{\frac{1}{2}-\frac{1}{2\alpha}}/\epsilon)$ quantum queries when $3\leq\alpha<\infty$.
   \item $\Omega(\sqrt{n}/\epsilon)$ quantum queries when $\alpha=\infty$.
\end{itemize}
\end{theorem}

\begin{table}
\centering
\resizebox{0.75\columnwidth}{!}{%
\begin{tabular}{|c|c|c|c|}
\hline
$\alpha$ & classical bounds & quantum bounds ({\color{red} this paper}) \\ \hline\hline
$\alpha=0$ & $\Theta(\frac{n}{\log n})$ \cite{wu2015chebyshev,orlitsky2016optimal} & $\tilde{O}(\sqrt{n})$, $\Omega(n^{\frac{1}{3}})$ \\ \hline
$0<\alpha<1$ & $O(\frac{n^{\frac{1}{\alpha}}}{\log n})$, $\Omega(n^{\frac{1}{\alpha}-o(1)})$ \cite{acharya2017estimating} & $\tilde{O}(n^{\frac{1}{\alpha}-\frac{1}{2}})$, $\Omega(\max\{n^{\frac{1}{7\alpha}-o(1)},n^{\frac{1}{3}}\})$ \\ \hline
$\alpha=1$ & $\Theta(\frac{n}{\log n})$ \cite{valiant2011estimating,jiao2015minimax,wu2016minimax} & $\tilde{O}(\sqrt{n})$, $\Omega(n^{\frac{1}{3}})$
\\ \hline
$\alpha>1,\alpha\notin\N$ & $O(\frac{n}{\log n})$, $\Omega(n^{1-o(1)})$ \cite{acharya2017estimating} & $\tilde{O}(n^{1-
\frac{1}{2\alpha}}\big)$, $\Omega(\max\{n^{\frac{1}{3}},n^{\frac{1}{2}-\frac{1}{2\alpha}}\})$
 \\ \hline
$\alpha=2$ & $\Theta(\sqrt{n})$ \cite{acharya2017estimating} & $\tilde{\Theta}(n^{\frac{1}{3}})$
 \\ \hline
$\alpha>2,\alpha\in\N$ & $\Theta(n^{1-1/\alpha})$ \cite{acharya2017estimating} & $\tilde{O}(n^{\nu(1-1/\alpha)})$, $\Omega(n^{\frac{1}{2}-\frac{1}{2\alpha}})$, $\nu<3/4$
 \\ \hline
$\alpha=\infty$ & $\Theta(\frac{n}{\log n})$ \cite{valiant2011estimating} & $\tilde{O}(Q(\textsf{$\lceil\log n\rceil$-distinctness}))$, $\Omega(\sqrt{n})$
 \\ \hline
\end{tabular}
}
\vspace{1mm}
\caption{Summary of classical and quantum query complexity of $H_{\alpha}(p)$, assuming $\epsilon=\Theta(1)$.}
\label{tab:Renyi}
\end{table}

\begin{figure}[htbp]
\centering
\includegraphics[width=3.9in]{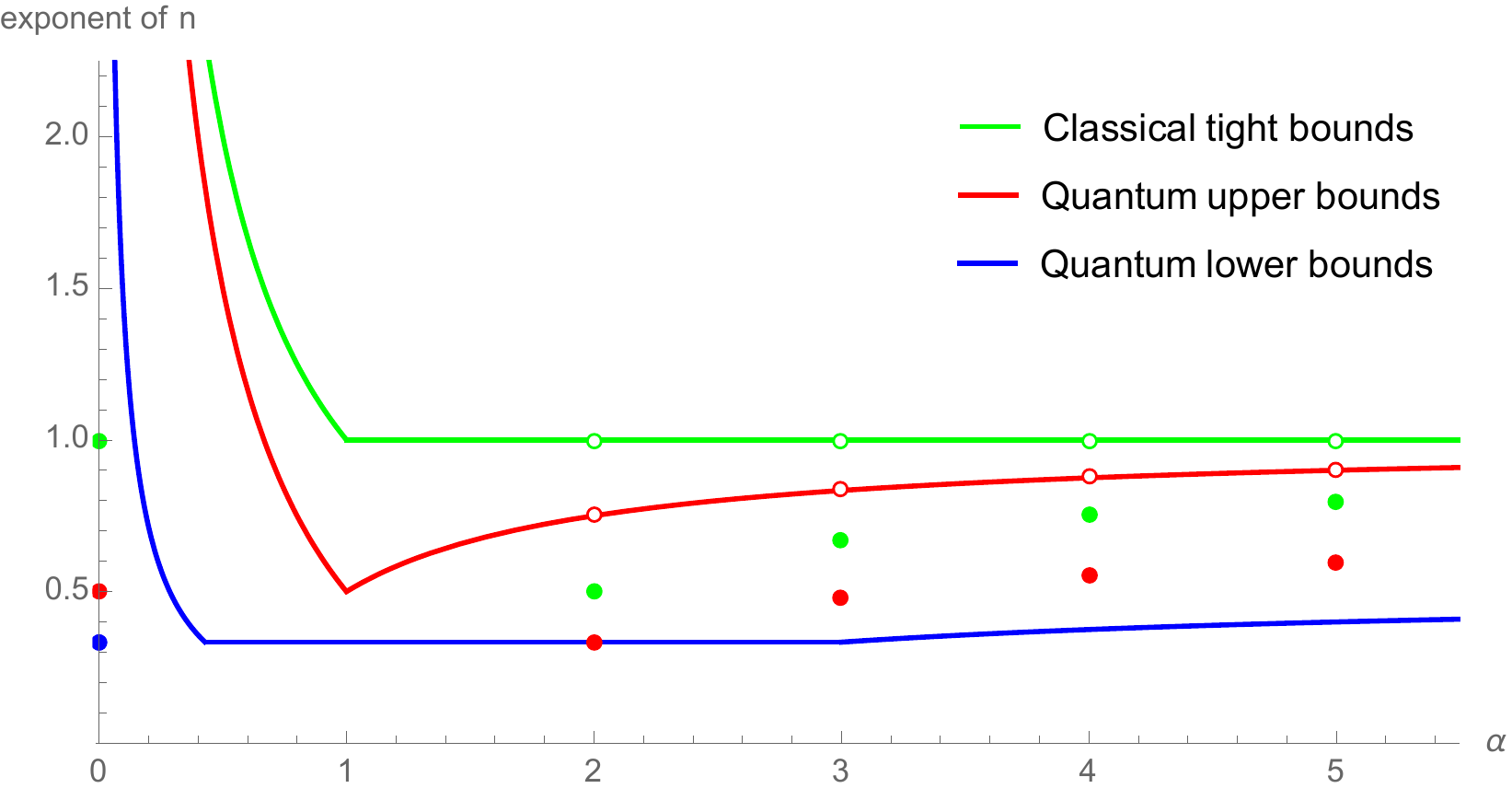}
\caption{Visualization of classical and quantum query complexity of $H_{\alpha}(p)$. The $x$-axis represents $\alpha$ and the $y$-axis represents the exponent of $n$. {\color{red}{Red}} curves and points represent quantum upper bounds. {\color{green}{Green}} curves and points represent classical tight bounds. The {\color{blue}{Blue}} curve represents quantum lower bounds.}
\label{fig:Renyi}
\end{figure}

\hd{Techniques.} At a high level, our upper bound is inspired by BHH~\cite{bravyi2011quantum}, where we formulate a framework (in Section~\ref{sec:master}) that generalizes the technique in BHH and makes it applicable in our case. Let $F(p)=\sum_{x} p_x f(p_x)$ for some function $f(\cdot)$ and distribution $p$. Similar to BHH, we design a master algorithm that samples $x$ from $p$ and then use the \emph{quantum counting primitive}~\cite{brassard2002quantum} to obtain an estimate $\tilde{p}_x$ of $p_x$ and outputs $f(\tilde{p}_x)$. It is easy to see that the expectation of the output of the master algorithm is roughly\footnote{The accurate expectation is $\sum_{x} p_x \E[f(\tilde{p}_x)]$. Intuitively, we expect $\tilde{p}_x$ to be a good estimate of $p_x$.} $F(p)$. By choosing appropriate $f(\cdot)$s, one can recover $H(p)$ or $H_{\alpha}(p)$ as well as the ones used in BHH. It suffices then to obtain a good estimate of the output expectation of the master algorithm, which was achieved by multiple independent runs of the master algorithm in BHH.

The performance of the above framework (and its analysis) critically depends on how close the expectation of the algorithm is to $F(p)$ and how concentrated the output distribution is around its expectation, which in turn heavily depends on the specific $f(\cdot)$ in use. Our first contribution is a fine-tuned error analysis for specific $f(\cdot)$s, such as in the case of Shannon entropy (i.e., $f(p_x)=-\log(p_x)$) whose values could be significant for boundary cases of $p_x$. Instead of only considering the case when $\tilde{p}_x$ is a good estimate of $p_x$ as in BHH, we need to analyze the entire distribution of $\tilde{p}_x$ using quantum counting. We also leverage a generic quantum speedup for estimating the expectation of the output of any quantum procedure with additive errors~\cite{montanaro2015quantum}, which significantly improves our error dependence as compared to BHH. These improvements already give a quadratic quantum speedup for Shannon (\sec{Shannon}) and 0-R{\'e}nyi (\sec{support}) entropy estimation. As an application, it also gives a quadratic speedup for estimating the KL-divergence between two distributions (see \sec{KL}).

For general $\alpha$-R{\'e}nyi entropy $H_{\alpha}(p)$, we choose $f(p_x)=p_x^{\alpha-1}$ and let $P_{\alpha}(p)=F(p)$ so that $H_{\alpha}(p) \propto \log P_{\alpha}(p)$. Instead of estimating $F(p)$ with \emph{additive errors} in the case of Shannon entropy, we switch to working with \emph{multiplicative errors} which is harder since the aforementioned quantum algorithm~\cite{montanaro2015quantum} is much weaker in this setting. Indeed, by following the same technique, we can only obtain quantum speedups for $\alpha$-R{\'e}nyi entropy when $1/2<\alpha<2$.

For general $\alpha>0$, our first observation is that if one knew the output expectation $\E[X]$ is within $[a,b]$ such that $b/a=\Theta(1)$, then one can slightly modify the technique in~\cite{montanaro2015quantum} (as shown in \thm{Monte-Carlo-multiplicative}) and obtain a quadratic quantum speedup similar to the additive error setting. This approach, however, seems circular since it is unclear how to obtain such $a,b$ in advance. Our second observation is that for any close enough $\alpha_1, \alpha_2$, $P_{\alpha_1}(p)$ can be used to bound $P_{\alpha_2}(p)$. Precisely, when $\alpha_1/\alpha_2 =1 \pm 1/\log(n)$, we have $P_{\alpha_1}(p)=\Theta(P_{\alpha_2}(p)^{\alpha_1/\alpha_2})$ (see \lem{Renyi-large-approx}). As a result, when estimating $P_{\alpha}(p)$, we can first estimate $P_{\alpha'}$ to provide a bound on $P_{\alpha}$, where  $\alpha', \alpha$ differ by a $1 \pm 1/\log(n)$ factor and $\alpha'$ moves toward 1. We apply this strategy recursively on estimating $P_{\alpha'}$ until $\alpha'$ is very close to 1 from above when initial $\alpha>1$ or from below when initial $\alpha<1$, where a quantum speedup is already known. At a high level, we recursively estimate a sequence (of size $O(\log n)$) of such $\alpha$s that eventually converges to 1, where in each iteration we establish some quantum speedup which leads to an overall quantum speedup. We remark that our approach is in spirit similar to the cooling schedules in simulated annealing (e.g.~\cite{stefankovic2009adaptive}). (See \sec{noninteger-Renyi}.)

For \emph{integer} $\alpha \geq 2$, we observe a connection between $P_{\alpha}(p)$ and the $\alpha$-distinctness problem which leads to a more significant quantum speedup.  Precisely, let $O_{p}: [S] \rightarrow [n]$ be the oracle in \eqn{BHH-query-model}, we observe that $P_{\alpha}(p)$ is proportional to the \emph{$\alpha$-frequency moment} of $O_{p}(1),\ldots,O_{p}(S)$ which can be solved quantumly~\cite{montanaro2015frequency} based on any quantum algorithm for the $\alpha$-distinctness problem (e.g.,~\cite{belovs2012learning}). However, there is a catch that a direct application of \cite{montanaro2015frequency} will lead to a dependence on $S$ rather than $n$. We remedy this situation by tweaking the algorithm and its analysis in~\cite{montanaro2015frequency} to remove the dependence on $S$ for our specific setting. (See \sec{integer-Renyi}.)

The integer $\alpha$ algorithm fails to extend to the \emph{min-entropy} case (i.e., $\alpha=+\infty$) because the hidden constant in $O(\cdot)$ has a poor dependence on $\alpha$ (see Remark \ref{remark:alpha_dependence}). Instead, we develop another reduction to the $\lceil\log n\rceil$-distinctness problem by exploiting the so-called ``Poissonized sampling" technique \cite{lecam2012asymptotic,valiant2011estimating,jiao2015minimax}. At a high level, we construct Poisson distributions that are parameterized by $p_i$s and leverage the ``threshold" behavior of Poisson distributions (see \lem{min-entropy}). Roughly, if $\max_i p_i$ passes some threshold, with high probability, these parameterized Poisson distributions will lead to a collision of size $\lceil\log n\rceil$ that will be caught by the $\lceil\log n\rceil$-distinctness algorithm. Otherwise, we run again with a lower threshold until the threshold becomes trivial. (See \sec{min-entropy}.)

Some of our lower bounds come from reductions to existing ones in quantum query complexity, such as the quantum-classical separation of symmetric boolean functions \cite{aaronson2014need}, the collision problem \cite{aaronson2004quantum,kutin2005quantum}, and the Hamming weight problem \cite{nayak1999quantum}, for different ranges of $\alpha$. We also obtain lower bounds with a better error dependence by the polynomial method, which is inspired by the celebrated quantum lower bound for the collision problem~\cite{aaronson2004quantum,kutin2005quantum}. (See \sec{lower-bound}.)

\hd{Open questions.} Our paper raises a few open questions. A natural question is to close the gaps between our quantum upper and lower bounds. Our quantum techniques on both ends are actually quite different from the state-of-the-art classical ones (e.g.,~\cite{valiant2011estimating}). It is interesting to see whether one can incorporate classical ideas to improve our quantum results.
It is also possible to achieve better lower bounds by improving our application of the polynomial method or exploiting the quantum adversary method (e.g., \cite{hoyer2007negative,belovs2013adversary}). Finally, our result motivates the study of the quantum algorithm for the $k$-distinctness problem with super-constant $k$, which might also be interesting by itself.

\hd{Notations.} Throughout the paper, we consider a discrete distribution $\{p_{i}\}_{i=1}^{n}$ on $\range{n}$, and $P_{\alpha}(p):=\sum_{i=1}^{n}p_{i}^{\alpha}$ represents the $\alpha$-power sum of $p$. In the analyses of our algorithms, `$\log$' is natural logarithm; `$\approx$' omits lower order terms.


\section{Master algorithm}\label{sec:master}
Let $p=(p_{i})_{i=1}^{n}$ be a discrete distribution on $\range{n}$ encoded by the quantum oracle $\hat{O}_{p}$ defined in \eqn{BHH-query-model}. Inspired by BHH, we develop the following master algorithm to estimate a property $F$ with the form $F(p):=\sum_{i\in\range{n}}p_{i}f(p_{i})$ for a function $f\colon (0,1]\rightarrow \R$.

\begin{algorithm}[htbp]
\SetKwFor{ForSubroutine}{Regard the following subroutine as $\mathcal{A}$:}{}{endfor}
Set $l,M\in\N$\;
\ForSubroutine{}
    {Draw a sample $i\in\range{n}$ according to $p$\;
    Use \textbf{EstAmp} or \textbf{EstAmp$'$} with $M$ queries to obtain an estimation $\tilde{p}_{i}$ of $p_{i}$\; \label{line:amplitude-estimation}
    Output $X=f(\tilde{p}_{i})$\; \label{line:random-variable}
}
Use $\mathcal{A}$ for $l$ executions in \thm{Monte-Carlo} or \thm{Monte-Carlo-multiplicative} and output $\tilde{F}(p)$ to estimate $F(p)$\; \label{line:Monte-Carlo}
\caption{Estimate $F(p)=\sum_{i} p_i f(p_{i})$ of a discrete distribution $p=(p_{i})_{i=1}^{n}$ on $\range{n}$.}
\label{algo:Master}
\end{algorithm}

Comparing to BHH, we introduce a few new technical ingredients in the design of  Algorithm~\ref{algo:Master} and its analysis, which significantly improve the performance of Algorithm~\ref{algo:Master} especially for specific $f(\cdot)$s in our case, e.g., $f(p_x)=-\log(p_x)$ (Shannon entropy) and $f(p_x)=p^{\alpha-1}_x$ (R{\'e}nyi entropy).

The \textbf{first} one is a generic quantum speedup of Monte Carlo methods~\cite{montanaro2015quantum}, in particular, a quantum algorithm that approximates the output expectation of a subroutine with additive errors that has a quadratic better sample complexity than the one implied by Chebyshev's inequality.

\begin{theorem}[Additive error; Theorem 5 of \cite{montanaro2015quantum}]\label{thm:Monte-Carlo}
Let $\mathcal{A}$ be a quantum algorithm with output $X$ such that $\var[X]\leq\sigma^{2}$. Then for $\epsilon$ where $0<\epsilon<4\sigma$, by using $O((\sigma/\epsilon)\log^{3/2}(\sigma/\epsilon)\log\log(\sigma/\epsilon))$ executions of $\mathcal{A}$ and $\mathcal{A}^{-1}$, Algorithm 3 in \cite{montanaro2015quantum} outputs an estimate $\tilde{\E}[X]$ of $\E[X]$ such that
\begin{align}
\Pr\big[\big|\tilde{\E}[X]-\E[X]\big|\geq\epsilon\big]\leq 1/5.
\end{align}
\end{theorem}

It is worthwhile mentioning that classically one needs to use $\Omega(\sigma^{2}/\epsilon^{2})$ executions of $\mathcal{A}$~\cite{dagum2000optimal} to estimate $\E[X]$. \thm{Monte-Carlo} demonstrates a quadratic improvement on the error dependence. In the case of approximating $H_{\alpha}(p)$, we need to work with multiplicative errors while existing results (e.g.~\cite{montanaro2015quantum}) have a worse error dependence which is insufficient for our purposes. Instead, inspired by \cite{montanaro2015quantum}, we prove the following theorem (our \textbf{second} ingredient) that takes auxiliary information about the range of $\E[X]$ into consideration, which might be of independent interest.

\begin{theorem}[Multiplicative error; \append{Monte-Carlo-multiplicative}] \label{thm:Monte-Carlo-multiplicative}
Let $\mathcal{A}$ be a quantum algorithm with output $X$ such that $\var[X]\leq\sigma^{2}\E[X]^{2}$ for a known $\sigma$. Assume that $\E[X]\in[a,b]$. Then for $\epsilon$ where $0<\epsilon<24\sigma$, by using $\mathcal{A}$ and $\mathcal{A}^{-1}$ for $O((\sigma b/\epsilon a)\log^{3/2}(\sigma b/\epsilon a)\log\log(\sigma b/\epsilon a))$ executions, \algo{Monte-Carlo-multiplicative} (given in \append{Monte-Carlo-multiplicative}) outputs an estimate $\tilde{\E}[X]$ of $\E[X]$ such that
\begin{align}
\Pr\big[\big|\tilde{\E}[X]-\E[X]\big|\geq\epsilon\E[X]\big]\leq 1/10.
\end{align}
\end{theorem}

The \textbf{third} ingredient is a fine-tuned error analysis due to the specific $f(\cdot)$s. Similar to BHH, we rely on quantum counting (named \textbf{EstAmp})~\cite{brassard2002quantum} to estimate the pre-image size of a Boolean function, which provides another source of quantum speedup. In particular, we approximate any probability $p_x$ in the query model (\eqn{BHH-query-model}) by $\tilde{p}_x$ by estimating the size of the pre-image of a Boolean function $\chi\colon\range{S}\rightarrow\{0,1\}$ with $\chi(s)=1$ if $O(s)=i$ and $\chi(s)=0$ otherwise. However, for cases in BHH, it suffices to only consider the probability when $p_x$ and $\tilde{p}_x$ are close, while in our case, we need to analyze the whole output distribution of quantum counting. Specifically, letting $t=\big| \chi^{-1}(1)\big|$ and $a=t/S=\sin^{2}(\omega\pi)$ for some $\omega$, we have

\begin{theorem}[\cite{brassard2002quantum}]\label{thm:quantum-counting}
For any $k, M \in \N$, there is a quantum algorithm (named \textbf{EstAmp}) with $M$ quantum queries to $\chi$ that outputs $\tilde{a}=\sin^{2}\big(\frac{l\pi}{M}\big)$ for some $l\in\{0,\ldots,M-1\}$ such that
\begin{align}
\Pr\Big[\tilde{a}=\sin^{2}\Big(\frac{l\pi}{M}\Big)\Big]=\frac{\sin^{2}(M\Delta\pi)}{M^{2}\sin^{2}(\Delta\pi)}\leq\frac{1}{(2M\Delta)^{2}},
\end{align}
where $\Delta=|\omega-\frac{l}{M}|$. This promises $|\tilde{a}-a|\leq 2\pi k\frac{\sqrt{a(1-a)}}{M}+k^{2}\frac{\pi^{2}}{M^{2}}$ with probability at least $\frac{8}{\pi^{2}}$ for $k=1$ and with probability greater than $1-\frac{1}{2(k-1)}$ for $k\geq 2$. If $a=0$ then $\tilde{a}=0$ with certainty.
\end{theorem}

Moreover, we also need to slightly modify \textbf{EstAmp} to avoid outputting $\tilde{p}_x=0$ in estimating Shannon entropy.
This is because $f(\tilde{p}_x)=\log(\tilde{p}_x)$ is not well-defined at $\tilde{p}_x=0$.
Let \textbf{EstAmp$'$} be the modified algorithm. It is required that \textbf{EstAmp$'$} outputs $\sin^{2}(\frac{\pi}{2M})$ when \textbf{EstAmp} outputs 0 and outputs \textbf{EstAmp}'s output otherwise.

By leveraging \thm{Monte-Carlo}, \thm{Monte-Carlo-multiplicative}, \thm{quantum-counting}, and carefully setting parameters in \algo{Master}, we have the following corollaries that describe the complexity of estimating any $F(p)$.

\begin{corollary}[additive error]\label{cor:Master-additive}
Given $\epsilon>0$. If $l=\Theta\big(\big(\frac{\sigma}{\epsilon}\big)\log^{3/2}\big(\frac{\sigma}{\epsilon}\big)\log\log\big(\frac{\sigma}{\epsilon}\big)\big)$ where $\var[X]\leq\sigma^{2}$ and $M$ is large enough such that $\big|\E[X]-F(p)\big|\leq\epsilon$, then \algo{Master} approximates $F(p)$ with an additive error $\epsilon$ and success probability $2/3$ using $O\big(M\cdot l)$ quantum queries to $p$.
\end{corollary}

\begin{corollary}[multiplicative error] \label{cor:Master-multiplicative}
Assume a procedure using $C_{a,b}$ quantum queries that returns an estimated range $[a,b]$, and that $\E[X]\in [a,b]$ with probability at least 0.9. Let $l=\Theta\big((\frac{\sigma b}{\epsilon a})\log^{3/2}(\frac{\sigma b}{\epsilon a})\log\log(\frac{\sigma b}{\epsilon a})\big)$ where $\var[X]/ (\E[X])^{2}\leq \sigma^{2}$ and $\epsilon>0$. For large enough $M$ such that $\big|\E[X]-F(p)\big|\leq\epsilon$, \algo{Master} estimates $F(p)$ with a multiplicative error $\epsilon$ and success probability $2/3$ with $O\big(M\cdot l+C_{a,b})$ queries.
\end{corollary}


\section{Shannon entropy estimation}\label{sec:Shannon}
We develop \algo{EstEntropy} for Shannon entropy estimation with \textbf{EstAmp$'$} in Line \ref{line:amplitude-estimation}, which provides quadratic quantum speedup in $n$.

\begin{algorithm}[htbp]
\SetKwFor{ForSubroutine}{Regard the following subroutine as $\mathcal{A}$:}{}{endfor}
Set $l=\Theta\big(\frac{\log(n/\epsilon^{2})}{\epsilon}\log^{3/2}\big(\frac{\log(n/\epsilon^{2})}{\epsilon}\big)\log\log\big(\frac{\log(n/\epsilon^{2})}{\epsilon}\big)\big)$\;
\ForSubroutine{}
    {Draw a sample $i\in\range{n}$ according to $p$\;
    Use \textbf{EstAmp$'$} with $M=2^{\lceil\log_{2}(\sqrt{n}/\epsilon)\rceil}$ queries to obtain an estimation $\tilde{p}_{i}$ of $p_{i}$\;
    Output $\tilde{x}_{i}=\log(1/\tilde{p}_{i})$\;
}
Use $\mathcal{A}$ for $l$ executions in \thm{Monte-Carlo} and output an estimation $\tilde{H}(p)$ of $H(p)$\;
\caption{Estimate the Shannon entropy of $p=(p_{i})_{i=1}^{n}$ on $\range{n}$.}
\label{algo:EstEntropy}
\end{algorithm}

\begin{theorem}\label{thm:Shannon}
\algo{EstEntropy} approximates $H(p)$ within an additive error $0<\epsilon\leq O(1)$ with success probability at least $\frac{2}{3}$ using $\tilde{O}\big(\frac{\sqrt{n}}{\epsilon^{2}}\big)$ quantum queries to $p$.
\end{theorem}

\begin{proof} We prove this theorem in two steps. The \textbf{first} step is to show that the expectation of the subroutine $\mathcal{A}$'s output (denoted $\tilde{E}:=\sum_{i\in\range{n}}p_{i}\cdot\log(1/\tilde{p}_{i})$) is close to $E:=\sum_{i\in\range{n}}p_{i}\cdot\log(1/p_{i})=H(p)$. To that end, we divide $[n]$ into partitions based on the corresponding probabilities.
Let $m=\lceil\log_{2}(\sqrt{n}/\epsilon)\rceil$ and $S_{0}=\{i: p_{i}\leq\sin^{2}(\pi/2^{m+1})\}$, $S_{1}=\{i: \sin^{2}(\pi/2^{m+1})<p_{i}\leq\sin^{2}(\pi/2^{m})\}$, $S_{2}=\{i: \sin^{2}(\pi/2^{m})<p_{i}\leq\sin^{2}(\pi/2^{m-1})\}, \ldots, S_{m}=\{i: \sin^{2}(\pi/4)<p_{i}\leq\sin^{2}(\pi/2)\}$. For convenience, denote $s_{0}=|S_{0}|,s_{1}=|S_{1}|,\ldots, s_{m}=|S_{m}|$. Then
\begin{align}\label{eqn:Shannon-cardinality}
\sum_{j=0}^{m}s_{j}=n,\qquad \sum_{j=0}^{m}\frac{2^{2j}}{2^{2m}}s_{j}=\Theta(1).
\end{align}

Our main technical contribution is the following upper bound on the expected difference between $\log\tilde{p}_{i}$ and $\log p_i$ in terms of the partition $S_i$, $i=1, \cdots, n$:
\begin{align}\label{eqn:Shannon-part-1-general}
\sum_{i\in S_{j}}p_{i}\E\big[\big|\log\tilde{p}_{i}-\log p_{i}\big|\big]=O\Big(\frac{2^{j}s_{j}}{2^{2m}}\Big)\quad\forall\,j\in\{1,\ldots,m\}.
\end{align}
By linearity of expectation, we have
\begin{align}\label{eqn:Shannon-binary-division}
|\tilde{E}-E|\leq\sum_{i\in\range{n}}p_{i}\E\big[\big|\log\tilde{p}_{i}-\log p_{i}\big|\big]=\sum_{j=0}^{m}\sum_{i\in S_{j}}p_{i}\E\big[\big|\log\tilde{p}_{i}-\log p_{i}\big|\big]=\sum_{j=0}^{m}O\Big(\frac{2^{j}s_{j}}{2^{2m}}\Big).
\end{align}
As a result, by applying \eqn{Shannon-cardinality} and Cauchy-Schwartz inequality to \eqn{Shannon-binary-division}, we have
\begin{align}
|\tilde{E}-E|=\sum_{j=0}^{m}O\Big(\frac{2^{j}s_{j}}{2^{2m}}\Big)\leq O\Bigg(\sqrt{\Big(\sum_{j=0}^{m}\frac{1}{2^{2m}}s_{j}\Big)\Big(\sum_{j=0}^{m}\frac{2^{2j}}{2^{2m}}s_{j}\Big)}\Bigg)=O(\epsilon). \label{eqn:Shannon-sum-of-expectation}
\end{align}
Because a constant overhead does not influence the query complexity, we may rescale \algo{EstEntropy} by a large enough constant so that $|\tilde{E}-E|\leq \epsilon/2$.

The \textbf{second} step is to bound the variance of the random variable, which is
\begin{align}
\sum_{i\in\range{n}}p_{i}(\log \tilde{p}_{i})^{2}-\Big(\sum_{i\in\range{n}}p_{i}\log \tilde{p}_{i}\Big)^{2}\leq \sum_{i\in\range{n}}p_{i}(\log \tilde{p}_{i})^{2}.
\end{align}
Since for any $i$, \textbf{EstAmp}$'$ outputs $\tilde{p}_{i}$ such that $\tilde{p}_{i}\geq\sin^{2}(\frac{\pi}{2M})\geq\frac{1}{M^{2}}\geq\frac{\epsilon^{2}}{4n}$, we have $\sum_{i\in\range{n}}p_{i}(\log \tilde{p}_{i})^{2}\leq \sum_{i\in\range{n}}p_{i}\cdot \big(\log\frac{4n}{\epsilon^{2}}\big)^{2}=\big(\log\frac{4n}{\epsilon^{2}}\big)^{2}$. As a result, by \cor{Master-additive} we can approximate $\tilde{E}$ up to additive error $\epsilon/2$ with failure probability at most $1/3$ using
\begin{align}
O\Big(\frac{\log(n/\epsilon^{2})}{\epsilon}\log^{3/2}\Big(\frac{\log(n/\epsilon^{2})}{\epsilon}\Big)\log\log\Big(\frac{\log(n/\epsilon^{2})}{\epsilon}\Big)\Big)\cdot 2^{\lceil\log_{2}(\sqrt{n}/\epsilon)\rceil}=\tilde{O}\Big(\frac{\sqrt{n}}{\epsilon^{2}}\Big)
\end{align}
quantum queries. Together with $|\tilde{E}-E|\leq \epsilon/2$, \algo{EstEntropy} approximates $E=H(p)$ up to additive error $\epsilon$ with failure probability at most $1/3$.
\end{proof}

It remains to prove \eqn{Shannon-part-1-general}. We prove:
\begin{align}\label{eqn:Shannon-part-1-rewritten}
\sum_{i\in S_{0}}p_{i}\E\big[\big|\log\tilde{p}_{i}-\log p_{i}\big|\big]=O\Big(\frac{s_{0}}{2^{2m}}\Big).
\end{align}
For $j\in\{1,2,\ldots,m\}$ in \eqn{Shannon-part-1-general}, the proof is similar because the dominating term has the angles of $\tilde{p}_{i}$ and $p_{i}$ fall into the same interval of length $\frac{1}{2^{m}}$, and as a result $|\log\tilde{p}_{i}-\log p_{i}\big|=O(\frac{1}{2^{j}})$.

\begin{proof}[Proof of \eqn{Shannon-part-1-rewritten}]
For convenience, denote $h(x):=x(\log t-\log x)\leq t/e$ where $0<t\leq 1$ and $x\in(0,t]$. Because $h'(x)=\log t-\log x-1$, when $x\in(0,t/e)$, $h'(x)>0$ hence $h(x)$ is an increasing function; when $x\in(t/e,t)$, $h'(x)<0$ hence $h(x)$ is a decreasing function; when $x=t/e$, $h'(x)=0$ and $h$ reaches its maximum $t/e$.

Since $i\in S_{0}$, we can write $p_{i}=\sin^{2}(\theta_{i}\pi)$ where $0<\theta_{i}\leq 1/2^{m+1}$. By \thm{quantum-counting}, for any $l\in\{1,\ldots,2^{m-1}\}$, the output of \textbf{EstAmp$'$} when taking $2^{m}$ queries satisfies
\begin{align}
\Pr\Big[\tilde{p}_{i}=\sin^{2}\Big(\frac{\pi}{2^{m+1}}\Big)\Big]&=\frac{\sin^{2}(2^{m}\theta_{i}\pi)}{2^{2m}\sin^{2}(\theta_{i}\pi)}\leq 1; \label{eqn:Shannon-trig-beginning} \\
\Pr\Big[\tilde{p}_{i}=\sin^{2}\Big(\frac{l\pi}{2^{m}}\Big)\Big]&=\frac{\sin^{2}(2^{m}(\frac{l}{2^{m}}-\theta_{i})\pi)}{2^{2m}\sin^{2}((\frac{l}{2^{m}}-\theta_{i})\pi)}\leq\frac{1}{(2^{m+1}(\frac{l}{2^{m}}-\theta_{i}))^{2}}. \label{eqn:Shannon-trig-tail}
\end{align}
Combining \eqn{Shannon-trig-beginning}, \eqn{Shannon-trig-tail}, and the property of function $h$ discussed above, for any $i\in S_{0}$ we have
\begin{align}
&p_{i}\E\big[\big|\log\tilde{p}_{i}-\log p_{i}\big|\big] \nonumber \\
&\quad\leq 1\cdot p_{i}\Big(\log\sin^{2}\Big(\frac{\pi}{2^{m+1}}\Big)-\log p_{i}\Big)+\sum_{l=1}^{2^{m-1}}\frac{p_{i}\big(\log\sin^{2}(\frac{l\pi}{2^{m}})-\log p_{i}\big)}{(2^{m+1}(\frac{l}{2^{m}}-\theta_{i}))^{2}} \label{eqn:Shannon-0-analyze-1} \\
&\quad\leq \frac{\sin^{2}(\frac{\pi}{2^{m+1}})}{e}+\sum_{l=1}^{2^{m-1}}\frac{1}{(2l-1)^{2}}\cdot \sin^{2}\Big(\frac{\pi}{2^{m+1}}\Big)\Big(\log\sin^{2}\Big(\frac{l\pi}{2^{m}}\Big)-\log \sin^{2}\Big(\frac{\pi}{2^{m+1}}\Big)\Big) \label{eqn:Shannon-0-analyze-2} \\
&\quad\leq \frac{\pi^{2}}{4e}\frac{1}{2^{2m}}+\frac{1}{2^{2m}}\cdot \frac{\pi^{2}}{4}\sum_{l=1}^{2^{m-1}}\frac{1}{(2l-1)^{2}}\log\Big(\frac{\sin(\frac{l\pi}{2^{m}})}{\sin\big(\frac{\pi}{2^{m+1}}\big)}\Big)^{2} \label{eqn:Shannon-0-analyze-3} \\
&\quad\leq \frac{\pi^{2}}{4e}\frac{1}{2^{2m}}+\frac{1}{2^{2m}}\cdot \frac{\pi^{2}}{2}\sum_{l=1}^{2^{m-1}}\frac{\log 2l}{(2l-1)^{2}} \label{eqn:Shannon-0-analyze-4} \\
&\quad= O\Big(\frac{1}{2^{2m}}\Big), \label{eqn:Shannon-0-analyze-5}
\end{align}
where \eqn{Shannon-0-analyze-1} comes from \eqn{Shannon-trig-beginning} and \eqn{Shannon-trig-tail}, \eqn{Shannon-0-analyze-2} comes from the property of $h$, \eqn{Shannon-0-analyze-3} holds because $\sin^{2}(\frac{\pi}{2^{m+1}})\leq\frac{\pi^{2}}{2^{2m+2}}$, \eqn{Shannon-0-analyze-4} holds because $\sin^{2}(\frac{l\pi}{2^{m}})\leq 4l^{2}\sin^{2}\big(\frac{\pi}{2^{m+1}}\big)$, and \eqn{Shannon-0-analyze-5} holds because $\sum_{l=1}^{\infty}\frac{\log l}{(2l-1)^{2}}=O(1)$. Consequently,
\begin{align}
\sum_{i\in S_{0}}p_{i}\E\big[\big|\log\tilde{p}_{i}-\log p_{i}\big|\big]=O\Big(\frac{1}{2^{2m}}\Big)\cdot s_{0}=O\Big(\frac{s_{0}}{2^{2m}}\Big).
\end{align}
\end{proof}


\section{Application: KL divergence estimation}\label{sec:KL}
Classically, there does not exist any consistent estimator that guarantees asymptotically small error over the set of all pairs of distributions \cite{han2016minimax,bu2016estimation}. These two papers then consider pairs of distributions with bounded probability ratios specified by a function $f:\N\rightarrow \R^{+}$, namely all pairs of distributions in the set as follows:
\begin{align}\label{eqn:KL-Uset}
\mathcal{U}_{n,f(n)}:=\big\{(p,q): |p|=|q|=n, \frac{p_{i}}{q_{i}}\leq f(n)\ \forall\,i\in\range{n}\big\}.
\end{align}
Denote the number of samples from $p$ and $q$ to be $M_{p}$ and $M_{q}$, respectively. References \cite{han2016minimax,bu2016estimation} shows that classically, $D_{\textrm{KL}}(p\|q)$ can be approximated within constant additive error with high success probability if and only if $M_{p}=\Omega(\frac{n}{\log n})$ and $M_{q}=\Omega(\frac{nf(n)}{\log n})$.

Quantumly, we are given unitary oracles $\hat{O}_{p}$ and $\hat{O}_{q}$ defined by \eqn{BHH-query-model}. \algo{KL-divergence} below estimates the KL-divergence between $p$ and $q$, which is similar to \algo{EstEntropy} that uses \textbf{EstAmp$'$}, while adapts $f$ to be mutually defined by $p$ and $q$.

\begin{algorithm}[htbp]
\SetKwFor{ForSubroutine}{Regard the following subroutine as $\mathcal{A}$:}{}{endfor}
Set $l=\Theta\big(\frac{\log^{2}(nf(n)/\epsilon^{2})}{\epsilon}\log^{3/2}\big(\frac{\log^{2}(nf(n)/\epsilon^{2})}{\epsilon}\big)\log\log\big(\frac{\log^{2}(nf(n)/\epsilon^{2})}{\epsilon}\big)\big)$\;
\ForSubroutine{}
    {Draw a sample $i\in\range{n}$ according to $p$\;
    Use the modified amplitude estimation procedure \textbf{EstAmp$'$} with $2^{\lceil\log_{2}(\sqrt{n}/\epsilon)\rceil}$ and $2^{\lceil\log_{2}(\sqrt{n}f(n)/\epsilon)\rceil}$ quantum queries to $p$ and $q$ to obtain estimates $\tilde{p}_{i}$ and $\tilde{q}_{i}$, respectively\;
    Output $\tilde{x}_{i}=\log\tilde{p}_{i}-\log\tilde{q}_{i}$\;
}
Use $\mathcal{A}$ for $l$ times in \thm{Monte-Carlo} and outputs an estimation $\tilde{D}_{\textrm{KL}}(p\|q)$ of $D_{\textrm{KL}}(p\|q)$\;
\caption{Estimate the KL divergence of $p=(p_{i})_{i=1}^{n}$ and $q=(q_{i})_{i=1}^{n}$ on $\range{n}$.}
\label{algo:KL-divergence}
\end{algorithm}

\begin{theorem}\label{thm:KL-divergence}
For $(p,q)\in \mathcal{U}_{n,f(n)}$, \algo{KL-divergence} approximates $D_{\textrm{KL}}(p\|q)$ within an additive error $\epsilon>0$ with success probability at least $\frac{2}{3}$ using $\tilde{O}\big(\frac{\sqrt{n}}{\epsilon^{2}}\big)$ quantum queries to $p$ and $\tilde{O}\big(\frac{\sqrt{n}f(n)}{\epsilon^{2}}\big)$ quantum queries to $q$,  where $\tilde{O}$ hides polynomials terms of $\log n$, $\log 1/\epsilon$, and $\log f(n)$.
\end{theorem}

\begin{proof}
If the estimates $\tilde{p}_{i}$ and $\tilde{q}_{i}$ were precisely accurate, the expectation of the subroutine's output would be $E:=\sum_{i\in\range{n}}p_{i}\cdot(\log p_{i}-\log q_{i})=D_{\textrm{KL}}(p\|q)$. On the one hand, we bound how far the actual expectation of the subroutine's output $\tilde{E}$ is from its exact value $E$. By linearity of expectation,
\begin{align}
|\tilde{E}-E|&\leq \sum_{i\in\range{n}}p_{i}\E\big[\big|(\log\tilde{p}_{i}-\log p_{i})+(\log\tilde{q}_{i}-\log q_{i})\big|\big] \\
&\leq \sum_{i\in\range{n}}p_{i}\E\big[\big|\log\tilde{p}_{i}-\log p_{i}\big|\big]+\sum_{i\in\range{n}}p_{i}\E\big[\big|\log\tilde{q}_{i}-\log q_{i}\big|\big] \\
&\leq \sum_{i\in\range{n}}p_{i}\E\big[\big|\log\tilde{p}_{i}-\log p_{i}\big|\big]+f(n)\sum_{i\in\range{n}}q_{i}\E\big[\big|\log\tilde{q}_{i}-\log q_{i}\big|\big], \label{eqn:KL-expectation-sum}
\end{align}
where \eqn{KL-expectation-sum} comes from the definition of $\mathcal{U}_{n,f(n)}$ in \eqn{KL-Uset}. By the proof of \thm{Shannon}, in particular equation \eqn{Shannon-sum-of-expectation}, $2^{\lceil\log_{2}(\sqrt{n}/\epsilon)\rceil}$ and $2^{\lceil\log_{2}(\sqrt{n}f(n)/\epsilon)\rceil}$ quantum queries to $p$ and $q$ give
\begin{align}
\sum_{i\in\range{n}}p_{i}\E\big[\big|\log\tilde{p}_{i}-\log p_{i}\big|\big]=O(\epsilon)\quad\text{and}\quad \sum_{i\in\range{n}}q_{i}\E\big[\big|\log\tilde{q}_{i}-\log q_{i}\big|\big]=O\Big(\frac{\epsilon}{f(n)}\Big),
\end{align}
respectively. Plugging them into \eqn{KL-expectation-sum} and rescaling \algo{KL-divergence} by a large enough constant, we get $|\tilde{E}-E|\leq\frac{\epsilon}{2}$.

On the other hand, the variance of the random variable is at most
\begin{align}
\sum_{i\in\range{n}}p_{i}(\log \tilde{p}_{i}-\log \tilde{q}_{i})^{2}=\sum_{i:\,\tilde{p}_{i}<\tilde{q}_{i}}p_{i}(\log \tilde{q}_{i}-\log \tilde{p}_{i})^{2}+\sum_{i:\,\tilde{p}_{i}\geq \tilde{q}_{i}}p_{i}(\log \tilde{p}_{i}-\log \tilde{q}_{i})^{2}. \label{eqn:KL-Var}
\end{align}
For the first term in \eqn{KL-Var}, because \textbf{EstAmp}$'$ outputs $\tilde{p}_{i}$ such that $\tilde{p}_{i}\geq\sin^{2}(\frac{\pi}{2^{\lceil\log_{2}(\sqrt{n}/\epsilon)\rceil+1}})\geq\frac{\epsilon^{2}}{4n}$  for any $i$, we have
\begin{align}\label{eqn:KL-Var-1}
\sum_{i:\,\tilde{p}_{i}<\tilde{q}_{i}}p_{i}(\log \tilde{q}_{i}-\log \tilde{p}_{i})^{2}\leq \sum_{i:\,\tilde{p}_{i}<\tilde{q}_{i}}p_{i}\Big(\log 1-\log\frac{\epsilon^{2}}{4n}\Big)^{2}\leq\Big(\log\frac{4n}{\epsilon^{2}}\Big)^{2}.
\end{align}
For the second term in \eqn{KL-Var}, we have
\begin{align}\label{eqn:KL-Var-2}
\sum_{i:\,\tilde{p}_{i}\geq\tilde{q}_{i}}p_{i}(\log \tilde{p}_{i}-\log \tilde{q}_{i})^{2}\leq\sum_{i:\,\tilde{p}_{i}\geq\tilde{q}_{i}}p_{i}(\log f(n))^{2}\leq (\log f(n))^{2}.
\end{align}
Plugging \eqn{KL-Var-1} and \eqn{KL-Var-2} into \eqn{KL-Var},  the variance of the random variable is at most
\begin{align}
\Big(\log\frac{4n}{\epsilon^{2}}\Big)^{2}+(\log f(n))^{2}=O\Big(\Big(\log\frac{nf(n)}{\epsilon^{2}}\Big)^{2}\Big).
\end{align}
As a result, by \cor{Master-additive} we can approximate $\tilde{E}$ up to additive error $\epsilon/2$ with success probability at least $2/3$ using $\tilde{O}(\frac{1}{\epsilon})\cdot 2^{\lceil\log_{2}(\sqrt{n}/\epsilon)\rceil}=\tilde{O}\big(\frac{\sqrt{n}}{\epsilon^{2}}\big)$ quantum queries to $p$ and $\tilde{O}(\frac{1}{\epsilon})\cdot 2^{\lceil\log_{2}(\sqrt{n}f(n)/\epsilon)\rceil}=\tilde{O}\big(\frac{\sqrt{n}f(n)}{\epsilon^{2}}\big)$ quantum queries to $q$, respectively. Together with $|\tilde{E}-E|\leq \epsilon/2$, \algo{KL-divergence} approximates $E=D_{\textrm{KL}}(p\|q)$ up to additive error $\epsilon$ with success probability at least $2/3$.
\end{proof}


\section{Non-integer R{\'e}nyi entropy estimation}\label{sec:noninteger-Renyi}
Recall the classical query complexity of non-integer and integer R{\'e}nyi entropy estimations are different \cite{acharya2017estimating}. Quantumly, we also consider them separately; in this section, we consider $\alpha$-R{\'e}nyi entropy estimation for general non-integer $\alpha>0$.

Let $P_{\alpha}(p):=\sum_{i=1}^{n}p_{i}^{\alpha}$. Since $H_{\alpha}(p)=\frac{1}{1-\alpha}\log P_{\alpha}(p)$, to approximate $H_{\alpha}(p)$ within an additive error $\epsilon>0$ it suffices to approximate $P_{\alpha}(p)$ within a multiplicative error $e^{(\alpha-1)\epsilon}-1=\Theta(\epsilon)$.

\subsection{Case 1: $\alpha>1,\alpha\notin\N$}\label{sec:Renyi-large}
We develop \algo{EstRenyiLarge} to approximate $P_{\alpha}(p)$ with a multiplicative error $\epsilon$.

\begin{algorithm}[htbp]
\SetKwFor{ForSubroutine}{Regard the following subroutine as $\mathcal{A}$:}{}{endfor}
\nonl \ForSubroutine{}
    {\nonl Draw a sample $i\in\range{n}$ according to $p$\;
    \nonl Use the amplitude estimation procedure \textbf{EstAmp} with $M=2^{\lceil\log_{2}(\frac{\sqrt{n}}{\epsilon}\log(\frac{\sqrt{n}}{\epsilon}))\rceil+1}$ queries to obtain an estimate $\tilde{p}_{i}$ of $p_{i}$\;
    \nonl Output $\tilde{x}_{i}=\tilde{p}_{i}^{\alpha-1}$\;
}
\hrule width \hsize height 1pt \vspace{5pt}

Input parameters $(\alpha,\epsilon,\delta)$, where $\epsilon$ is the multiplicative error and $\delta$ is the failure probability\; \label{line:EstRenyiLarge-input}
\eIf{$\alpha<1+\frac{1}{\log n}$\label{line:EstRenyiLarge-if}}
	{Take $a=\frac{1}{e}$ and $b=1$ as lower and upper bounds on $P_{\alpha}(p)$, respectively\; \label{line:EstRenyiLarge-ab-1}}
	{Recursively call \algo{EstRenyiLarge} with $\alpha'=\alpha(1+\frac{1}{\log n})^{-1}$, $\epsilon=1/4$, and $\delta=\frac{1}{12\log n\log \alpha}$ therein to give an estimate $\tilde{P}_{\alpha'}(p)$ of $P_{\alpha'}(p)$. For simplicity, denote $P:=\tilde{P}_{\alpha'}(p)$. Take $a=\frac{(3P/4)^{1+\frac{1}{\log n}}}{e}$ and $b=\big(\frac{5P}{4}\big)^{1+\frac{1}{\log n}}$ as lower and upper bounds on $P_{\alpha}(p)$, respectively\; \label{line:EstRenyiLarge-ab-2}}

Set $l=\Theta\big(\frac{n^{\frac{1}{2}-\frac{1}{2\alpha}}}{\epsilon}\log^{3/2}(\frac{n^{\frac{1}{2}-\frac{1}{2\alpha}}}{\epsilon}) \log\log(\frac{n^{\frac{1}{2}-\frac{1}{2\alpha}}}{\epsilon})\big)$\; \label{Line:EstRenyiLarge-executions}
Use $\mathcal{A}$ for $l$ executions in \thm{Monte-Carlo-multiplicative} using $a$ and $b$ as auxiliary information and output an estimation of $P_{\alpha}(p)$\; \label{Line:EstRenyiLarge-output-2/3}
Run Line 1 to Line \ref{Line:EstRenyiLarge-output-2/3} for $\lceil 48\log\frac{1}{\delta}\rceil$ executions and take the median of all outputs in Line \ref{Line:EstRenyiLarge-output-2/3}, denoted as $\tilde{P}_{\alpha}(p)$. Output $\tilde{P}_{\alpha}(p)$\; \label{Line:EstRenyiLarge-output-delta}
\caption{Estimate the $\alpha$-power sum $P_{\alpha}(p)$ of $p=(p_{i})_{i=1}^{n}$ on $\range{n}$, $\alpha>1,\alpha\notin\N$.}
\label{algo:EstRenyiLarge}
\end{algorithm}

\begin{theorem}\label{thm:Renyi-large}
The output of \algo{EstRenyiLarge} approximates $P_{\alpha}(p)$ within a multiplicative error $0<\epsilon\leq 1/4$ with success probability at least $1-\delta$ for some $\delta>0$ using $\tilde{O}\big(\frac{n^{1-\frac{1}{2\alpha}}}{\epsilon^{2}}\big)$ quantum queries to $p$, where $\tilde{O}$ hides polynomials terms of $\log n$, $\log 1/\epsilon$, and $\log 1/\delta$.
\end{theorem}

\begin{proof}[Proof of \thm{Renyi-large}]
First, we design a subroutine $\mathcal{A}$ in \algo{EstRenyiLarge} to approximate $P_{\alpha}(p)$ following the same principle as in \algo{EstEntropy}. If the estimate $\tilde{p}_{i}$ in $\mathcal{A}$ were precisely accurate, its expectation would be $E:=\sum_{i\in\range{n}}p_{i}\cdot p_{i}^{\alpha-1}=P_{\alpha}(p)$. To be precise, we bound how far the actual expectation of the subroutine's output $\tilde{E}$ is from the exact value $P_{\alpha}(p)$. In \lem{Renyi-large-1}, we show that when taking $M=2^{\lceil\log_{2}(\frac{\sqrt{n}}{\epsilon}\log(\frac{\sqrt{n}}{\epsilon}))\rceil+1}$ queries in \textbf{EstAmp}, we have $|\tilde{E}-E|=O(\epsilon E)$.

As a result, to approximate $P_{\alpha}(p)$ within multiplicative error $\Theta(\epsilon)$, it is equivalent to approximate $\tilde{E}$ within multiplicative error $\Theta(\epsilon)$. Recall \thm{Monte-Carlo-multiplicative} showed that if the variance of the random variable output by $\mathcal{A}$ is at most $\sigma^{2}\tilde{E}^{2}$ for a known $\sigma$, and if we can obtain two values $a,b$ such that $\tilde{E}\in[a,b]$, then $\tilde{O}(\sigma b/\epsilon a)$ executions of $\mathcal{A}$ suffice to approximate $\tilde{E}$ within multiplicative error $\epsilon$ with success probability at least $9/10$. In the main body of the algorithm (Line \ref{line:EstRenyiLarge-input} to Line \ref{Line:EstRenyiLarge-output-delta}), we use \thm{Monte-Carlo-multiplicative} to approximate $\tilde{E}$.

On the one hand, in \lem{Renyi-large-2}, we show that for $\alpha>1$ and large enough $n$, the variance is at most $5n^{1-1/\alpha} \tilde{E}^{2}$ with probability at least $\frac{8}{\pi^{2}}$. This gives $\sigma=\sqrt{5n^{1-1/\alpha}}=O(n^{1/2-1/2\alpha})$.

On the other hand, we need to compute the lower bound $a$ and upper bound $b$. A key observation (\lem{Renyi-large-approx}) is that for any $0<\alpha_{1}<\alpha_{2}$, we have
\begin{align}\label{eqn:Renyi-large-approx-rewriting}
\Big(\sum_{i\in\range{n}}p_{i}^{\alpha_{2}}\Big)^{\frac{\alpha_{1}}{\alpha_{2}}}\leq \sum_{i\in\range{n}}p_{i}^{\alpha_{1}}\leq n^{1-\frac{\alpha_{1}}{\alpha_{2}}} \Big(\sum_{i\in\range{n}}p_{i}^{\alpha_{2}}\Big)^{\frac{\alpha_{1}}{\alpha_{2}}}.
\end{align}
Because $n^{1/\log n}=e$, if $\frac{\alpha_{2}}{\alpha_{1}}=1+O(\frac{1}{\log n})$, then
\begin{align}
\sum_{i\in\range{n}}p_{i}^{\alpha_{1}}=\Theta\Big(\Big(\sum_{i\in\range{n}}p_{i}^{\alpha_{2}}\Big)^{\frac{\alpha_{1}}{\alpha_{2}}}\Big).
 \end{align}
As a result, we compute $a$ and $b$ by recursively calling \algo{EstRenyiLarge} to estimate $P_{\alpha'}(p)$ for $\alpha'=\alpha/(1+1/\log n)$, which is used to compute the lower bound $a$ and upper bound $b$ in Line \ref{line:EstRenyiLarge-ab-2}; the recursive call keeps until $\alpha<1+\frac{1}{\log n}$, when $a=\frac{1}{e}$ and $b=1$ (as in Line \ref{line:EstRenyiLarge-ab-1}) are simply lower and upper bounds on $P_{\alpha}(p)$ by \eqn{Renyi-large-approx-rewriting}.

To be precise, in \lem{Renyi-large-3}, we prove that $b/a<4e=O(1)$, and with probability at least $1/e^{1/12}>0.92$, $a$ and $b$ are indeed lower and upper bounds on $P_{\alpha}(p)$, respectively; furthermore, in Line \ref{line:EstRenyiLarge-ab-2}, \algo{EstRenyiLarge} is recursively called by at most $\log n\log\alpha$ times, and each recursive call takes at most $\tilde{O}(n^{1-\frac{1}{2\alpha}})$ queries. This promises that when we apply \cor{Master-multiplicative}, the cost $C_{a,b}$ is dominated by the query cost from \algo{Monte-Carlo-multiplicative}.

Combining all points above, \cor{Master-multiplicative} approximates $\tilde{E}$ up to multiplicative error $\Theta(\epsilon)$ with success probability at least $\frac{8}{\pi^{2}}\cdot 0.92\cdot 9/10>2/3$ using
\begin{align}
\log n\log\alpha\cdot\tilde{O}\Big(\frac{4e\cdot\sqrt{5n^{1-1/\alpha}} }{\epsilon}\Big)\cdot 2^{\lceil\log_{2}(\frac{\sqrt{n}}{\epsilon}\log(\frac{\sqrt{n}}{\epsilon}))\rceil+1}=\tilde{O}\Big(\frac{n^{1-1/2\alpha}}{\epsilon^{2}}\Big)
\end{align}
quantum queries. Together with $|\tilde{E}-E|=O(\epsilon E)$ and rescale $l,M$ by a large enough constant, Line 1 to Line \ref{Line:EstRenyiLarge-output-2/3} in \algo{EstRenyiLarge} approximates $E=P_{\alpha}(p)$ up to multiplicative error $\epsilon$ with success probability at least $2/3$.

Finally, in \lem{Renyi-large-4}, we show that after repeating the procedure for $\lceil 48\log\frac{1}{\delta}\rceil$ executions and taking the median $\tilde{P}_{\alpha}(p)$ (as in Line \ref{Line:EstRenyiLarge-output-delta}), the success probability that $\tilde{P}_{\alpha}(p)$ approximates $P_{\alpha}(p)$ within multiplicative error $\epsilon$ is boosted to $1-\delta$.
\end{proof}

It remains to prove the lemmas mentioned above.

\subsubsection{Expectation of $\mathcal{A}$ is $\epsilon$-close to $P_{\alpha}(p)$}
\begin{lemma}\label{lem:Renyi-large-1}
$|\tilde{E}-E|=O(\epsilon E)$.
\end{lemma}

\begin{proof}[Proof of \lem{Renyi-large-1}]
For convenience, denote $m=\lceil\log_{2}(\sqrt{n}/\epsilon\log(\sqrt{n}/\epsilon))\rceil+1$, and $S_{0}, S_{1},\ldots, S_{m}$ the same as in \sec{Shannon}. We still have \eqn{Shannon-cardinality}. By linearity of expectation,
\begin{align}\label{eqn:Renyi-large-binary-division}
|\tilde{E}-E|\leq \sum_{i\in\range{n}}p_{i}\E\big[\big|\tilde{p}_{i}^{\alpha-1}-p_{i}^{\alpha-1}\big|\big]=\sum_{j=0}^{m}\sum_{i\in S_{j}}p_{i}\E\big[\big|\tilde{p}_{i}^{\alpha-1}-p_{i}^{\alpha-1}\big|\big].
\end{align}
Therefore, to prove $|\tilde{E}-E|=O(\epsilon E)$ it suffices to show
\begin{align}\label{eqn:Renyi-large-binary-division-bound}
\sum_{j=0}^{m}\sum_{i\in S_{j}}p_{i}\E\big[\big|\tilde{p}_{i}^{\alpha-1}-p_{i}^{\alpha-1}\big|\big]=O\Big(\epsilon\sum_{i\in\range{n}}p_{i}^{\alpha}\Big).
\end{align}

For each $i\in\range{n}$ we write $p_{i}=\sin^{2}(\theta_{i}\pi)$. Assume $k\in\Z$ such that $k\leq 2^{m}\theta_{i}<k+1$. By \thm{quantum-counting}, for any $l\in\{1,2,\ldots,\max\{k-1,2^{m}-k-1\}\}$ the output of \textbf{EstAmp} taking $2^{m}$ queries satisfies
\begin{align}
\Pr\Big[\tilde{p}_{i}=\sin^{2}\big(\frac{(k+l+1)\pi}{2^{m}}\big)\Big],\Pr\Big[\tilde{p}_{i}=\sin^{2}\big(\frac{(k-l-1)\pi}{2^{m}}\big)\Big]\leq \frac{1}{4l^{2}}. \label{eqn:Renyi-large-trig-tail}
\end{align}
Furthermore, because $\sin\theta_{i}=\theta_{i}-O(\theta_{i}^{3})$, $\cos\theta_{i}=1-O(\theta_{i}^{2})$, and $(1+\theta_{i})^{2\alpha-1}=1+(2\alpha-1)\theta_{i}+o(\theta_{i})$,
\begin{align}
\big(\sin((\theta_{i}+\frac{l}{2^{m}})\pi)\big)^{2(\alpha-1)}-\big(\sin(\theta_{i}\pi)\big)^{2(\alpha-1)}=O\Big(\frac{l}{2^{m}}(\theta_{i}\pi)^{2\alpha-3}\Big). \label{eqn:Renyi-large-tail-taylor}
\end{align}
Combining \eqn{Renyi-large-trig-tail}, \eqn{Renyi-large-tail-taylor}, and the fact that $\sum_{l=1}^{2^{m}}\frac{1}{l}=\Theta(m)$, we have
\begin{align}
\sum_{j=0}^{m}\sum_{i\in S_{j}}p_{i}\E\big[\big|\tilde{p}_{i}^{\alpha-1}-p_{i}^{\alpha-1}\big|\big]&= O\Big(\sum_{j=0}^{m}s_{j}\cdot \big(\frac{2^{j}}{2^{m}}\pi\big)^{2}\cdot 2\sum_{l=1}^{2^{m}}\frac{1}{4l^{2}}\frac{l}{2^{m}}\big(\frac{2^{j}}{2^{m}}\pi\big)^{2\alpha-3}\Big) \\
&= O\Big(\frac{\pi^{2\alpha-1}m}{2^{2\alpha m}}\cdot \sum_{j=0}^{m}s_{j}2^{(2\alpha-1)j}\Big). \label{eqn:Renyi-large-LHS}
\end{align}
On the other side,
\begin{align}
\epsilon\sum_{i\in\range{n}}p_{i}^{\alpha}=\Theta\Big(\epsilon\sum_{j=0}^{m}s_{j}\cdot\big(\frac{2^{j}}{2^{m}}\pi\big)^{2\alpha}\Big)=\Theta\Big(\frac{\epsilon\cdot \pi^{2\alpha}}{2^{2\alpha m}}\sum_{j=0}^{m}s_{j}2^{2\alpha j}\Big). \label{eqn:Renyi-large-RHS}
\end{align}
Therefore, to prove equation \eqn{Renyi-large-binary-division-bound}, by \eqn{Renyi-large-LHS} and \eqn{Renyi-large-RHS} it suffices to prove
\begin{align}
\sum_{j=0}^{m}s_{j}2^{(2\alpha-1)j}=O\Big(\frac{\epsilon}{m}\sum_{j=0}^{m}s_{j}2^{2\alpha j}\Big).
\end{align}
Since $m=\lceil\log_{2}(\frac{\sqrt{n}}{\epsilon}\log(\frac{\sqrt{n}}{\epsilon}))\rceil+1$, we have $\frac{2^{m}}{m}\geq \frac{\sqrt{n}}{\epsilon}$, thus $\frac{\epsilon}{m}\geq\frac{\sqrt{n}}{2^{m}}$. Therefore, it suffices to show
\begin{align}\label{eqn:Renyi-large-Holder}
\sum_{j=0}^{m}s_{j}2^{(2\alpha-1)j}=O\Big(\frac{\sqrt{n}}{2^{m}}\sum_{j=0}^{m}s_{j}2^{2\alpha j}\Big).
\end{align}

If $\alpha\geq 3/2$, by H\"{o}lder's inequality we have
\begin{align}
\Big(\sum_{j=0}^{m}s_{j}\Big)^{\frac{1}{2\alpha}}\Big(\sum_{j=0}^{m}s_{j}2^{2\alpha j}\Big)^{\frac{2\alpha-1}{2\alpha}}\geq \sum_{j=0}^{m}s_{j}2^{(2\alpha-1)j}.
\end{align}
By equation \eqn{Shannon-cardinality}, this gives
\begin{align}\label{eqn:Renyi-large>3/2-1}
n^{\frac{1}{2\alpha-1}}\Big(\sum_{j=0}^{m}s_{j}2^{2\alpha j}\Big)\geq \Big(\sum_{j=0}^{m}s_{j}2^{(2\alpha-1)j}\Big)\Big(\sum_{j=0}^{m}s_{j}2^{(2\alpha-1)j}\Big)^{\frac{1}{2\alpha-1}}.
\end{align}
By H\"{o}lder's inequality and also equation \eqn{Shannon-cardinality}, we have
\begin{align}
\Big(\sum_{j=0}^{m}s_{j}\Big)^{\frac{2\alpha-3}{2\alpha-1}}\Big(\sum_{j=0}^{m}s_{j}2^{(2\alpha-1)j}\Big)^{\frac{2}{2\alpha-1}}\geq \sum_{j=0}^{m}s_{j}2^{2j}=\Theta(2^{2m}).
\end{align}
This is equivalent to
\begin{align}\label{eqn:Renyi-large>3/2-2}
n^{\frac{2\alpha-3}{2(2\alpha-1)}}\Big(\sum_{j=0}^{m}s_{j}2^{(2\alpha-1)j}\Big)^{\frac{1}{2\alpha-1}}\geq\Theta(2^{m}).
\end{align}
Combining \eqn{Renyi-large>3/2-1} and \eqn{Renyi-large>3/2-2}, we get exactly \eqn{Renyi-large-Holder}.

If $1<\alpha<3/2$, by H\"{o}lder's inequality we have
\begin{align}
\Big(\sum_{j=0}^{m}s_{j}2^{2\alpha j}\Big)^{\frac{1}{\alpha}}\Big(\sum_{j=0}^{m}s_{j}\Big)^{\frac{\alpha-1}{\alpha}}&\geq \sum_{j=0}^{m}s_{j}2^{2j}; \\
\Big(\sum_{j=0}^{m}s_{j}2^{2j}\Big)^{\frac{2\alpha-1}{2}}\Big(\sum_{j=0}^{m}s_{j}\Big)^{\frac{3-2\alpha}{2}}&\geq \sum_{j=0}^{m}s_{j}2^{(2\alpha-1)j}.
\end{align}
By equation \eqn{Shannon-cardinality}, the two inequalities above give
\begin{align}
\sum_{j=0}^{m}s_{j}2^{2\alpha j}&\geq n^{1-\alpha}2^{2\alpha m}\quad\text{and}\quad \sum_{j=0}^{m}s_{j}2^{(2\alpha-1)j}\leq n^{1.5-\alpha}2^{(2\alpha-1)m},
\end{align}
which give \eqn{Renyi-large-Holder}.
\end{proof}

\subsubsection{Bound the variance of $\mathcal{A}$ by the square of its expectation}
\begin{lemma}\label{lem:Renyi-large-2}
With probability at least $\frac{8}{\pi^{2}}$, the variance of the random variable output by $\mathcal{A}$ is at most $5n^{1-1/\alpha} \tilde{E}^{2}$.
\end{lemma}

\begin{proof}[Proof of \lem{Renyi-large-2}]
The expectation and variance of the output by $\mathcal{A}$ are $\tilde{E}=\sum_{i=1}^{n}p_{i}\cdot \tilde{p}_{i}^{\alpha-1}$ and $\sum_{i\in\range{n}}p_{i}\cdot (\tilde{p}_{i}^{\alpha-1})^{2}-\big(\sum_{i\in\range{n}}p_{i}\cdot \tilde{p}_{i}^{\alpha-1}\big)^{2}$, respectively. Therefore, it suffices to show that with probability at least $\frac{8}{\pi^{2}}$,
\begin{align}\label{eqn:Renyi-large-fact-2-1}
\sum_{i\in\range{n}}p_{i}\cdot (\tilde{p}_{i}^{\alpha-1})^{2}\leq 5n^{1-1/\alpha}\Big(\sum_{i=1}^{n}p_{i}\cdot \tilde{p}_{i}^{\alpha-1}\Big)^{2}.
\end{align}

By \thm{quantum-counting}, with probability at least $\frac{8}{\pi^{2}}$, we have
\begin{align}\label{eqn:Renyi-large-fact-2-2}
|\tilde{p}_{i}-p_{i}|\leq\frac{2\pi\sqrt{p_{i}}}{2^{m}}\leq\frac{\epsilon\pi\sqrt{p_{i}}}{\sqrt{n}}\qquad i\in\range{n}.
\end{align}
For convenience, denote $p:=p_{i^{*}}$ to be the maximal one among $p_{1},\ldots,p_{n}$, i.e., $p=\max_{i\in\{1,\ldots,n\}}p_{i}$. We also denote $\tilde{p}:=\tilde{p}_{i^{*}}$. Then we have
\begin{align}
\frac{\big(\sum_{i=1}^{n}p_{i}\cdot \tilde{p}_{i}^{\alpha-1}\big)^{2}}{\sum_{i\in\range{n}}p_{i}\cdot (\tilde{p}_{i}^{\alpha-1})^{2}}\geq \frac{\big(\sum_{i=1}^{n}p_{i}\cdot \tilde{p}_{i}^{\alpha-1}\big)^{2}}{\tilde{p}^{\alpha-1}\cdot \sum_{i\in\range{n}}p_{i}\cdot \tilde{p}_{i}^{\alpha-1}}=\frac{\sum_{i=1}^{n}p_{i}\cdot \tilde{p}_{i}^{\alpha-1}}{\tilde{p}^{\alpha-1}}.
\end{align}
Furthermore, because $x^{\alpha}$ is a convex function in $[0,1]$, by \eqn{Renyi-large-fact-2-2} and Jensen's inequality we have
\begin{align}
\frac{\sum_{i=1}^{n}p_{i}\cdot \tilde{p}_{i}^{\alpha-1}}{\tilde{p}^{\alpha-1}}&=\frac{p\cdot \tilde{p}^{\alpha-1}+\sum_{i\neq i^{*}}p_{i}\cdot \tilde{p}_{i}^{\alpha-1}}{\tilde{p}^{\alpha-1}} \\
&\geq \frac{p\Big(p+\frac{\epsilon\pi\sqrt{p}}{\sqrt{n}}\Big)^{\alpha-1}+(n-1)\cdot\frac{1-p}{n-1}\Big(\frac{1-p}{n-1}-\frac{\epsilon\pi\sqrt{(1-p)/(n-1)}}{\sqrt{n}}\Big)^{\alpha-1}}{\Big(p+\frac{\epsilon\pi\sqrt{p}}{\sqrt{n}}\Big)^{\alpha-1}} \\
&\approx p+(1-p)\Big(\frac{1-p-\epsilon\pi\sqrt{1-p}}{np+\epsilon\pi\sqrt{np}}\Big)^{\alpha-1}. \label{eqn:Renyi-large-fact-2-3}
\end{align}
Therefore, it suffices to show that for large enough $n$,
\begin{align}\label{eqn:relative-variance-large-equivalence}
p+(1-p)\Big(\frac{1-p-\epsilon\pi\sqrt{1-p}}{np+\epsilon\pi\sqrt{np}}\Big)^{\alpha-1}\geq 0.2n^{-(1-1/\alpha)}.
\end{align}
If $p\geq 0.2n^{-(1-1/\alpha)}$, equation \eqn{relative-variance-large-equivalence} directly follows. If $p<0.2n^{-(1-1/\alpha)}$,
\begin{align}
&\lim_{n\rightarrow\infty} n^{1-1/\alpha}\cdot (1-p)\Big(\frac{1-p-\epsilon\pi\sqrt{1-p}}{np+\epsilon\pi\sqrt{np}}\Big)^{\alpha-1} \nonumber \\
&\quad\geq \lim_{n\rightarrow\infty} n^{1-1/\alpha}(1-0.2n^{-(1-1/\alpha)})\Big(\frac{1-0.2n^{-(1-1/\alpha)}-\epsilon\pi}{0.2n^{1/\alpha}+\sqrt{0.2}\epsilon\pi n^{1/2\alpha}}\Big)^{\alpha-1} \\
&\quad= \lim_{n\rightarrow\infty} (1-0.2n^{-(1-1/\alpha)})\Big(\frac{n^{1/\alpha}(1-0.2n^{-(1-1/\alpha)}-\epsilon\pi)}{0.2n^{1/\alpha}+\sqrt{0.2}\epsilon\pi n^{1/2\alpha}}\Big)^{\alpha-1} \\
&\quad= 1\cdot \Big(\frac{1-\epsilon\pi}{0.2}\Big)^{\alpha-1}>1>0.2, \label{eqn:Renyi-large-fact-2-4}
\end{align}
where \eqn{Renyi-large-fact-2-4} is true because $\frac{1-\epsilon\pi}{0.2}>\frac{1-3.2/4}{0.2}=1$. Because \eqn{Renyi-large-fact-2-3} only omits lower order terms and the limit in \eqn{Renyi-large-fact-2-4} is a constant larger than 0.2, \lem{Renyi-large-2} follows.
\end{proof}

\subsubsection{Give tight bounds on $P_{\alpha}(p)$ by $P_{\alpha'}(p)$}
\begin{lemma}\label{lem:Renyi-large-approx}
For any distribution $(p_{i})_{i=1}^{n}$ and $0<\alpha_{1}<\alpha_{2}$, we have
\begin{align}\label{eqn:Renyi-large-approx}
\Big(\sum_{i\in\range{n}}p_{i}^{\alpha_{2}}\Big)^{\frac{\alpha_{1}}{\alpha_{2}}}\leq \sum_{i\in\range{n}}p_{i}^{\alpha_{1}}\leq n^{1-\frac{\alpha_{1}}{\alpha_{2}}} \Big(\sum_{i\in\range{n}}p_{i}^{\alpha_{2}}\Big)^{\frac{\alpha_{1}}{\alpha_{2}}}.
\end{align}
\end{lemma}

\begin{proof}[Proof of \lem{Renyi-large-approx}]
On the one hand, by the generalized mean inequality, we have
\begin{align}
\Bigg(\frac{\sum_{i\in\range{n}}p_{i}^{\alpha_{2}}}{n}\Bigg)^{\frac{1}{\alpha_{2}}}\geq \Bigg(\frac{\sum_{i\in\range{n}}p_{i}^{\alpha_{1}}}{n}\Bigg)^{\frac{1}{\alpha_{1}}},
\end{align}
which gives the second inequality in \eqn{Renyi-large-approx}.

On the other hand, since $\frac{\alpha_{1}}{\alpha_{2}}\leq 1$ and
\begin{align}
0\leq\frac{p_{i}^{\alpha_{2}}}{\sum_{j\in\range{n}}p_{j}^{\alpha_{2}}}\leq 1\qquad\forall\,i\in\range{n},
\end{align}
we have
\begin{align}
\frac{\sum_{i\in\range{n}}p_{i}^{\alpha_{1}}}{\Big(\sum_{j\in\range{n}}p_{j}^{\alpha_{2}}\Big)^{\frac{\alpha_{1}}{\alpha_{2}}}}= \sum_{i\in\range{n}}\Bigg(\frac{p_{i}^{\alpha_{2}}}{\sum_{j\in\range{n}}p_{j}^{\alpha_{2}}}\Bigg)^{\frac{\alpha_{1}}{\alpha_{2}}}\geq \sum_{i\in\range{n}}\frac{p_{i}^{\alpha_{2}}}{\sum_{j\in\range{n}}p_{j}^{\alpha_{2}}}=1,
\end{align}
which is equivalent to the first inequality in \eqn{Renyi-large-approx}.
\end{proof}

\subsubsection{Analyze the recursive calls}
\begin{lemma}\label{lem:Renyi-large-3}
With probability at least $0.92$, the $a$ and $b$ in Line \ref{line:EstRenyiLarge-ab-1} or Line \ref{line:EstRenyiLarge-ab-2} of \algo{EstRenyiLarge} are indeed lower and upper bounds on $P_{\alpha}(p)$, respectively, and $b/a=O(1)$; furthermore, in Line \ref{line:EstRenyiLarge-ab-2}, \algo{EstRenyiLarge} is recursively called for at most $\log n\log\alpha$ executions, and each recursive call takes at most $\tilde{O}(n^{1-\frac{1}{2\alpha}})$ queries.
\end{lemma}

\begin{proof}[Proof of \lem{Renyi-large-3}]
We decompose the proof into two parts:
\begin{itemize}
\item \hd{In Line \ref{line:EstRenyiLarge-ab-2}, \algo{EstRenyiLarge} is recursively called for at most $\log n\log\alpha$ executions, and each recursive call takes at most $\tilde{O}(n^{1-\frac{1}{2\alpha}})$ queries:}

Because each recursive call of \algo{EstRenyiLarge} reduces $\alpha$ by multiplying $(1+\frac{1}{\log n})^{-1}$ and the recursion ends when $\alpha<1+\frac{1}{\log n}$, the total number of recursive calls is at most $\frac{\log\alpha}{\log(1+\frac{1}{\log n})}\leq \log n\log\alpha$.
\\\\
When $\alpha<1+\frac{1}{\log n}$, $a$ and $b$ are set in Line \ref{line:EstRenyiLarge-ab-1} and no extra queries are needed; when Line \ref{line:EstRenyiLarge-ab-2} calls $\alpha(1+\frac{1}{\log n})^{-k}$-power sum estimation for some $k\in\N$, by induction on $k$, we see that this call takes at most $\tilde{O}\big(n^{1-\frac{(1+1/\log n)^{k}}{2\alpha}}\big)\leq \tilde{O}(n^{1-\frac{1}{2\alpha}})$ queries. As a result, when we apply \cor{Master-multiplicative}, the cost $C_{a,b}$ is dominated by the query cost from \algo{Monte-Carlo-multiplicative}.

\item \hd{With probability at least $0.92$, $a$ and $b$ are lower and upper bounds on $P_{\alpha}(p)$ respectively, and $b/a=O(1)$:}

When $1<\alpha<1+\frac{1}{\log n}$, on the one hand we have $\sum_{i=1}^{n}p_{i}^{\alpha}\leq\sum_{i=1}^{n}p_{i}=1$; on the other hand, because $n^{\frac{1}{\log n}}=e$, by \lem{Renyi-large-approx} we have
\begin{align}
\sum_{i=1}^{n}p_{i}^{\alpha}\geq\frac{\big(\sum_{i=1}^{n}p_{i}\big)^{\alpha}}{n^{\alpha-1}}\geq\frac{1}{e}.
\end{align}
Therefore, $a=1/e$ and $b=1$ in Line \ref{line:EstRenyiLarge-ab-1} are lower and upper bounds on $P_{\alpha}(p)$ respectively, and $b/a=e=O(1)$.
\\\\
When $\alpha>1+\frac{1}{\log n}$, for convenience denote $\alpha'=\alpha(1+\frac{1}{\log n})^{-1}$. As justified above, the total number of recursive calls in Line \ref{line:EstRenyiLarge-ab-2} is at most $\log n\log\alpha$. Because we take $\delta=\frac{1}{12\log n\log\alpha}$ in Line \ref{line:EstRenyiLarge-ab-2}, with probability at least
\begin{align}
\Big(1-\frac{1}{12\log n\log\alpha}\Big)^{\log n\log\alpha}\geq\frac{1}{e^{1/12}}>0.92,
\end{align}
the output of every recursive call is within $1/4$-multiplicative error. As a result, the $P$ in Line \ref{line:EstRenyiLarge-ab-2} satisfies $3P/4\leq \sum_{i=1}^{n}p_{i}^{\alpha'}\leq 5P/4$. Combining this with \lem{Renyi-large-approx} and using $n^{\frac{1}{\log n}}=e$, we have
\begin{align}
\frac{(3P/4)^{1+\frac{1}{\log n}}}{e}\leq \sum_{i=1}^{n}p_{i}^{\alpha}\leq \Big(\frac{5P}{4}\Big)^{1+\frac{1}{\log n}}.
\end{align}
In other words, $a$ and $b$ are indeed lower and upper bounds on $P_{\alpha}(p)$, respectively. Furthermore, $b/a=O(1)$ because
\begin{align}
\frac{b}{a}=e\cdot \Big(\frac{5}{3}\Big)^{1+\frac{1}{\log n}}<4e=O(1).
\end{align}
\end{itemize}
\end{proof}

\subsubsection{Boost the success probability}\label{sec:Renyi-key-inequality}
\begin{lemma}\label{lem:Renyi-large-4}
By repeating Line 1 to Line \ref{Line:EstRenyiLarge-output-2/3} in \algo{EstRenyiLarge} for $\lceil 48\log\frac{1}{\delta}\rceil$ executions and taking the median $\tilde{P}_{\alpha}(p)$, the success probability is boosted to $1-\delta$.
\end{lemma}

\begin{proof}[Proof of \lem{Renyi-large-4}]
Denote the outputs after running Line 1 to Line \ref{Line:EstRenyiLarge-output-2/3} for $\lceil 48\log\frac{1}{\delta}\rceil$ executions as $\tilde{P}_{\alpha}(p)^{(1)},\ldots,\tilde{P}_{\alpha}(p)^{(\lceil 48\log\frac{1}{\delta}\rceil)}$, respectively. Based on the correctness of \lem{Renyi-large-1}, \lem{Renyi-large-2}, and \lem{Renyi-large-3}, for each $i\in\{1,\ldots,\lceil 48\log\frac{1}{\delta}\rceil\}$, with probability at least $2/3$ we have
\begin{align}\label{eqn:Renyi-large-fact-4}
|\tilde{P}_{\alpha}(p)^{(i)}-P_{\alpha}(p)|\leq\epsilon P_{\alpha}(p).
\end{align}

For each $i\in\{1,\ldots,\lceil 48\log\frac{1}{\delta}\rceil\}$, denote $X_{i}$ to be a Boolean random variable such that $X_{i}=1$ if \eqn{Renyi-large-fact-4} holds, and $X_{i}=0$ otherwise. Then $\Pr[X_{i}=1]\geq 2/3$. Because in Line \ref{Line:EstRenyiLarge-output-delta} the output $\tilde{P}_{\alpha}(p)$ is the median of all $\tilde{P}_{\alpha}(p)^{(1)},\ldots,\tilde{P}_{\alpha}(p)^{(\lceil 48\log\frac{1}{\delta}\rceil)}$, $|\tilde{P}_{\alpha}(p)-P_{\alpha}(p)|>\epsilon P_{\alpha}(p)$ leads to $\sum_{i=1}^{\lceil 48\log\frac{1}{\delta}\rceil}X_{i}< \lceil 48\log\frac{1}{\delta}\rceil/2$. On the other hand, by Chernoff bound we have
\begin{align}
\Pr\Big[\sum_{i=1}^{\lceil 48\log\frac{1}{\delta}\rceil}X_{i}< \frac{\lceil 48\log\frac{1}{\delta}\rceil}{2}\Big]\leq \exp\Big(-\frac{2/3\lceil 48\log\frac{1}{\delta}\rceil\cdot(1/4)^{2}}{2}\Big)\leq \delta.
\end{align}
Therefore, with probability at least $1-\delta$, we have $|\tilde{P}_{\alpha}(p)-P_{\alpha}(p)|\leq \epsilon P_{\alpha}(p)$.
\end{proof}

\subsection{Case 2: $0<\alpha<1$}\label{sec:Renyi-small}

When $0<\alpha<1$, our quantum algorithm follows the same structure as  \algo{EstRenyiLarge}:
\begin{algorithm}[htbp]
\SetKwFor{ForSubroutine}{Regard the following subroutine as $\mathcal{A}$:}{}{endfor}
\nonl \ForSubroutine{}
    {\nonl Draw a sample $i\in\range{n}$ according to $p$\;
    \nonl Use the amplitude estimation procedure \textbf{EstAmp} with $M=2^{\lceil\log_{2}(\frac{n^{1/2\alpha}}{\epsilon}\log(\frac{n^{1/2\alpha}}{\epsilon}))\rceil+1}$ queries to obtain an estimate $\tilde{p}_{i}$ of $p_{i}$\;
    \nonl Output $\tilde{x}_{i}=\tilde{p}_{i}^{\alpha-1}$\;
}
\hrule width \hsize height 1pt \vspace{5pt}

Input parameters $(\alpha,\epsilon,\delta)$, where $\epsilon$ is the multiplicative error and $\delta$ is the failure probability\; \label{line:EstRenyiSmall-input}
\eIf{$\alpha>1-\frac{1}{\log n}$}
	{Take $a=1$ and $b=e$ as lower and upper bounds on $P_{\alpha}(p)$, respectively\; \label{line:EstRenyiSmall-ab-1}}
	{Recursively call \algo{EstRenyiSmall} with $\alpha'=\alpha(1-\frac{1}{\log n})^{-1}$, $\epsilon=1/2$, and $\delta=\frac{1}{12\log n\log 1/\alpha}$ therein to give an estimate $\tilde{P}_{\alpha'}(p)$ of $P_{\alpha'}(p)$. For simplicity, denote $P:=\tilde{P}_{\alpha'}(p)$. Take $a=(P/2)^{1-\frac{1}{\log n}}$ and $b=e(2P)^{1-\frac{1}{\log n}}$ as lower and upper bounds on $P_{\alpha}(p)$, respectively\; \label{line:EstRenyiSmall-ab-2}}

Set $l=\Theta\big(\frac{n^{\frac{1}{2\alpha}-\frac{1}{2}}}{\epsilon}\log^{3/2}(\frac{n^{\frac{1}{2\alpha}-\frac{1}{2}}}{\epsilon}) \log\log(\frac{n^{\frac{1}{2\alpha}-\frac{1}{2}}}{\epsilon})\big)$\;
Use $\mathcal{A}$ for $l$ executions in \thm{Monte-Carlo-multiplicative} using $a$ and $b$ as auxiliary information and output an estimation of $P_{\alpha}(p)$\; \label{Line:EstRenyiSmall-output-2/3}
Run Line 1 to Line \ref{Line:EstRenyiSmall-output-2/3} for $\lceil 48\log\frac{1}{\delta}\rceil$ executions and take the median of all outputs in Line \ref{Line:EstRenyiSmall-output-2/3}, denoted as $\tilde{P}_{\alpha}(p)$. Output $\tilde{P}_{\alpha}(p)$\; \label{Line:EstRenyiSmall-output-delta}
\caption{Estimate the $\alpha$-power sum $P_{\alpha}(p)$ of $p=(p_{i})_{i=1}^{n}$ on $\range{n}$, $0<\alpha<1$.}
\label{algo:EstRenyiSmall}
\end{algorithm}

\noindent
The main difference is that, in the case $\alpha>1$, \algo{EstRenyiLarge} makes $\alpha'$ smaller and smaller by multiplying $(1+\frac{1}{\log n})^{-1}$ each time, whereas in the case $0<\alpha<1$, \algo{EstRenyiSmall} makes $\alpha'$ larger and larger by multiplying $(1-\frac{1}{\log n})^{-1}$ each time; nevertheless, both recursions end when $\alpha'$ is close enough to 1. On the more technical level, they have different $M$ in $\mathcal{A}$, different upper bounds on the variance of $\mathcal{A}$, and different expressions for $a$ and $b$ in Line \ref{line:EstRenyiSmall-ab-1} and Line \ref{line:EstRenyiSmall-ab-2}.

\begin{theorem}\label{thm:Renyi-small}
The output of \algo{EstRenyiSmall} approximates $P_{\alpha}(p)$ within a multiplicative error $0<\epsilon\leq O(1)$ with success probability at least $1-\delta$ for some $\delta>0$ using $\tilde{O}\big(\frac{n^{\frac{1}{\alpha}-\frac{1}{2}}}{\epsilon^{2}}\big)$ quantum queries to $p$, where $\tilde{O}$ hides polynomials terms of $\log n$, $\log 1/\epsilon$, and $\log 1/\delta$.
\end{theorem}
\noindent
Before we give the formal proof of \thm{Renyi-small}, we compare the similarities and differences between \algo{EstRenyiLarge} and \algo{EstRenyiSmall}, listed below:
\begin{itemize}
\item In both algorithms, the subroutine $\mathcal{A}$ has the same structure, and is designed to estimate $P_{\alpha}(p)$. However, to make the expectation of $\mathcal{A}$ $\epsilon$-close to $P_{\alpha}(p)$,  the \textbf{EstAmp} in \algo{EstRenyiLarge} suffices to take $M=2^{\lceil\log_{2}(\frac{\sqrt{n}}{\epsilon}\log(\frac{\sqrt{n}}{\epsilon}))\rceil+1}$ queries (see \lem{Renyi-large-1}), whereas the \textbf{EstAmp} in \algo{EstRenyiSmall} needs to take $M=2^{\lceil\log_{2}(\frac{n^{1/2\alpha}}{\epsilon}\log(\frac{n^{1/2\alpha}}{\epsilon}))\rceil+1}$ queries (see \lem{Renyi-small-1});
\item In both algorithms, we use \thm{Monte-Carlo-multiplicative} to approximate the expectation of $\mathcal{A}$ (denoted $\tilde{E}$), hence they both need to upper-bound the variance of $\mathcal{A}$ by a multiple of $\tilde{E}^{2}$. However, technically the proofs are different, and we obtain different upper bounds in \lem{Renyi-large-2} and \lem{Renyi-small-2}, respectively;
\item Since both algorithms use \thm{Monte-Carlo-multiplicative}, they both need to compute a lower bound $a$ and upper bound $b$ on $P_{\alpha}(p)$. Both algorithms achieve this by observing \lem{Renyi-large-approx}, and they both compute $a$ and $b$ by recursively call the estimation of $P_{\alpha'}(p)$ for some $\alpha'$ closer to 1. However, in the case $\alpha>1$, \algo{EstRenyiLarge} makes $\alpha'$ smaller and smaller by multiplying $(1+\frac{1}{\log n})^{-1}$ each time, and ends the recursion when $\alpha'<1+\frac{1}{\log n}$; in the case $0<\alpha<1$, \algo{EstRenyiSmall} makes $\alpha'$ larger and larger by multiplying $(1-\frac{1}{\log n})^{-1}$ each time, and ends the recursion when $\alpha'>1-\frac{1}{\log n}$. This leads to different expressions for $a$ and $b$ in Line \ref{line:EstRenyiSmall-ab-1} and Line \ref{line:EstRenyiSmall-ab-2} of both algorithms, and technically the proofs for \lem{Renyi-large-3} and \lem{Renyi-small-3} is different;
\item Both algorithms boost the success probability to $1-\delta$ by repeating the algorithm for $\lceil 48\log\frac{1}{\delta}\rceil$ executions and taking the median, and their correctness is both promised by \lem{Renyi-large-4}.
\end{itemize}

\begin{proof}[Proof of \thm{Renyi-small}]
First, if the estimate $\tilde{p}_{i}$ in the subroutine $\mathcal{A}$ of \algo{EstRenyiSmall} were precisely accurate, the expectation of the subroutine's output would be $E:=\sum_{i\in\range{n}}p_{i}\cdot p_{i}^{\alpha-1}=P_{\alpha}(p)$. To be precise, we bound how far the actual expectation of the subroutine's output $\tilde{E}$ is from the exact value $P_{\alpha}(p)$. In \lem{Renyi-small-1}, we show that when taking $M=2^{\lceil\log_{2}(\frac{n^{1/2\alpha}}{\epsilon}\log(\frac{n^{1/2\alpha}}{\epsilon}))\rceil+1}$ queries in \textbf{EstAmp}, we have $|\tilde{E}-E|=O(\epsilon E)$.

As a result, to approximate $P_{\alpha}(p)$ within multiplicative error $\Theta(\epsilon)$, it is equivalent to approximate $\tilde{E}$ within multiplicative error $\Theta(\epsilon)$. Recall \thm{Monte-Carlo-multiplicative} showed that if the variance of the random variable output by $\mathcal{A}$ is at most $\sigma^{2}\tilde{E}^{2}$ for a known $\sigma$, and if we can obtain two values $a,b$ such that $\tilde{E}\in[a,b]$, then $\tilde{O}(\sigma b/\epsilon a)$ executions of $\mathcal{A}$ suffice to approximate $\tilde{E}$ within multiplicative error $\epsilon$ with success probability at least $9/10$. In the main body of the algorithm (Line \ref{line:EstRenyiSmall-input} to Line \ref{Line:EstRenyiSmall-output-delta}), we use \thm{Monte-Carlo-multiplicative} to approximate $\tilde{E}$.

On the one hand, in \lem{Renyi-small-2}, we show that for any $0<\alpha<1$, the variance is at most $2n^{1/\alpha-1}\tilde{E}^{2}$ with probability at least $\frac{8}{\pi^{2}}$. This gives $\sigma=\sqrt{2n^{1/\alpha-1}}=O(n^{1/2\alpha-1/2})$.

On the other hand, we need to compute the lower bound $a$ and upper bound $b$. As stated in the proof of \thm{Renyi-large}, for any $0<\alpha_{1}<\alpha_{2}$ with $\frac{\alpha_{2}}{\alpha_{1}}=1+O(\frac{1}{\log n})$,
\begin{align}
\sum_{i\in\range{n}}p_{i}^{\alpha_{1}}=\Theta\Big(\Big(\sum_{i\in\range{n}}p_{i}^{\alpha_{2}}\Big)^{\frac{\alpha_{1}}{\alpha_{2}}}\Big).
 \end{align}
As a result, we compute $a$ and $b$ by recursively calling \algo{EstRenyiSmall} to estimate $P_{\alpha'}(p)$ for $\alpha'=\alpha/(1-1/\log n)$, which is used to compute the lower bound $a$ and upper bound $b$ in Line \ref{line:EstRenyiSmall-ab-2}; the recursive call keeps until $\alpha>1-\frac{1}{\log n}$, when $a=1$ and $b=e$ (as in Line \ref{line:EstRenyiSmall-ab-1}) are simply lower and upper bounds on $P_{\alpha}(p)$.

To be precise, in \lem{Renyi-small-3}, we prove that $b/a\leq 4e=O(1)$, and with probability at least $1/e^{1/12}>0.92$, $a$ and $b$ are indeed lower and upper bounds on $P_{\alpha}(p)$, respectively; furthermore, in Line \ref{line:EstRenyiSmall-ab-2}, \algo{EstRenyiSmall} is recursively called by at most $\log n\log\frac{1}{\alpha}$ times, and each recursive call takes at most $\tilde{O}(n^{\frac{1}{\alpha}-\frac{1}{2}})$ queries. This promises that when we apply \cor{Master-multiplicative}, the cost $C_{a,b}$ is dominated by the query cost from \algo{Monte-Carlo-multiplicative}.

Combining all points above, \cor{Master-multiplicative} approximates $\tilde{E}$ up to multiplicative error $\Theta(\epsilon)$ with success probability at least $\frac{8}{\pi^{2}}\cdot 0.92\cdot 9/10>2/3$ using
\begin{align}
\log n\log\frac{1}{\alpha}\cdot\tilde{O}\Big(\frac{4e\cdot\sqrt{2n^{1/\alpha-1}}}{\epsilon}\Big)\cdot 2^{\lceil\log_{2}(\frac{n^{1/2\alpha}}{\epsilon}\log(\frac{n^{1/2\alpha}}{\epsilon}))\rceil+1}=\tilde{O}\Big(\frac{n^{\frac{1}{\alpha}-\frac{1}{2}}}{\epsilon^{2}}\Big)
\end{align}
quantum queries. Together with $|\tilde{E}-E|=O(\epsilon E)$ and rescale $l,M$ by a large enough constant, Line 1 to Line \ref{Line:EstRenyiSmall-output-2/3} in \algo{EstRenyiSmall} approximates $E=P_{\alpha}(p)$ up to multiplicative error $\epsilon$ with success probability at least $2/3$.

Finally, following from \lem{Renyi-large-4}, after repeating the procedure for $\lceil 48\log\frac{1}{\delta}\rceil$ executions and taking the median $\tilde{P}_{\alpha}(p)$ (as in Line \ref{Line:EstRenyiSmall-output-delta}), the success probability that $\tilde{P}_{\alpha}(p)$ approximates $P_{\alpha}(p)$ within multiplicative error $\epsilon$ is boosted to $1-\delta$.
\end{proof}

It remains to prove the lemmas mentioned above.

\subsubsection{Expectation of $\mathcal{A}$ is $\epsilon$-close to $P_{\alpha}(p)$}
\begin{lemma}\label{lem:Renyi-small-1}
$|\tilde{E}-E|=O(\epsilon E)$.
\end{lemma}

\begin{proof}[Proof of \lem{Renyi-small-1}]
For convenience, denote $m=\lceil\log_{2}(\frac{n^{1/2\alpha}}{\epsilon}\log(\frac{n^{1/2\alpha}}{\epsilon}))\rceil+1$, and $S_{0}, S_{1},\ldots, S_{m}$ the same as previous definitions. We still have \eqn{Shannon-cardinality}. By linearity of expectation,
\begin{align}\label{eqn:Renyi-small-binary-division}
|\tilde{E}-E|\leq \sum_{i\in\range{n}}p_{i}\E\Big[\Big|\frac{1}{\tilde{p}_{i}^{1-\alpha}}-\frac{1}{p_{i}^{1-\alpha}}\Big|\Big]=\sum_{j=0}^{m}\sum_{i\in S_{j}}p_{i}\E\big[\big|\tilde{p}_{i}^{\alpha-1}-p_{i}^{\alpha-1}\big|\big].
\end{align}
Therefore, to prove $|\tilde{E}-E|=O(\epsilon E)$ it suffices to show
\begin{align}\label{eqn:Renyi-small-binary-division-bound}
\sum_{j=0}^{m}\sum_{i\in S_{j}}p_{i}\E\big[\big|\tilde{p}_{i}^{\alpha-1}-p_{i}^{\alpha-1}\big|\big]=O\Big(\epsilon\sum_{i\in\range{n}}p_{i}^{\alpha}\Big).
\end{align}
Similar to the proof of \lem{Renyi-large-1}, we have
\begin{align}
\sum_{j=0}^{m}\sum_{i\in S_{j}}p_{i}\E\big[\big|\tilde{p}_{i}^{\alpha-1}-p_{i}^{\alpha-1}\big|\big]=O\Big(\frac{\pi^{2\alpha-1}m}{2^{2\alpha m}}\cdot \sum_{j=0}^{m}s_{j}2^{(2\alpha-1)j}\Big). \label{eqn:Renyi-small-LHS}
\end{align}
On the other side,
\begin{align}
\epsilon\sum_{i\in\range{n}}p_{i}^{\alpha}=\Theta\Big(\epsilon\sum_{j=0}^{m}s_{j}\cdot\big(\frac{2^{j}}{2^{m}}\pi\big)^{2\alpha}\Big)=\Theta\Big(\frac{\epsilon\cdot \pi^{2\alpha}}{2^{2\alpha m}}\sum_{j=0}^{m}s_{j}2^{2\alpha j}\Big). \label{eqn:Renyi-small-RHS}
\end{align}
Therefore, to prove Equation \eqn{Renyi-small-binary-division-bound}, by \eqn{Renyi-small-LHS} and \eqn{Renyi-small-RHS} it suffices to prove
\begin{align}
\sum_{j=0}^{m}s_{j}2^{(2\alpha-1)j}=O\Big(\frac{\epsilon}{m}\sum_{j=0}^{m}s_{j}2^{2\alpha j}\Big).
\end{align}
Since $m=\lceil\log_{2}(\frac{n^{1/2\alpha}}{\epsilon}\log(\frac{n^{1/2\alpha}}{\epsilon}))\rceil+1$, we have $\frac{2^{m}}{m}\geq \frac{n^{1/2\alpha}}{\epsilon}$, thus $\frac{\epsilon}{m}\geq\frac{n^{1/2\alpha}}{2^{m}}$. Therefore, it suffices to prove
\begin{align}\label{eqn:Renyi-small-Holder}
\sum_{j=0}^{m}s_{j}2^{(2\alpha-1)j}=O\Big(\frac{n^{1/2\alpha}}{2^{m}}\sum_{j=0}^{m}s_{j}2^{2\alpha j}\Big).
\end{align}

Since $s_{j}\in\N$, $s_{j}\leq s_{j}^{1/\alpha}$; as a result,
\begin{align}
\frac{\sum_{j=0}^{m}s_{j}2^{2j}}{\big(\sum_{j=0}^{m}s_{j}2^{2\alpha j}\big)^{1/\alpha}}\leq \frac{\sum_{j=0}^{m}(s_{j}2^{2\alpha j})^{1/\alpha}}{\big(\sum_{j=0}^{m}s_{j}2^{2\alpha j}\big)^{1/\alpha}}=\sum_{j=0}^{m}\Big(\frac{s_{j}2^{2\alpha j}}{\sum_{k=0}^{m}s_{k}2^{2\alpha k}}\Big)^{1/\alpha}\leq \sum_{j=0}^{m}\frac{s_{j}2^{2\alpha j}}{\sum_{k=0}^{m}s_{k}2^{2\alpha k}}=1.
\end{align}
Plugging \eqn{Shannon-cardinality} into the inequality above, we have
\begin{align}\label{eqn:Renyi-small-Holder-1}
\Big(\sum_{j=0}^{m}s_{j}2^{2\alpha j}\Big)^{\frac{1}{2\alpha}}=\Omega(2^{m}).
\end{align}
On the other side, by H\"{o}lder's inequality we have
\begin{align}\label{eqn:Renyi-small-Holder-2}
\Big(\sum_{j=0}^{m}s_{j}\Big)^{\frac{1}{2\alpha}}\Big(\sum_{j=0}^{m}s_{j}2^{2\alpha j}\Big)^{\frac{2\alpha-1}{2\alpha}}\geq \sum_{j=0}^{m}s_{j}2^{(2\alpha-1)j}.
\end{align}
Combining \eqn{Shannon-cardinality}, \eqn{Renyi-small-Holder-1}, and \eqn{Renyi-small-Holder-2}, we get exactly \eqn{Renyi-small-Holder}.
\end{proof}

\subsubsection{Bound the variance of $\mathcal{A}$ by the square of its expectation}
\begin{lemma}\label{lem:Renyi-small-2}
With probability at least $\frac{8}{\pi^{2}}$, the variance of the random variable output by $\mathcal{A}$ is at most $2n^{1/\alpha-1}\tilde{E}^{2}$.
\end{lemma}

\begin{proof}[Proof of \lem{Renyi-small-2}]
Because $\tilde{E}=\sum_{i=1}^{n}p_{i}\cdot\tilde{p}_{i}^{\alpha-1}$ and the variance is $\sum_{i=1}^{n}p_{i}\cdot (\tilde{p}_{i}^{\alpha-1})^{2}-\big(\sum_{i=1}^{n}p_{i}\cdot \tilde{p}_{i}^{\alpha-1}\big)^{2}\leq \sum_{i=1}^{n}p_{i}\cdot (\tilde{p}_{i}^{\alpha-1})^{2}$, it suffices to show that
\begin{align}\label{eqn:Renyi-small-fact-2-1}
\sum_{i=1}^{n}p_{i}\cdot (\tilde{p}_{i}^{\alpha-1})^{2}\leq 2n^{1/\alpha-1}\Big(\sum_{i=1}^{n}p_{i}\cdot\tilde{p}_{i}^{\alpha-1}\Big)^{2}.
\end{align}

By \thm{quantum-counting}, with probability at least $\frac{8}{\pi^{2}}$, we have
\begin{align}\label{eqn:Renyi-small-fact-2-2}
|\tilde{p}_{i}-p_{i}|\leq\frac{2\pi\sqrt{p_{i}}}{2^{m}}\leq\frac{\epsilon\pi\sqrt{p_{i}}}{n^{1/2\alpha}}\qquad i\in\range{n}.
\end{align}
As a result,
\begin{align}
\sum_{i=1}^{n}p_{i}(\tilde{p}_{i}^{\alpha-1})^{2}&\leq \sum_{i=1}^{n}p_{i}\Big(p_{i}-\frac{\epsilon\pi\sqrt{p_{i}}}{n^{1/2\alpha}}\Big)^{-2(1-\alpha)} \\
&= \sum_{i=1}^{n}p_{i}^{2\alpha-1}\Big(1-\frac{\epsilon\pi}{n^{1/2\alpha}\sqrt{p_{i}}}\Big)^{-2(1-\alpha)} \\
&\approx \sum_{i=1}^{n}p_{i}^{2\alpha-1}\Big(1+2(1-\alpha)\frac{\epsilon\pi}{n^{1/2\alpha}\sqrt{p_{i}}}\Big) \label{eqn:Renyi-small-fact-2-approx} \\
&= \sum_{i=1}^{n}p_{i}^{2\alpha-1}+\frac{2(1-\alpha)\epsilon\pi}{n^{1/2\alpha}}\sum_{i=1}^{n}p_{i}^{2\alpha-0.5}. \label{eqn:Renyi-small-fact-2-3}
\end{align}
Furthermore,
\begin{align}
\sqrt{n}\Big(\sum_{i=1}^{n}p_{i}^{2\alpha-1}\Big)\geq \Big(\sum_{j=1}^{n}\sqrt{p_{j}}\Big)\Big(\sum_{i=1}^{n}p_{i}^{2\alpha-1}\Big)\geq \sum_{i=j=1}^{n}\sqrt{p_{j}}p_{i}^{2\alpha-1}= \sum_{i=1}^{n}p_{i}^{2\alpha-0.5}.
\end{align}
Plugging this into \eqn{Renyi-small-fact-2-3}, we have
\begin{align}\label{eqn:Renyi-small-fact-2-4}
\sum_{i=1}^{n}p_{i}(\tilde{p}_{i}^{\alpha-1})^{2}&\leq\Big(1+\frac{2(1-\alpha)\epsilon\pi}{n^{1/2\alpha-1/2}}\Big)\sum_{i=1}^{n}p_{i}^{2\alpha-1}.
\end{align}
Using similar techniques, we can show
\begin{align}\label{eqn:Renyi-small-fact-2-5}
\Big(\sum_{i=1}^{n}p_{i}\cdot\tilde{p}_{i}^{\alpha-1}\Big)^{2}&\geq\Big(1-\frac{2(1-\alpha)\epsilon\pi}{n^{1/\alpha-1}}\Big)\Big(\sum_{i=1}^{n}p_{i}^{\alpha}\Big)^{2}.
\end{align}
Since $0<\alpha<1$,
\begin{align}\label{eqn:Renyi-small-fact-2-6}
\lim_{n\rightarrow\infty}1+\frac{2(1-\alpha)\epsilon\pi}{n^{1/2\alpha-1/2}}=1,\qquad \lim_{n\rightarrow\infty}1-\frac{2(1-\alpha)\epsilon\pi}{n^{1/\alpha-1}}=1.
\end{align}
Because \eqn{Renyi-small-fact-2-approx} only omits lower order terms and the limits in \eqn{Renyi-small-fact-2-6} are both 1, to prove \eqn{Renyi-small-fact-2-1} it suffices to prove that for large enough $n$,
\begin{align}
\sum_{i=1}^{n}p_{i}^{2\alpha-1}\leq n^{1/\alpha-1}\Big(\sum_{i=1}^{n}p_{i}^{\alpha}\Big)^{2}.
\end{align}

By generalized mean inequality, we have
\begin{align}
\Big(\frac{1}{n}\sum_{i=1}^{n}p_{i}^{2\alpha-1}\Big)^{\frac{1}{2\alpha-1}} \leq \Big(\frac{1}{n}\sum_{i=1}^{n}p_{i}^{\alpha}\Big)^{\frac{1}{\alpha}}.
\end{align}
Therefore,
\begin{align}
\sum_{i=1}^{n}p_{i}^{2\alpha-1}\leq n^{1-\frac{2\alpha-1}{\alpha}}\Big(\sum_{i=1}^{n}p_{i}^{\alpha}\Big)^{\frac{2\alpha-1}{\alpha}}= n^{1/\alpha-1}\Big(\sum_{i=1}^{n}p_{i}^{\alpha}\Big)^{2-1/\alpha} \leq n^{1/\alpha-1}\Big(\sum_{i=1}^{n}p_{i}^{\alpha}\Big)^{2}.
\end{align}
Hence the result follows.
\end{proof}

\subsubsection{Analyze the recursive calls}
\begin{lemma}\label{lem:Renyi-small-3}
With probability at least $0.92$, the $a$ and $b$ in Line \ref{line:EstRenyiSmall-ab-1} or Line \ref{line:EstRenyiSmall-ab-2} of \algo{EstRenyiSmall} are indeed lower and upper bounds on $P_{\alpha}(p)$, respectively, and $b/a=O(1)$; furthermore, in Line \ref{line:EstRenyiSmall-ab-2}, \algo{EstRenyiSmall} is recursively called for at most $\log n\log\frac{1}{\alpha}$ executions, and each recursive call takes at most $\tilde{O}(n^{\frac{1}{\alpha}-\frac{1}{2}})$ queries.
\end{lemma}

\begin{proof}[Proof of \lem{Renyi-small-3}]
Similar to \lem{Renyi-large-3}, we decompose the proof into two parts:
\begin{itemize}
\item \hd{In Line \ref{line:EstRenyiSmall-ab-2}, \algo{EstRenyiSmall} is recursively called for at most $\log n\log\frac{1}{\alpha}$ executions, and each recursive call takes at most $\tilde{O}(n^{\frac{1}{\alpha}-\frac{1}{2}})$ queries:}

Because each recursive call of \algo{EstRenyiSmall} increases $\alpha$ by multiplying $(1-\frac{1}{\log n})^{-1}$ and the recursion ends when $\alpha>1-\frac{1}{\log n}$, the total number of recursive calls is at most $\frac{\log\alpha}{\log(1-\frac{1}{\log n})}\leq \log n\log\frac{1}{\alpha}$.
\\\\
When $\alpha>1-\frac{1}{\log n}$, $a$ and $b$ are set in Line \ref{line:EstRenyiSmall-ab-1} and no extra queries are needed; when Line \ref{line:EstRenyiSmall-ab-2} calls $\alpha(1-\frac{1}{\log n})^{-k}$-power sum estimation for some $k\in\N$, by induction on $k$, we see that this call takes at most $\tilde{O}\big(n^{\frac{(1-1/\log n)^{k}}{\alpha}-\frac{1}{2}}\big)\leq \tilde{O}(n^{\frac{1}{\alpha}-\frac{1}{2}})$ queries. As a result, when we apply \cor{Master-multiplicative}, the cost $C_{a,b}$ is dominated by the query cost from \algo{Monte-Carlo-multiplicative}.

\item \hd{With probability at least $0.92$, $a$ and $b$ are lower and upper bounds on $P_{\alpha}(p)$ respectively, and $b/a=O(1)$:}

When $1-\frac{1}{\log n}<\alpha<1$, on the one hand we have $\sum_{i=1}^{n}p_{i}^{\alpha}\geq\sum_{i=1}^{n}p_{i}=1$; on the other hand, because $n^{\frac{1}{\log n}}=e$, by \lem{Renyi-large-approx} we have
\begin{align}
\sum_{i=1}^{n}p_{i}^{\alpha}\leq n^{1-\alpha}\Big(\sum_{i=1}^{n}p_{i}\Big)^{\alpha}\leq e.
\end{align}
Therefore, $a=1$ and $b=e$ in Line \ref{line:EstRenyiSmall-ab-1} are lower and upper bounds on $P_{\alpha}(p)$ respectively, and $b/a=e=O(1)$.
\\\\
When $\alpha<1-\frac{1}{\log n}$, for convenience denote $\alpha'=\alpha(1-\frac{1}{\log n})^{-1}$. As justified above, the total number of recursive calls in Line \ref{line:EstRenyiSmall-ab-2} is at most $\log n\log\frac{1}{\alpha}$. Because we take $\delta=\frac{1}{12\log n\log 1/\alpha}$ in Line \ref{line:EstRenyiSmall-ab-2}, with probability at least
\begin{align}
\Big(1-\frac{1}{12\log n\log 1/\alpha}\Big)^{\log n\log 1/\alpha}\geq\frac{1}{e^{1/12}}>0.92,
\end{align}
the output of every recursive call is within $1/2$-multiplicative error. As a result, the $P$ in Line \ref{line:EstRenyiSmall-ab-2} satisfies $P/2\leq \sum_{i=1}^{n}p_{i}^{\alpha'}\leq 2P$. Combining this with \lem{Renyi-large-approx} and using $n^{\frac{1}{\log n}}=e$, we have
\begin{align}
(P/2)^{1-\frac{1}{\log n}}\leq \sum_{i=1}^{n}p_{i}^{\alpha}\leq e(2P)^{1-\frac{1}{\log n}}.
\end{align}
In other words, $a$ and $b$ are indeed lower and upper bounds on $P_{\alpha}(p)$, respectively. Furthermore, $b/a=O(1)$ because
\begin{align}
\frac{b}{a}=e\cdot 4^{1-\frac{1}{\log n}}\leq 4e=O(1).
\end{align}
\end{itemize}
\end{proof}


\section{Integer R{\'e}nyi entropy estimation}\label{sec:integer-Renyi}
Recall the classical query complexity of $\alpha$-R{\'e}nyi entropy estimation for $\alpha\in\N,\alpha\geq 2$ is $\Theta(n^{1-1/\alpha})$ \cite{acharya2017estimating}, which is smaller than non-integer cases. Quantumly, we also provide a more significant speedup.

Given the oracle $O_{p}\colon \range{S}\rightarrow\range{n}$ in \eqn{BHH-query-model}, we denote the occurrences of $1,2,\ldots,n$ among $O_{p}(1),\ldots$, $O_{p}(S)$ as $m_{1},\ldots,m_{n}$, respectively. A key observation is that by \eqn{BHH-Op}, we have
\begin{align}
P_{\alpha}(p)=\sum_{i=1}^{n}(m_{i}/S)^{\alpha}=S^{-\alpha}\sum_{i=1}^{n}m_{i}^{\alpha}.
\end{align}
Therefore, it suffices to approximate $\sum_{i\in\range{n}}m_{i}^{\alpha}$, which is known as the \emph{$\alpha$-frequency moment} of $O_{p}(1),\ldots,O_{p}(S)$. Based on the quantum algorithm for $\alpha$-distinctness~\cite{belovs2012learning}, Montanaro~\cite{montanaro2015frequency} proved:

\begin{fact}[\cite{montanaro2015frequency}, Step 3b-step 3e in Algorithm 2; Lemma 4]\label{fac:frequency-moments}
Fix $l$ where $l\in\{1,\ldots,n\}$. Let $s_{1},\ldots,s_{l}\in\range{S}$ be picked uniformly at random, and denote the number of $\alpha$-wise collisions in $\{O_{p}(s_{1}),\ldots,O_{p}(s_{l})\}$ as $C(s_{1},\ldots,s_{l})$. Then:
\begin{itemize}
\item $C(s_{1},\ldots,s_{l})$ can be computed using $O(l^{\nu}\log(l/\epsilon^{2}))$ queries to $\hat{O}_{p}$ with failure probability at most $O(\epsilon^{2}/l)$, where $\nu:=1-2^{\alpha-2}/(2^{\alpha}-1)<\frac{3}{4}$;
\item $\E[C(s_{1},\ldots,s_{l})]={l\choose\alpha}P_{\alpha}(p)$ and $\var[C(s_{1},\ldots,s_{l})]=O(1)$.
\end{itemize}
\end{fact}

However, a direct application of \cite{montanaro2015frequency} will lead to a complexity depending on $S$ (in particular, $l$ in \fac{frequency-moments} can be as large as $S$) rather than $n$. Our solution is \algo{Renyi-integer} that is almost the same as Algorithm 2 in \cite{montanaro2015frequency} except Line \ref{line:integer-Renyi-l} and Line \ref{line:integer-Renyi-i}, where we set $2^{\lceil\log_{2}\alpha n\rceil}$ as an upper bound on $l$. We claim that such choice of $l$ is valid because by the pigeonhole principle, $\alpha n$ elements $O_{p}(s_{1}),\ldots,O_{p}(s_{\alpha n})$ in $\range{n}$ must have an $\alpha$-collision, so the first for-loop must terminate at some $i\leq\lceil\log_{2}\alpha n\rceil$. With this modification, we have \thm{Renyi-integer} for integer R{\'e}nyi entropy estimation.

\begin{algorithm}[htbp]
Set $l=2^{\lceil\log_{2}\alpha n\rceil}$\; \label{line:integer-Renyi-l}
\For{$i=0,\ldots,\lceil\log_{2}\alpha n\rceil$\label{line:integer-Renyi-i}}
	{Pick $s_{1},\ldots,s_{2^{i}}\in\range{S}$ uniformly at random and let $\mathcal{S}$ be the sequence $O_{p}(s_{1}),\ldots,O_{p}(s_{2^{i}})$\;
	Apply the $\alpha$-distinctness algorithm in \cite{belovs2012learning} to $\mathcal{S}$ with failure probability $\frac{1}{10\lceil\log_{2}\alpha n\rceil}$\;
	If it returns a set of $\alpha$ equal elements, set $l=2^{i}$ and terminate the loop\; \label{line:integer-pigeonhole}}
Set $M=\lceil K/\epsilon^{2}\rceil$ for some $K=\Theta(1)$ \label{line:integer-Renyi-K}\;
\For{$r=1,\ldots,M$}
	{Pick $s_{1},\ldots,s_{l}\in\range{S}$ uniformly at random\;
	Apply the first bullet in \fac{frequency-moments} to give an estimate $C^{(r)}$ of the number of $\alpha$-wise collisions in $\{O_{p}(s_{1}),\ldots,O_{p}(s_{l})\}$\;}
Output $\tilde{P}_{\alpha}(p)=\frac{1}{M{l\choose \alpha}}\sum_{r=1}^{M}C^{(r)}$\; \label{line:integer-Renyi-output}
\caption{Estimate the $\alpha$-power sum $P_{\alpha}(p)$ of $p=(p_{i})_{i=1}^{n}$ on $\range{n}$, $\alpha>1,\alpha\in\N$.}
\label{algo:Renyi-integer}
\end{algorithm}

\begin{theorem}\label{thm:Renyi-integer}
Assume $\alpha>1,\alpha\in\N$. \algo{Renyi-integer} approximates $P_{\alpha}(p)$ within a multiplicative error $0<\epsilon\leq O(1)$ with success probability at least $\frac{2}{3}$ using $\tilde{O}\big(\frac{n^{\nu(1-1/\alpha)}}{\epsilon^{2}}\big)=o\big(\frac{n^{\frac{3}{4}(1-1/\alpha)}}{\epsilon^{2}}\big)$ quantum queries to $p$, where $\nu:=1-2^{\alpha-2}/(2^{\alpha}-1)<\frac{3}{4}$.
\end{theorem}

Our proof of \thm{Renyi-integer} is inspired by the proof of Theorem 5 in \cite{montanaro2015frequency}.
\begin{proof}
Because $O_{p}$ takes values in $\range{n}$, by pigeonhole principle, for any $s_{1},\ldots,s_{\alpha n}\in\range{S}$ there exists a $\alpha$-wise collision among $O_{p}(s_{1}),\ldots,O_{p}(s_{\alpha n})$. Therefore, Line \ref{line:integer-pigeonhole} terminates the first loop with some $l\leq 2^{\lceil\log_{2}\alpha n\rceil}$ with probability at least $(1-1/10\lceil\log_{2}\alpha n\rceil)^{\lceil\log_{2}\alpha n\rceil}\geq e^{-1/10}>0.9$.

Moreover, tighter bounds on $l$ are established next. On the one hand, by Chebyshev's inequality and \fac{frequency-moments}, the probability that the first for-loop fails to terminate when $l\leq\frac{B}{P_{\alpha}(p)^{1/\alpha}}$ for some constant $B>0$ is at most
\begin{align}
\Pr\big[C(s_{1},\ldots,s_{l})=0\big]\leq\frac{\var[C(s_{1},\ldots,s_{l})]}{\E[C(s_{1},\ldots,s_{l})]^{2}}=O\Big(\frac{1}{l^{2\alpha}P_{\alpha}(p)^{2}}\Big)=O\Big(\frac{1}{B^{2\alpha}}\Big).
\end{align}
Therefore, taking a large enough $B$ ensures that $l=O\big(\frac{1}{P_{\alpha}(p)^{1/\alpha}}\big)$ with failure probability at most $1/20$. On the other hand, by Markov's inequality and \fac{frequency-moments}, we have
\begin{align}
\Pr\big[C(s_{1},\ldots,s_{l})\geq 1\big]\leq \E[C(s_{1},\ldots,s_{l})]=O(l^{\alpha}P_{\alpha}(p)).
\end{align}
As a result, the probability that the first for-loop terminates when  $l\leq \frac{A}{P_{\alpha}(p)^{1/\alpha}}$ for some constant $A>0$ is at most
\begin{align}
O(P_{\alpha}(p))\cdot\sum_{i=0}^{\big\lfloor\log_{2}\big(\frac{A}{P_{\alpha}(p)^{1/\alpha}}\big)\big\rfloor}2^{i\alpha}=O(A^{\alpha}).
\end{align}
Therefore, taking a small enough $A>0$ ensures that $l=\Omega\big(\frac{1}{P_{\alpha}(p)^{1/\alpha}}\big)$ with failure probability at most $1/20$. In all, we have $l=\Theta\big(\frac{1}{P_{\alpha}(p)^{1/\alpha}}\big)$ with probability at least 0.9.

By \fac{frequency-moments}, the output $\E[\tilde{P}_{\alpha}(p)]$ in Line \ref{line:integer-Renyi-output} of \algo{Renyi-integer} satisfies
\begin{align}
\E[\tilde{P}_{\alpha}(p)]=\frac{1}{M{l\choose \alpha}}\sum_{r=1}^{M}\E[C^{(r)}]=P_{\alpha}(p),\quad \var[\tilde{P}_{\alpha}(p)]=\frac{1}{(M{l\choose \alpha})^{2}}\sum_{r=1}^{M}\var[C^{(r)}]=O\Big(\frac{1}{Ml^{2\alpha}}\Big).
\end{align}
Therefore, by Chebyshev's inequality and recall $l=\Theta\big(\frac{1}{P_{\alpha}(p)^{1/\alpha}}\big)$, we have
\begin{align}
\Pr\big[|\tilde{P}_{\alpha}(p)-P_{\alpha}(p)|\geq\epsilon P_{\alpha}(p)\big]\leq O\Big(\frac{1}{Ml^{2\alpha}\epsilon^{2}P_{\alpha}(p)^{2}}\Big)=O\Big(\frac{1}{K}\Big).
\end{align}
Taking a large enough constant $K$ in Line \ref{line:integer-Renyi-K} of \algo{Renyi-integer}, we have $\Pr\big[|\tilde{P}_{\alpha}(p)-P_{\alpha}(p)|\leq\epsilon P_{\alpha}(p)\big]\geq 0.9$. In all, with probability at least $0.9\times  0.9\times 0.9>2/3$, $\tilde{P}_{\alpha}(p)$ approximates $P_{\alpha}(p)$ within multiplicative error $\epsilon$.

For the rest of the proof, it suffices to compute the quantum query complexity of \algo{Renyi-integer}. Because the $\alpha$-distinctness algorithm  on $m$ elements in \cite{belovs2012learning} takes $O(m^{\nu}\log(1/\delta))$ quantum queries when the success probability is $1-\delta$, the first for-loop in \algo{Renyi-integer} takes $\sum_{i=0}^{\log_{2}l}O(2^{\nu i}\log\lceil\log_{2}\alpha n\rceil)=\tilde{O}(l^{\nu})=\tilde{O}(n^{\nu(1-1/\alpha)})$ quantum queries because
\begin{align}\label{eqn:Renyi-integer-power-sum}
l=\Theta\Big(\frac{1}{P_{\alpha}(p)^{1/\alpha}}\Big)=O(n^{1-1/\alpha}),
\end{align}
following from $P_{\alpha}(p)\geq n^{1-\alpha}$. The second for-loop takes $\lceil K/\epsilon^{2}\rceil\cdot O(l^{\nu}\log(l/\epsilon^{2}))=\tilde{O}(\frac{n^{\nu(1-1/\alpha)}}{\epsilon^{2}}\big)$ quantum queries by \fac{frequency-moments} and \eqn{Renyi-integer-power-sum}. In total, the number of quantum queries is $\tilde{O}(\frac{n^{\nu(1-1/\alpha)}}{\epsilon^{2}}\big)$.
\end{proof}

\begin{remark}\label{remark:alpha_dependence}
In \thm{Renyi-integer}, we regard $\alpha$ as a constant, i.e., the query complexity $\tilde{O}(\frac{n^{\nu(1-1/\alpha)}}{\epsilon^{2}}\big)$ hides the multiple in $\alpha$. In fact, by analyzing the dependence on $\alpha$ carefully in the above proof, the query complexity of \algo{Renyi-integer} is actually
\begin{align}\label{eqn:integer-Renyi-alpha-dependence}
\tilde{O}\Big(\alpha^{8\alpha^{2}}\cdot\frac{n^{\nu(1-1/\alpha)}}{\epsilon^{2}}\Big).
\end{align}
The dependence on $\alpha$ is super-exponential; therefore, \algo{Renyi-integer} is not good enough to approximate min-entropy (i.e., $\alpha=\infty$). As a result, we give the quantum algorithm for estimating min-entropy separately (see \sec{min-entropy}).
\end{remark}


\section{Min-entropy estimation}\label{sec:min-entropy}
Since the min-entropy of $p$ is $H_{\infty}(p)=-\log\max_{i\in\range{n}}p_{i}$ by \eqn{min-entropy-definition}, it is equivalent to approximate $\max_{i\in\range{n}}p_{i}$ within multiplicative error $\epsilon$. We propose \algo{min-entropy} below to achieve this task.

\begin{algorithm}[htbp]
Set $\lambda=1$\;
\While{$\lambda\leq n$}
	{Take $M\sim\poi(\frac{16\lambda\log n}{\epsilon^{2}})$. Pick $s_{1},\ldots,s_{M}\in\range{S}$ uniformly at random and let $\mathcal{S}$ be the sequence $O_{p}(s_{1}),\ldots,O_{p}(s_{M})$\; \label{line:min-entropy-Poisson}
	Apply a $\lceil\frac{16\log n}{\epsilon^{2}}\rceil$-distinctness quantum algorithm to $\mathcal{S}$ with failure probability at most $\frac{\epsilon}{2\log n}$\; \label{line:min-entropy-distinctness}
	If Line \ref{line:min-entropy-distinctness} outputs a $\lceil\frac{16\log n}{\epsilon^{2}}\rceil$-collision of elements $i^{*}\in\range{n}$, apply \thm{quantum-counting} to approximate $p_{i^{*}}$ with multiplicative error $\epsilon$ and output its result; if not, set $\lambda\leftarrow\lambda\cdot\sqrt{1+\epsilon}$ and jump to the start of the loop\;  \label{line:min-entropy-quantum-counting}}
If $\lambda>n$ and no output has been given, output $1/n$;
\caption{Estimate $\max_{i\in\range{n}}p_{i}$ of a discrete distribution $p=(p_{i})_{i=1}^{n}$ on $\range{n}$.}
\label{algo:min-entropy}
\end{algorithm}

A key property of the Poisson distribution is that if we take $M\sim\poi(\nu)$ samples from $p$ (as in Line \ref{line:min-entropy-Poisson}), then for each $j\in\range{n}$, the number of occurrences of $j$ in $O_{p}(s_{1}),\ldots,O_{p}(s_{M})$ follows the Poisson distribution $M_{j}\sim\poi(\nu p_{j})$, and $M_{j},M_{j'}$ are independent for all $j\neq j'$. Furthermore:
\begin{lemma}\label{lem:min-entropy}
Let $X\sim\poi(\mu)$. Then, if $\mu<\frac{1}{\sqrt{1+\epsilon}}\cdot \frac{16\log n}{\epsilon^{2}}$, we have
\begin{align}\label{eqn:Poisson-lower}
\Pr\Big[X\geq\frac{16\log n}{\epsilon^{2}}\Big]\leq\frac{1}{n^{2}};
\end{align}
If $\mu\geq \frac{16\log n}{\epsilon^{2}}$, we have
\begin{align}\label{eqn:Poisson-upper}
\Pr\Big[X\geq\frac{16\log n}{\epsilon^{2}}\Big]>0.15.
\end{align}
\end{lemma}

Based on \lem{min-entropy}, our strategy is to set $\frac{16\log n}{\epsilon^{2}}$ as a threshold, take $\nu=\frac{16\lambda\log n}{\epsilon^{2}}$ as in Line \ref{line:min-entropy-Poisson}, and gradually increase the parameter $\lambda$. For convenience, denote $p_{i^{*}}=\max_{i}p_{i}$. As long as $\nu\cdot p_{i^{*}}<\frac{16\log n}{\epsilon^{2}}$, with high probability there is no $\lceil\frac{16\log n}{\epsilon^{2}}\rceil$-collision in $\mathcal{S}$, the distinctness quantum algorithm in Line \ref{line:min-entropy-distinctness} rejects, and $\lambda$ increases by multiplying $\sqrt{1+\epsilon}$ in Line \ref{line:min-entropy-quantum-counting}; right after the first time when $\nu\cdot p_{i^{*}}\geq\frac{16\log n}{\epsilon^{2}}$, with probability at least $0.15$, $i^{*}$ has a $\lceil\frac{16\log n}{\epsilon^{2}}\rceil$-collision in $\mathcal{S}$, while all other entries in $\range{n}$ do not (with failure probability at most $1/n^{2}$). In this case, with probability at least $\Omega(1)$, the distinctness quantum algorithm in Line \ref{line:min-entropy-distinctness} captures $i^{*}$, and the quantum counting (\thm{quantum-counting}) in Line \ref{line:min-entropy-quantum-counting} computes $p_{i^{*}}$ within multiplicative error $\epsilon$.

\begin{theorem}\label{thm:min-entropy}
\algo{min-entropy} approximates $\max_{i\in\range{n}}p_{i}$ within a multiplicative error $0<\epsilon\leq 1$ with success probability at least $\Omega(1)$ using $\tilde{O}\big(Q(\lceil\frac{16\log n}{\epsilon^{2}}\rceil\textsf{-distinctness})\big)$ quantum queries to $p$, where $Q(\lceil\frac{16\log n}{\epsilon^{2}}\rceil\textsf{-distinctness})$ is the quantum query complexity of the $\lceil\frac{16\log n}{\epsilon^{2}}\rceil$-distinctness problem.
\end{theorem}

We first prove \lem{min-entropy}.
output
\begin{proof}[Proof of \lem{min-entropy}]
First, we prove \eqn{Poisson-lower}\footnote{The tail bound of Poisson distributions is also studied elsewhere, for example, in \cite[Exercise 4.7]{mitzenmacher2005probability}.}. In \cite{glynn1987upper}, it is shown that if $\lambda>0$ and $X\sim\poi(\lambda)$, then for any $\nu>1$ we have
\begin{align}
\Pr[X\geq\nu\lambda]\leq \frac{e^{-\lambda}\lambda^{\nu\lambda}}{(\nu\lambda)!(1-1/\nu)}.
\end{align}
Taking $\lambda=\frac{1}{\sqrt{1+\epsilon}}\cdot \frac{16\log n}{\epsilon^{2}}$ and $\nu=\sqrt{1+\epsilon}$, by Sterling's formula we have
\begin{align}\label{eqn:Poisson-estimation-1}
\Pr\Big[X\geq\frac{16\log n}{\epsilon^{2}}\Big]\leq \frac{e^{-\lambda}\lambda^{\nu\lambda}}{(\nu\lambda)!(1-1/\nu)}\approx \frac{2}{\epsilon}\frac{e^{-\lambda}\lambda^{\nu\lambda}}{\sqrt{2\pi\nu\lambda}(\nu\lambda/e)^{\nu\lambda}}\approx \sqrt{\frac{2}{\pi}}\frac{1}{\epsilon\sqrt{\lambda}}\Bigg(\frac{e^{\epsilon/2}}{(1+\frac{\epsilon}{2})^{1+\epsilon/2}}\Bigg)^{\lambda}.
\end{align}
Because
\begin{align}
\lim_{\epsilon\rightarrow 0}\Bigg(\frac{e^{\epsilon/2}}{(1+\frac{\epsilon}{2})^{1+\epsilon/2}}\Bigg)^{8/\epsilon^{2}}&=\lim_{\epsilon\rightarrow 0}\exp\Big[\frac{4}{\epsilon}-\Big(\frac{8}{\epsilon^{2}}+\frac{4}{\epsilon}\Big)\ln\Big(1+\frac{\epsilon}{2}\Big)\Big] \\
&=\lim_{\epsilon\rightarrow 0}\exp\Big[\frac{4}{\epsilon}-\Big(\frac{8}{\epsilon^{2}}+\frac{4}{\epsilon}\Big)\Big(\frac{\epsilon}{2}-\frac{\epsilon^{2}}{8}+O(\epsilon^{3})\Big)\Big] \\
&=\lim_{\epsilon\rightarrow 0}\exp[-1+O(\epsilon)]=e^{-1},
\end{align}
we have
\begin{align}\label{eqn:Poisson-estimation-2}
\Bigg(\frac{e^{\epsilon/2}}{(1+\frac{\epsilon}{2})^{1+\epsilon/2}}\Bigg)^{\lambda}\approx e^{-\frac{\epsilon^{2}}{8}\lambda}\approx \frac{1}{n^{2}}.
\end{align}
Plugging \eqn{Poisson-estimation-2} into \eqn{Poisson-estimation-1}, we have
\begin{align}
\Pr\Big[X\geq\frac{16\log n}{\epsilon^{2}}\Big]\lesssim \sqrt{\frac{2}{\pi}}\frac{1}{\sqrt{16 \log n}}\frac{1}{n^{2}}\leq\frac{1}{n^{2}}.
\end{align}

Now we prove \eqn{Poisson-upper}. A theorem of Ramanujan \cite[Question 294]{hardy1927collected} states that for any positive integer $M$,
\begin{align}\label{eqn:Ramanujan-equation}
\frac{1}{2}e^{M}=\sum_{m=0}^{M-1}\frac{M^{m}}{m!}+\theta(M)\cdot\frac{M^{M}}{M!},
\end{align}
where $\frac{1}{3}\leq\theta(M)\leq\frac{1}{2}\ \forall\,M\in\N$. Because $\sum_{m=0}^{\infty}\frac{M^{m}}{m!}=e^{M}$, by \eqn{Ramanujan-equation} we have
\begin{align}
\sum_{m=M+1}^{\infty}\frac{M^{m}}{m!}+\frac{2}{3}\cdot\frac{M^{M}}{M!}\geq\frac{1}{2}e^{M}.
\end{align}
By Stirling's formula, $M!\geq\sqrt{2\pi M}\big(\frac{M}{e}\big)^{M}$. As a result,
\begin{align}\label{eqn:Ramanujan}
\sum_{m=M+1}^{\infty}\frac{M^{m}}{m!}\geq \frac{1}{2}e^{M}-\frac{2}{3}\cdot\frac{M^{M}}{\sqrt{2\pi M}\big(\frac{M}{e}\big)^{M}}=\Big(\frac{1}{2}-\frac{1}{\sqrt{4.5\pi M}}\Big)e^{M}.
\end{align}
We take $M=\lfloor\frac{16\log n}{\epsilon^{2}}\rfloor$. By \eqn{Ramanujan}, we have
\begin{align}
\Pr\Big[X\geq\frac{16\log n}{\epsilon^{2}}\Big]=e^{-\frac{16\log n}{\epsilon^{2}}}\sum_{m=M+1}^{\infty}\frac{(\frac{16\log n}{\epsilon^{2}})^{m}}{m!}\geq e^{-M-1}\sum_{m=M+1}^{\infty}\frac{M^{m}}{m!}\geq\frac{1}{2e}-\frac{1}{e\sqrt{4.5\pi M}}.
\end{align}
Because $0<\epsilon\leq 1$, we have $M\geq\lfloor 16\log 2\rfloor=11$. Therefore,
\begin{align}
\Pr\Big[X\geq\frac{16\log n}{\epsilon^{2}}\Big]\geq\frac{1}{2e}-\frac{1}{e\sqrt{4.5\pi\cdot 11}}>0.15.
\end{align}
\end{proof}

\begin{proof}[Proof of \thm{min-entropy}]
Denote $\sigma$ to be the permutation on $\range{n}$ such that $p_{\sigma(1)}\geq p_{\sigma(2)}\geq\cdots\geq p_{\sigma(n)}$.
Without loss of generality, we assume that $p_{\sigma(2)}\leq\frac{p_{\sigma(1)}}{1+\epsilon}$; otherwise, $p_{\sigma(2)}$ is close enough to $p_{\sigma(1)}$ in the sense that applying quantum counting to $p_{\sigma(2)}$ within multiplicative error $\epsilon$ gives an approximation to $p_{\sigma(1)}$ within multiplicative error $2\epsilon$. We may assume that every call of the $\lceil\frac{16\log n}{\epsilon^{2}}\rceil$-distinctness quantum algorithm in Line \ref{line:min-entropy-distinctness} of \algo{min-entropy} succeeds if and only if a $\lceil\frac{16\log n}{\epsilon^{2}}\rceil$-collision exists, because this happens with probability at least $\big(1-\frac{\epsilon}{2\log n}\big)^{\log_{\sqrt{1+\epsilon}} n}\geq e^{-1}=\Omega(1)$; for convenience, this is always assumed in the result of the proof.

On the one hand, when
\begin{align}\label{eqn:Poisson-close-to-bound}
\frac{p_{\sigma(1)}\cdot 16\lambda\log n}{\epsilon^{2}}<\frac{1}{\sqrt{1+\epsilon}}\cdot \frac{16\log n}{\epsilon^{2}},
\end{align}
by \lem{min-entropy} we have $\Pr\big[M_{i}\geq\frac{16\log n}{\epsilon^{2}}\big]\leq\frac{1}{n^{2}}\ \forall\,i\in\range{n}$, where $M_{i}$ is the occurences of $i$. Therefore, by the union bound, with probability at least $1-n\cdot\frac{1}{n^{2}}=1-\frac{1}{n}$, there is no $\frac{16\log n}{\epsilon^{2}}$-collision in $\mathcal{S}$. Since the while loop only has at most $\log_{\sqrt{1+\epsilon}} n=O(\frac{\log n}{\epsilon})$ rounds and $(1-\frac{1}{n})^{\log n/\epsilon}=1-o(1)$, we may assume that as long as \eqn{Poisson-close-to-bound} holds, Line \ref{line:min-entropy-distinctness} of \algo{min-entropy} always has a negative output and Line \ref{line:min-entropy-quantum-counting} enforces $\lambda\leftarrow\lambda\cdot\sqrt{1+\epsilon}$ and jumps to the start of the while loop.

The while loop keeps iterating until \eqn{Poisson-close-to-bound} is violated. In the second iteration after \eqn{Poisson-close-to-bound} is violated, we have
\begin{align}\label{eqn:Poisson-close-to-bound-2}
\frac{16\log n}{\epsilon^{2}}\leq \frac{p_{\sigma(1)}\cdot 16\lambda\log n}{\epsilon^{2}}<\sqrt{1+\epsilon}\cdot \frac{16\log n}{\epsilon^{2}};
\end{align}
since $p_{\sigma(2)}\leq\frac{p_{\sigma(1)}}{1+\epsilon}$, we have
\begin{align}\label{eqn:Poisson-close-to-bound-3}
\frac{p_{\sigma(2)}\cdot 16\lambda\log n}{\epsilon^{2}}<\frac{1}{\sqrt{1+\epsilon}}\frac{16\log n}{\epsilon^{2}}.
\end{align}
As a result, by \lem{min-entropy} we have
\begin{align}
\Pr\Big[M_{\sigma(1)}\geq\frac{16\log n}{\epsilon^{2}}\Big]>0.15; \qquad \Pr\Big[M_{i}\geq\frac{16\log n}{\epsilon^{2}}\Big]\leq\frac{1}{n^{2}}\quad\forall\,i\in\range{n}/\{\sigma(1)\}.
\end{align}
Therefore, $\Pr\big[\text{Line \ref{line:min-entropy-distinctness} outputs $\sigma(1)$ in the second iteration after \eqn{Poisson-close-to-bound} is violated}\big]\geq 0.15\cdot\big(1-\frac{n-1}{n^{2}}\big)^{n-1}$. In the first iteration after \eqn{Poisson-close-to-bound} is violated, we still have $\Pr\Big[M_{i}\geq\frac{16\log n}{\epsilon^{2}}\Big]\leq\frac{1}{n^{2}}\ \forall\,i\in\range{n}/\{\sigma(1)\}$. Therefore,
\begin{align}
&\Pr\big[\text{Line \ref{line:min-entropy-distinctness} outputs $\sigma(1)$ in the first or second iteration after \eqn{Poisson-close-to-bound} is violated}\big] \nonumber \\
&\qquad\geq 0.15\cdot\Big(1-\frac{n-1}{n^{2}}\Big)^{n-1}\cdot\Big(1-\frac{n-1}{n^{2}}\Big)^{n-1}\geq\frac{0.15}{e^{2}}=\Omega(1).
\end{align}

In all, with probability $\Omega(1)$, Line \ref{line:min-entropy-distinctness} of \algo{min-entropy} outputs $\sigma(1)$ correctly in the first or second iteration after \eqn{Poisson-close-to-bound} is violated; after that, the quantum counting in Line \ref{line:min-entropy-quantum-counting} approximates $p_{\sigma(1)}=\max_{i\in\range{n}}p_{i}$ within multiplicative error $\epsilon$. This establishes the correctness of \algo{min-entropy}.

It remains to show that the quantum query complexity of \algo{min-entropy} is $\tilde{O}\big(Q(\lceil\frac{16\log n}{\epsilon^{2}}\rceil\textsf{-distinctness})\big)$. Because there are at most $\log_{\sqrt{1+\epsilon}} n=O(\frac{\log n}{\epsilon})$ iterations in the while loop, the $\lceil\frac{16\log n}{\epsilon^{2}}\rceil$-distinctness algorithm in Line \ref{line:min-entropy-distinctness} is called for at most $O(\frac{\log n}{\epsilon})$ times; if it gives a $\lceil\frac{16\log n}{\epsilon^{2}}\rceil$-collision, because $\max_{i\in\range{n}}p_{i}\geq 1/n$, the quantum query complexity caused by Line \ref{line:min-entropy-quantum-counting} is $O(\frac{\sqrt{n}}{\epsilon})$ by \thm{quantum-counting}, which is smaller than the $\Omega(n^{2/3})$ quantum lower bound on the distinctness problems \cite{aaronson2004quantum}. As a result, the query complexity of \algo{min-entropy} in total is at most
\begin{align}
O\Big(\frac{\log n}{\epsilon}\Big)\cdot Q\Big(\Big\lceil\frac{16\log n}{\epsilon^{2}}\Big\rceil\textsf{-distinctness}\Big)+O\Big(\frac{\sqrt{n}}{\epsilon}\Big)=\tilde{O}\Big(Q\Big(\Big\lceil\frac{16\log n}{\epsilon^{2}}\Big\rceil\textsf{-distinctness}\Big)\Big).
\end{align}
\end{proof}
\vspace{-2mm}
\begin{remark}
In some special cases, \algo{min-entropy} already demonstrates provable quantum speedup. Recall the state-of-the-art quantum algorithm for $k$-distinctness is \cite{belovs2012learning} by Belovs, which has query complexity $O(2^{k^{2}}n^{1-2^{k-2}/(2^{k}-1)})$; however, this is superlinear when $k=\Theta(\log n)$. Nevertheless, if we are promised that $H_{\infty}(p)\leq f(n)$ for some $f(n)=o(\sqrt{\log n})$, then we can replace the $n$ in Line 2 of \algo{min-entropy} by $e^{f(n)}$ and replace every $\lceil\frac{16\log n}{\epsilon^{2}}\rceil$ by $\lceil\frac{16f(n)}{\epsilon^{2}}\rceil$, and it can be shown that the quantum query complexity of min-entropy estimation is $\tilde{O}\big(e^{(\frac{3}{4}+o(1))\cdot f(n)}\big)$, whereas the best classical algorithm takes $\tilde{\Theta}(e^{f(n)})$ queries. In this case, we obtain a $(\frac{3}{4}+o(1))$-quantum speedup, but the classical query complexity is already small ($e^{\sqrt{\log n}}=n^{1/\sqrt{\log n}}=o(n^{c})$ for any $c>0$).
\end{remark}


\section{0-R\'{e}nyi entropy estimation}\label{sec:support}
\hd{Motivations}. Estimating the support size of distributions (i.e., the 0-R{\'e}nyi entropy) is also important in various fields, ranging from vocabulary size estimation \cite{efron1976estimating,thisted1987did}, database attribute variation \cite{haas1995sampling}, password and security \cite{florencio2007large}, diversity study in microbiology \cite{kroes1999bacterial,paster2001bacterial,hughes2001counting}, etc. The study of support estimation was initiated by naturalist Corbet in 1940s, who spent two years at Malaya for trapping butterflies and recorded how many times he had trapped various butterfly species. He then asked the leading statistician at that time, Fisher, to predict how many new species he would observe if he returned to Malaya for another two years of butterfly trapping. Fisher answered by alternatively putting plus or minus sign for the number of species that showed up one, two, three times, and so on, which was proven to be an unbiased estimator \cite{fisher1943relation}.

Formally, assuming $n$ independent samples are drawn from an unknown distribution, the goal of \cite{fisher1943relation} is to estimate the number of hitherto unseen symbols that would be observed if $t\cdot n$ ($t$ being a pre-determined parameter) additional independent samples were collected from the same distribution. Reference \cite{fisher1943relation} solved the case $t=1$, which was later improved to $t\leq 1$ \cite{good1956number} and $t=O(\log n)$ \cite{orlitsky2016optimal}; the last work also showed that $t=\Theta(\log n)$ is the largest possible range to give an estimator with provable guarantee.

However, such estimation always assumes $n$ samples; a more natural question is, \emph{can we estimate the support of a distribution per se?} Specifically, given a discrete distribution $p$ over a finite set $X$ where $p_x$ denotes the probability of $x \in X$, can we estimate its \emph{support}, defined by
\begin{align}\label{eqn:support-definition}
  \supp(p):= |\{x: x\in X,\,p_{x}>0\}|,
\end{align}
with high precision and success probability?

Unfortunately, this is impossible in general because elements with negligible but nonzero probability will be very unlikely to appear in the samples, while still contribute to $\supp(p)$. As an evidence, $\supp(p)$ is the exponent of the 0-R{\'e}nyi entropy of $p$, but the sample complexity of $\alpha$-R{\'e}nyi entropy goes to infinity when $\alpha\rightarrow 0^{+}$ by \thm{lower-bound-rewrite}, both classically and quantumly.

To circumvent this difficulty, two related properties have been considered \emph{as an alternative to estimate 0-R{\'e}nyi entropy}:
\begin{itemize}
\item \emph{Support coverage}: $S_{n}(p):=\sum_{x\in X}\big(1-(1-p_{x})^{n}\big)$, the expected number of elements observed when taking $n$ samples. To estimate $S_{n}(p)$ within $\pm\epsilon n$, \cite{good1956number} showed that $n/2$ samples from $p$ suffices for any constant $\epsilon$; recently, \cite{zou2016quantifying} improved the sample complexity to $O\big(\frac{n}{\log n}\big)$, and \cite{orlitsky2016optimal,acharya2017unified} also considered the dependence in $\epsilon$ by showing that $\Theta\big(\frac{n}{\log n}\cdot\log\frac{1}{\epsilon}\big)$ is a tight bound, as long as $\epsilon=\Omega(n^{-0.2})$.

\item \emph{Support size}: $\supp(p)$, under the assumption that for any $x\in X$, $p_{x}=0$ or $p_{x}\geq 1/m$ for some given $m\in\N$. Reference \cite{raskhodnikova2009strong} proposed the problem and gave a lower bound $\Omega(m^{1-o(1)})$, and \cite{valiant2011estimating} gave an upper bound $O\big(\frac{m}{\log m}\cdot\frac{1}{\epsilon^{2}}\big)$. Recently, \cite{wu2015chebyshev} and \cite{orlitsky2016optimal} both proved that $\Theta\big(\frac{m}{\log m}\cdot\log^{2}\frac{1}{\epsilon}\big)$ is the tight bound for the problem (both optimal in $m$ and $\epsilon$).
\end{itemize}

Quantumly, we give upper and lower bounds on both support coverage and support size estimation, summarized in \tab{support}.

\begin{table}
\centering
\resizebox{0.8\columnwidth}{!}{%
\begin{tabular}{|c|c|c|c|}
\hline
Problem & classical bounds & quantum bounds ({\color{red} this paper}) \\ \hline\hline
Support coverage & $\Theta\big(\frac{n}{\log n}\cdot\log\frac{1}{\epsilon}\big)$ \cite{orlitsky2016optimal,acharya2017unified} [$\epsilon=\Omega(n^{-0.2})$] & $\tilde{O}\big(\frac{\sqrt{n}}{\epsilon^{1.5}}\big)$, $\Omega\Big(\frac{n^{1/3}}{\epsilon^{1/6}}\Big)$ \\ \hline
Support size & $\Theta\big(\frac{m}{\log m}\cdot\log^{2}\frac{1}{\epsilon}\big)$ \cite{wu2015chebyshev,orlitsky2016optimal} [$\epsilon=\Omega(1/m)$] & $\tilde{O}\big(\frac{\sqrt{m}}{\epsilon^{1.5}}\big)$, $\Omega\Big(\frac{m^{1/3}}{\epsilon^{1/6}}\Big)$ \\ \hline
\end{tabular}
}
\vspace{1mm}
\caption{Summary of the classical and quantum query complexity of support coverage and size estimation.}
\label{tab:support}
\end{table}

\hd{Support coverage estimation.} We give the following upper bound on support coverage estimation; its lower bound is given in \prop{coverage-lower-bound}.
\begin{algorithm}[htbp]
\SetKwFor{ForSubroutine}{Regard the following subroutine as $\mathcal{A}$:}{}{endfor}
\ForSubroutine{}
    {Draw a sample $i\in X$ according to $p$\;
    Use \textbf{EstAmp} with $M=2^{\lceil\log_{2}(\sqrt{n/\epsilon})\rceil}$ queries to obtain an estimation $\tilde{p}_{i}$ of $p_{i}$\;
    Output $\tilde{x}_{i}=\frac{1-(1-\tilde{p}_{i})^{n}}{\tilde{p}_{i}}$ if $\tilde{p}_{i}\neq 0$; otherwise, output $n$\;
}
Use $\mathcal{A}$ for $\Theta\big(\frac{1}{\epsilon}\log^{3/2}\big(\frac{1}{\epsilon}\big)\log\log\big(\frac{1}{\epsilon}\big)\big)$ executions in \thm{Monte-Carlo} and output an estimation $\tilde{S}_{n}(p)$ of $S_{n}(p)$\;
\caption{Estimate the support coverage $S_{n}(p)$.}
\label{algo:coverage}
\end{algorithm}

\begin{theorem}\label{thm:coverage-upper-bound}
\algo{coverage} approximates $\frac{S_{n}(p)}{n}:=\frac{\sum_{x\in X}(1-(1-p_{x})^{n})}{n}$ within an additive error $0<\epsilon\leq O(1)$ with success probability at least $2/3$ using $\tilde{O}\big(\frac{\sqrt{n}}{\epsilon^{1.5}}\big)$ quantum queries to $p$.
\end{theorem}

\begin{proof} We prove this theorem in two steps. The \textbf{first} step is to show that the expectation of the subroutine $\mathcal{A}$'s output (denoted $\tilde{E}:=\sum_{i\in X}p_{i}\cdot\frac{1-(1-\tilde{p}_{i})^{n}}{\tilde{p}_{i}}$) satisfies $|\tilde{E}-E|=O(\epsilon n)$, where $E:=\sum_{i\in X}p_{i}\cdot \frac{1-(1-p_{i})^{n}}{p_{i}}=S_{n}(p)$.

To achieve this, it suffices to prove that for each $i\in X$,
\begin{align}\label{eqn:supp-coverage}
\E\Big[\Big|\frac{1-(1-\tilde{p}_{i})^{n}}{\tilde{p}_{i}}-\frac{1-(1-p_{i})^{n}}{p_{i}}\Big|\Big]=O(\epsilon n).
\end{align}
We write $p_{i}=\sin^{2}(\theta_{i}\pi)$. Assume $k\in\Z$ such that $k\leq M\theta_{i}<k+1$. By \thm{quantum-counting}, for any $l\in\{1,2,\ldots,\max\{k-1,M-k-1\}\}$, the output of \textbf{EstAmp} taking $M$ queries satisfies
\begin{align}\label{eqn:supp-coverage-tail}
\Pr\Big[\tilde{p}_{i}=\sin^{2}\big(\frac{(k+l+1)\pi}{M}\big)\Big],\ \Pr\Big[\tilde{p}_{i}=\sin^{2}\big(\frac{(k-l-1)\pi}{M}\big)\Big]\leq \frac{1}{4l^{2}}.
\end{align}
We first consider the case when $\tilde{p}_{i}>p_{i}$, and $\tilde{p}_{i}=\sin^{2}\big(\frac{(k+l+1)\pi}{M}\big)$ for some $l\in\N$. For convenience, denote $f(x)=\frac{1-(1-x)^{n}}{x}$ where $x\in(0,1]$. Because
\begin{align}
f'(x)=\frac{nx(1-x)^{n-1}+(1-x)^{n}-1}{x^{2}}\leq\frac{nx+(1-nx)-1}{x^{2}}=0,
\end{align}
$f$ is a decreasing function on $(0,1]$. Therefore,
\begin{align}
&\Big|\frac{1-(1-\tilde{p}_{i})^{n}}{\tilde{p}_{i}}-\frac{1-(1-p_{i})^{n}}{p_{i}}\Big|\leq \sin^{2}\frac{k\pi}{M}\cdot\frac{1-\big(1-\sin^{2}\frac{(k+l+1)\pi}{M}\big)^{n}}{\sin^{2}\frac{(k+l+1)\pi}{M}}-\Big(1-\big(1-\sin^{2}\frac{k\pi}{M}\big)^{n}\Big) \\
&\quad\qquad\qquad= \sin^{2}\frac{k\pi}{M}\cdot\frac{1-\cos^{2n}\frac{(k+l+1)\pi}{M}}{\sin^{2}\frac{(k+l+1)\pi}{M}}-\Big(1-\cos^{2n}\frac{k\pi}{M}\Big) \\
&\quad\qquad\qquad= \frac{\big(\sin^{2}\frac{k\pi}{M}-\sin^{2}\frac{(k+l+1)\pi}{M}\big)+\big(\sin^{2}\frac{(k+l+1)\pi}{M}\cos^{2n}\frac{k\pi}{M}-\sin^{2}\frac{k\pi}{M}\cos^{2n}\frac{(k+l+1)\pi}{M}\big)}{\sin^{2}\frac{(k+l+1)\pi}{M}}. \label{eqn:supp-coverage-1}
\end{align}
By Taylor expansion, we have
\begin{align}\label{eqn:supp-coverage-2}
\sin^{2}\frac{k\pi}{M}=\frac{k^{2}\pi^{2}}{M^{2}}+O\Big(\frac{k^{6}}{M^{6}}\Big),\quad \sin^{2}\frac{(k+l+1)\pi}{M}=\frac{(k+l+1)^{2}\pi^{2}}{M^{2}}+O\Big(\frac{(k+l)^{6}}{M^{6}}\Big),
\end{align}
and
\begin{align}\label{eqn:supp-coverage-3}
\cos^{2n}\frac{k\pi}{M}=\Big(1-\frac{k^{2}\pi^{2}}{2M^{2}}+O\big(\frac{k^{4}}{M^{4}}\big)\Big)^{2n}=\Big(1-\frac{k^{2}\pi^{2}\epsilon}{2n}+O\big(\frac{k^2\epsilon^{2}}{n^{2}}\big)\Big)^{2n}=1-k^{2}\pi^{2}\epsilon+O\big(\epsilon^{2}k^{2}\big);
\end{align}
similarly
\begin{align}\label{eqn:supp-coverage-4}
\cos^{2n}\frac{(k+l+1)\pi}{M}=1-(k+l+1)^{2}\pi^{2}\epsilon+O\big(\epsilon^{2}(k+l)^{2}\big).
\end{align}
Plugging \eqn{supp-coverage-2}, \eqn{supp-coverage-3}, and \eqn{supp-coverage-4} into \eqn{supp-coverage-1} and noticing that the tail in \eqn{supp-coverage-2} has $1/M^{6}$, much smaller than that of \eqn{supp-coverage-3} and \eqn{supp-coverage-4}, we have
\begin{align}
&\Big|\frac{1-(1-\tilde{p}_{i})^{n}}{\tilde{p}_{i}}-\frac{1-(1-p_{i})^{n}}{p_{i}}\Big| \nonumber \\
&\leq \frac{\big(\frac{k^{2}\pi^{2}}{M^{2}}-\frac{(k+l+1)^{2}\pi^{2}}{M^{2}}\big)+\big(\frac{(k+l+1)^{2}\pi^{2}}{M^{2}}(1-k^{2}\pi^{2}\epsilon+O(\epsilon^{2}k^{2}))-\frac{k^{2}\pi^{2}}{M^{2}}(1-(k+l+1)^{2}\pi^{2}\epsilon+O(\epsilon^{2}(k+l)^{2}))\big)}{\frac{(k+l+1)^{2}\pi^{2}}{M^{2}}} \\
&= \frac{\big(k^{2}-(k+l+1)^{2}\big)+(k+l+1)^{2}(1-k^{2}\pi^{2}\epsilon)-k^{2}(1-(k+l+1)^{2}\pi^{2}\epsilon)}{(k+l+1)^{2}} \nonumber \\
&\qquad+O(\epsilon^{2}k^{2})-\frac{k^{2}}{(k+l+1)^{2}}O(\epsilon^{2}(k+l)^{2}) \\
&= 0+O\big(\epsilon^{2}(k+l)^{2}\big) \\
&\leq O(\epsilon n), \label{eqn:supp-coverage-5}
\end{align}
where \eqn{supp-coverage-5} holds because $k+l\leq M$ and $M=\Theta(\sqrt{n/\epsilon})$. Similarly, for the case $\tilde{p}_{i}<p_{i}$, we have
\begin{align}
\Big|\frac{1-(1-\tilde{p}_{i})^{n}}{\tilde{p}_{i}}-\frac{1-(1-p_{i})^{n}}{p_{i}}\Big|\leq O(\epsilon n). \label{eqn:supp-coverage-6}
\end{align}
In all, summing all $l\in\{1,2,\ldots,\max\{k-1,M-k-1\}\}$ in cases $\tilde{p}_{i}>p_{i}$ and $\tilde{p}_{i}<p_{i}$ and by \eqn{supp-coverage-tail}, the expectation of the deviation in \eqn{supp-coverage} is at most
\begin{align}
\E\Big[\Big|\frac{1-(1-\tilde{p}_{i})^{n}}{\tilde{p}_{i}}-\frac{1-(1-p_{i})^{n}}{p_{i}}\Big|\Big]\leq 2\cdot\sum_{l=1}^{M}\frac{1}{4l^{2}}\cdot O(\epsilon n)\leq\frac{\pi^{2}}{3}\cdot O(\epsilon n)=O(\epsilon n).
\end{align}
Therefore, \eqn{supp-coverage} follows and $|\tilde{E}-E|=O(\epsilon n)$. By rescaling $M$ by a constant, without loss of generality we have $|\tilde{E}-E|\leq \epsilon n/2$.

The \textbf{second} step is to bound the variance of the random variable, which is
\begin{align}
\sum_{i\in X}p_{i}\cdot\Big(\frac{1-(1-\tilde{p}_{i})^{n}}{\tilde{p}_{i}}\Big)^{2}-\Big(\sum_{i\in X}p_{i}\cdot\frac{1-(1-\tilde{p}_{i})^{n}}{\tilde{p}_{i}}\Big)^{2} \leq \sum_{i\in X}p_{i}\cdot n^{2} = n^{2},
\end{align}
because $1-(1-\tilde{p}_{i})^{n}\leq n\tilde{p}_{i}$ by $0\leq \tilde{p}_{i}\leq 1$. As a result of \thm{Monte-Carlo}, we can approximate $\tilde{E}$ up to additive error $\epsilon n/2$ with failure probability at most $1/3$ using
\begin{align}
O\Big(\frac{n}{\epsilon n}\log^{3/2}\Big(\frac{n}{\epsilon n}\Big)\log\log\Big(\frac{n}{\epsilon n}\Big)\Big)\cdot 2^{\lceil\log_{2}(\sqrt{n/\epsilon})\rceil}=\tilde{O}\Big(\frac{\sqrt{n}}{\epsilon^{1.5}}\Big)
\end{align}
quantum queries. Together with $|\tilde{E}-E|\leq \epsilon n/2$, \algo{coverage} approximates $E=S_{n}(p)$ up to additive error $\epsilon n$ with failure probability at most $1/3$; in other words, \algo{coverage} approximates $E=S_{n}(p)/n$ up to $\epsilon$ with success probability at least $2/3$.
\end{proof}

\hd{Support size estimation.}  We give the following upper bound on support size estimation; its lower bound is given in \prop{size-lower-bound}.
\begin{algorithm}[htbp]
\SetKwFor{ForSubroutine}{Regard the following subroutine as $\mathcal{A}$:}{}{endfor}
Call \algo{coverage} with $n=\lceil m\log(2/\epsilon)\rceil$ and error $\frac{\epsilon}{2\log(2/\epsilon)}$, and denote the output as $\tilde{S}_{n}(p)$\;
Denote $\widetilde{\supp}(p):=\lceil\tilde{S}_{n}(p)\rceil$. Output $\widetilde{\supp}(p)$ as an estimation of $\supp(p)$\;
\caption{Estimate $\supp(p)$, under the promise that $p_{x}=0$ or $p_{x}\geq 1/m$ for any $x\in X$.}
\label{algo:size}
\end{algorithm}

\begin{theorem}\label{thm:size-upper-bound}
Under the promise that for any $x\in X$, $p_{x}=0$ or $p_{x}\geq 1/m$, \algo{size} approximates $\supp(p)/m$ within an additive error $0<\epsilon\leq O(1)$ with success probability at least $2/3$ using $\tilde{O}\big(\frac{\sqrt{m}}{\epsilon^{1.5}}\big)$ quantum queries to $p$.
\end{theorem}

\begin{proof}
For convenience, denote $X_{1/m}:=\{x\in X: p_{x}\geq 1/m\}$. Then $\supp(p)=|X_{1/m}|$ by the promise, and
\begin{align}
1-\frac{\epsilon}{2}\leq 1-(1-p_{x})^{n}\leq 1\qquad \forall\,x\in X_{1/m}; \\
1-(1-p_{x})^{n}=0\qquad \forall\,x\notin X_{1/m};
\end{align}
As a result,
\begin{align}\label{eqn:size-1}
S_{n}(p)=\sum_{x\in X}1-(1-p_{x})^{n}\in\Big[\big(1-\frac{\epsilon}{2}\big)\supp(p),\supp(p)\Big].
\end{align}
Furthermore, by the correctness of \algo{coverage}, with probability at least $2/3$ we have
\begin{align}
|\tilde{S}_{n}(p)-S_{n}(p)|\leq \frac{\epsilon}{2\log(2/\epsilon)}\cdot n=\frac{m\epsilon}{2}.
\end{align}
Together with \eqn{size-1},
\begin{align}
\tilde{S}_{n}(p)\in\Big[\big(1-\frac{\epsilon}{2}\big)\supp(p)-\frac{m\epsilon}{2},\supp(p)+\frac{m\epsilon}{2}\Big]\subseteq \Big[\supp(p)-m\epsilon,\supp(p)+\frac{m\epsilon}{2}\Big].
\end{align}
Therefore, with probability at least $2/3$, $\frac{\lceil\tilde{S}_{n}(p)\rceil}{m}$ approximates $\frac{\supp(p)}{m}$ up to $\epsilon$ with success probability at least $2/3$.
\end{proof}

\section{Quantum lower bounds}\label{sec:lower-bound}
In this section, we prove \thm{lower-bound}, which is rewritten below:
\begin{theorem}\label{thm:lower-bound-rewrite}
Any quantum algorithm that approximates $H_\alpha(p)$ of distribution $p$ on $\range{n}$ within additive error $\epsilon$ with success probability at least 2/3 must use
\begin{itemize}
   \item $\Omega(n^{\frac{1}{3}}/\epsilon^{\frac{1}{6}})$ quantum queries when $\alpha=0$, assuming $1/n\leq\epsilon\leq 1$.
   \item $\tilde{\Omega}(n^{\frac{1}{7\alpha}-o(1)}/\epsilon^{\frac{2}{7}})$ quantum queries when $0<\alpha<\frac{3}{7}$.
   \item $\Omega(n^{\frac{1}{3}}/\epsilon^{\frac{1}{6}})$ quantum queries when $\frac{3}{7}\leq\alpha\leq 3$, assuming $1/n\leq\epsilon\leq 1$.
   \item $\Omega(n^{\frac{1}{2}-\frac{1}{2\alpha}}/\epsilon)$ quantum queries when $3\leq\alpha<\infty$.
   \item $\Omega(\sqrt{n}/\epsilon)$ quantum queries when $\alpha=\infty$.
\end{itemize}
\end{theorem}

Because we use different techniques for different ranges of $\alpha$, we divide the proofs into three categories.

\subsection{Reduction from classical lower bounds ($0<\alpha<\frac{3}{7}$)}
We prove that the quantum lower bound when $0<\alpha<\frac{3}{7}$ is indeed $\Omega(n^{\frac{1}{7\alpha}-o(1)}/\epsilon^{\frac{2}{7}})$, as claimed in \thm{lower-bound-rewrite}.
\begin{proof}
First, by \cite{acharya2017estimating}, we know that $\Omega(n^{\frac{1}{\alpha}-o(1)}/\epsilon^{2})$ is a lower bound on the classical query complexity of $\alpha$-R{\'e}nyi entropy estimation. On the other hand, reference \cite{aaronson2014need} shows that for any problem that is invariant under permuting inputs and outputs and that has sufficiently many outputs, the quantum query complexity is at least the seventh root of the classical randomized query complexity (up to poly-logarithmic factors). Our query oracle $O_{p}\colon \range{S}\rightarrow\range{n}$ has $n$ outputs with tend to infinity when $n$ is large; the distribution $p$ is invariant under permutations on $\range{S}$ since $p_{i}=|\{s\in\range{S}: O_{p}(s)=i\}|/S$ is invariant for all $i$; R{\'e}nyi entropy is invariant under permutations on $\range{n}$ since it does not depend on the order of $p_{i}$. Therefore, our problem satisfies the requirements from \cite{aaronson2014need}, and $\Omega(n^{\frac{1}{7\alpha}-o(1)}/\epsilon^{\frac{2}{7}})$ is a lower bound on the quantum query complexity of $\alpha$-R{\'e}nyi entropy estimation.
\end{proof}

\subsection{Exploitation of the collision lower bound ($\alpha=0$ and $\frac{3}{7}\leq\alpha\leq 3$)}

We prove lower bounds on entropy estimation by further exploiting the famous collision lower bound \cite{aaronson2004quantum,kutin2005quantum}. First, we define the following problem:

\begin{definition}[$l$-pairs distinctness]\label{defn:l-pair-distinctness}
Given positive integers $n$ and $l$ such that $1\leq l\leq n/2$, and a function $f\colon\range{n}\rightarrow\range{n}$.  Under the promise that either $f$ is 1-to-1 or their exists $l$ pairwise different pairs $(x_{i_{1}},y_{i_{1}}),\ldots,(x_{i_{l}},y_{i_{l}})\in\range{n}\times\range{n}$ such that $x_{i_{j}}\neq y_{i_{j}}$ but $f(x_{i_{j}})=f(y_{i_{j}})$ for all $j\in\range{l}$, the $l$-pairs distinctness problem is to determine which is the case, with success probability at least $2/3$.
\end{definition}

Note that when $l=1$, $l$-pairs distinctness reduces to the element distinctness problem, whose quantum query complexity is $\Theta(n^{2/3})$ \cite{ambainis2007quantum,aaronson2004quantum}; when $l=n/2$, $l$-pairs distinctness reduces to the collision problem, whose quantum query complexity is $\Theta(n^{1/3})$ \cite{aaronson2004quantum,kutin2005quantum}. Inspired by the reduction from the collision lower bound to the element distinctness lower bound in \cite{aaronson2004quantum}, we prove a more general quantum lower bound for $l$-pairs distinctness:

\begin{proposition}\label{prop:l-pair-distinctness}
The quantum query complexity of $l$-pairs distinctness is at least $\Omega(n^{\alpha})$, where $l^{\alpha}=n^{\frac{2}{3}-\alpha}$.
\end{proposition}

\begin{proof}
Assume the contrary that the quantum query complexity of $l$-pairs distinctness is $o(n^{\alpha})$. Consider a function $f\colon\range{n}\rightarrow\range{n}$ that is promised to be either 1-to-1 or 2-to-1. By \cite{kutin2005quantum}, it takes $\Omega(n^{1/3})$ quantum queries to decide whether $f$ is 1-to-1 or 2-to-1.

Denote $\mathcal{S}$ to be a subset of $\range{n}$, where $|\mathcal{S}|=\lceil2\sqrt{nl}\rceil$ and the elements in $\mathcal{S}$ are chosen uniformly at random. If $f$ is 1-to-1, then $f$ restricted on $\mathcal{S}$, denoted $f|_{\mathcal{S}}$, is still 1-to-1 on $\mathcal{S}$. If $f$ is 2-to-1, denote the set of its images as $\{a_{1},\ldots,a_{n/2}\}$. For any $j\in\range{n/2}$, denote $X_{j}$ to be a binary random variable that equals to 1 when the collision pair of $a_{j}$ appears in $\mathcal{S}$, and equals to 0 otherwise. Then
\begin{align}
\Pr[X_{j}=1]=\frac{{|\mathcal{S}|\choose 2}}{{n\choose 2}},\ \ \Pr[X_{j}X_{k}=1]=\frac{{|\mathcal{S}|\choose 4}}{{n\choose 4}}\qquad\forall\,j,k\in\range{n},j\neq k.
\end{align}
Denote $X=\sum_{j=1}^{n/2}X_{j}$, which is the number of collision pairs in $\mathcal{S}$. By linearity of expectation,
\begin{align}
\E[X]=\frac{n}{2}\cdot \frac{|\mathcal{S}|(|\mathcal{S}|-1)}{n(n-1)}\gtrsim 2l.
\end{align}
On the other hand,
\begin{align}
\var[X]&=\E[X^{2}]-\E[X]^{2} \\
&=\sum_{j=1}^{n/2}\E[X_{j}]+\sum_{j\neq k}\E[X_{j}X_{k}]-\E[X]^{2} \\
&\leq\frac{n}{2}\cdot\frac{{|\mathcal{S}|\choose 2}}{{n\choose 2}}+\frac{n}{2}\Big(\frac{n}{2}-1\Big)\cdot\frac{{|\mathcal{S}|\choose 4}}{{n\choose 4}}-\frac{n^{2}}{4}\cdot\frac{{|\mathcal{S}|\choose 2}^{2}}{{n\choose 2}^{2}} \\
&\lesssim 2l.
\end{align}
Therefore, by Chebyshev's inequality,
\begin{align}
\Pr[X<l]\leq\Pr\big[X\leq \E[X]-2\sqrt{2l}\big]\leq 1/4.
\end{align}
In other words, with probability at least $3/4$, $f|_{\mathcal{S}}$ on $\mathcal{S}$ has at least $l$ collision pairs. By our assumption, it takes $o(|\mathcal{S}|^{\alpha})=o(n^{\alpha/2}\cdot n^{1/3-\alpha/2})=o(n^{1/3})$ quantum queries to decide whether $f|_{\mathcal{S}}$ is 1-to-1 or has $l$ collision pairs, which suffices to decide whether $f$ is 1-to-1 or 2-to-1. However, this contradicts with the $\Omega(n^{1/3})$ quantum lower bound for the collision problem \cite{kutin2005quantum}.
\end{proof}

\subsubsection{$\alpha=0$}
For 0-R{\'e}nyi entropy estimation, we use \prop{l-pair-distinctness} to give quantum lower bounds for both support coverage estimation and support size estimation (both defined in \sec{support}).

\begin{proposition}\label{prop:coverage-lower-bound}
The quantum query complexity of support coverage estimation is $\Omega\big(\frac{n^{1/3}}{\epsilon^{1/6}}\big)$, for all $\frac{1}{n}\leq\epsilon\leq\frac{1}{12}$.
\end{proposition}

\begin{proof}
Because $\frac{1}{n}\leq\epsilon\leq\frac{1}{12}$, we may denote $\epsilon=n^{r}$ where $r\in[-1,0]$.

Consider two distributions $p_{1}$ and $p_{2}$ encoded by $O_{p_{1}},O_{p_{2}}:\range{n}\rightarrow X$ ($S=n$ in \eqn{BHH-query-model}), where the nonzero probabilities in $p_{1}$ are $1/n$ for $n$ times, and the nonzero probabilities in $p_{2}$ are $2/n$ for $l=\lceil 6n\epsilon\rceil$ times and $1/n$ for $n-2l$ times. In other words, $O_{p_{1}}$ is injective, and $O_{p_{2}}$ has $l$ collision pairs but otherwise injective. On the one hand, by \prop{l-pair-distinctness}, it takes $\Omega(n^{\alpha})$ quantum queries to distinguish between $O_{p_{1}}$ and $O_{p_{2}}$, where
\begin{align}
l^{\alpha}=n^{\frac{2}{3}-\alpha}\ \Rightarrow\ \alpha=\frac{2}{3(2+r)}+O\Big(\frac{1}{\log n}\Big)\geq \frac{1}{3}-\frac{r}{6}\quad\forall\,r\in[-1,0].
\end{align}
As a result, $n^{\alpha}\geq n^{1/3-r/6}=\frac{n^{1/3}}{\epsilon^{1/6}}$.

On the other hand,
\begin{align}
\frac{S_{n}(p_{1})}{n}&=\frac{n\cdot\big(1-(1-1/n)^{n}\big)}{n}\approx 1-\frac{1}{e}; \\
\frac{S_{n}(p_{2})}{n}&=\frac{l\cdot\big(1-(1-2/n)^{n}\big)+(n-2l)\cdot\big(1-(1-1/n)^{n}\big)}{n}\approx 1-\frac{1}{e}-\frac{l(1-1/e)^{2}}{n}.
\end{align}

As a result,
\begin{align}
\Big|\frac{S_{n}(p_{1})}{n}-\frac{S_{n}(p_{2})}{n}\Big|\approx\frac{l(1-1/e)^{2}}{n}>2\epsilon.
\end{align}
Therefore, if a quantum algorithm can estimate support coverage with error $\epsilon$, it can distinguish between $p_{1}$ and $p_{2}$ with success probability at least $2/3$. In conclusion, the quantum query complexity of support coverage estimation is $\Omega\big(\frac{n^{1/3}}{\epsilon^{1/6}}\big)$.
\end{proof}

Similar to the proof of \prop{coverage-lower-bound}, we can prove (with details omitted):
\begin{proposition}\label{prop:size-lower-bound}
The quantum query complexity of support size estimation is $\Omega\big(\frac{m^{1/3}}{\epsilon^{1/6}}\big)$, for all $\frac{1}{m}\leq\epsilon\leq\frac{1}{4}$.
\end{proposition}

\subsubsection{$\frac{3}{7}\leq\alpha\leq 3$}
Using \prop{l-pair-distinctness}, we show that the quantum query complexity of entropy estimation when $\frac{3}{7}\leq\alpha\leq 3$ is also $\Omega(n^{\frac{1}{3}}/\epsilon^{\frac{1}{6}})$.

\begin{proof}
We consider the case $\alpha=1$, i.e., Shannon entropy estimation; the proof for other $\alpha\in[\frac{3}{7},3]$ is basically identical.

Consider two distributions $p_{1}$ and $p_{2}$ encoded by $O_{p_{1}},O_{p_{2}}:\range{n}\rightarrow X$ ($S=n$ in \eqn{BHH-query-model}), where the nonzero probabilities in $p_{1}$ are $1/n$ for $n$ times, and the nonzero probabilities in $p_{2}$ are $2/n$ for $l=\lceil n\epsilon/\log 2\rceil$ times and $1/n$ for $n-2l$ times. In other words, $O_{p_{1}}$ is injective, and $O_{p_{2}}$ has $l$ collision pairs but otherwise injective. On the one hand, similar to the proof of \prop{coverage-lower-bound}, it takes it takes $\Omega(n^{1/3-\epsilon/6})$ quantum queries to distinguish between $O_{p_{1}}$ and $O_{p_{2}}$.

On the other hand,
\begin{align}
H(p_{1})&=n\cdot\frac{1}{n}\log n=\log n; \\
H(p_{2})&=l\cdot\frac{2}{n}\log\frac{n}{2}+(n-2l)\cdot\frac{1}{n}\log n=\log n-\frac{2l}{n}\log 2.
\end{align}

As a result,
\begin{align}
|H(p_{1})-H(p_{2})|=\frac{2l}{n}\log 2\geq 2\epsilon.
\end{align}
Therefore, if a quantum algorithm can estimate support coverage with error $\epsilon$, it can distinguish between $p_{1}$ and $p_{2}$ with success probability at least $2/3$. In conclusion, the quantum query complexity of support coverage estimation is $\Omega\big(\frac{n^{1/3}}{\epsilon^{1/6}}\big)$.
\end{proof}

\subsection{Polynomial method ($3\leq\alpha\leq\infty$)}
We use the polynomial method \cite{beals2001quantum} to show quantum lower bounds for entropy estimation when $3\leq\alpha\leq\infty$. Inspired by the symmetrization technique in \cite{kutin2005quantum}, we obtain a bivariate polynomial whose degree is at most two times the corresponding quantum query complexity. Next, similar to \cite{nayak1999quantum}, we apply Paturi's lemma \cite{paturi1992degree} to give a lower bound on the degree of the polynomial. To be more specific, we prove:
\begin{proposition}\label{prop:lower-bound-1}
The quantum query complexity of estimating min-entropy with error $\epsilon$ is $\Omega(\frac{\sqrt{n}}{\epsilon})$.
\end{proposition}
\begin{proposition}\label{prop:lower-bound-2}
When the constant $\alpha$ satisfies $1<\alpha<\infty$, the quantum query complexity of estimating $\alpha$-R{\'e}nyi entropy with error $\epsilon$ is $\Omega(\frac{\alpha n^{\frac{1}{2}-\frac{1}{2\alpha}}}{\epsilon})$.
\end{proposition}

Without loss of generality, we assume that the oracle $O_{p}$ in \eqn{BHH-query-model} satisfies $n|S$, otherwise consider the oracle $O'_{p}\colon \range{Sn}\rightarrow\range{n}$ such that $O'_{p}(s+Sl)=O_{p}(s)$ for all $s\in\range{S}$ and $l\in\range{n}$; this gives an oracle for the same distribution.

We consider the special case where the probabilities $\{p_{i}\}_{i=1}^{n}$ takes at most two different values; to integrate the probabilities, we assume the existence of two integers $c,d$ where $c\in\{1,\ldots,n-1\}$, such that $p_{i}=\frac{1}{n}-\frac{d}{S}$ for $n-c$ different $i$'s in $\{1,\ldots,n\}$, and $p_{i}=\frac{1}{n}+\frac{(n-c)d}{cS}$ for the other $c$ $i$'s in $\{1,\ldots,n\}$.

\begin{proof}[Proof of \prop{lower-bound-1}]
Following the symmetrization technique in \cite{kutin2005quantum}, we obtain a bivariate polynomial $Q(c,d)$ where such that the degree of $Q$ is at most two times the query complexity of min-entropy estimation, and:
\begin{itemize}
\item $c\in\{1,\ldots,n-1\}$ and $d\in\{-\big\lfloor\frac{Sc}{n(n-c)}\big\rfloor,\ldots,\frac{S}{n}\}$. This is because $p_{i}\geq 0$ for all $i\in\range{n}$.
\item $0\leq Q(c,d)\leq 1$ if $c|nd$. Only if $c|nd$, $S\cdot \big(\frac{1}{n}+\frac{(n-c)d}{cS}\big)$ is an integer and the distribution $\{p_{i}\}_{i=1}^{n}$ is valid under our model in \eqn{BHH-query-model}.
\end{itemize}
Furthermore, we consider the property testing problem of determining whether $\max_{i}p_{i}=\frac{1}{n}$ or $\max_{i}p_{i}\geq\frac{1+\epsilon}{n}$, where the accept probability should be at most $1/3$ for the former case and at least $2/3$ for the latter case. As a result,
\begin{itemize}
\item $0\leq Q(c,0)\leq 1/3$: In this case, $p_{i}=\frac{1}{n}$ for all $i\in\range{n}$.
\item $2/3\leq Q(c,d)\leq 1$ if $c|nd$, $\frac{(n-c)d}{Sc}\geq\frac{\epsilon}{n}$: In this case, $\exists\, i$ such that $p_{i}=\frac{1}{n}+\frac{(n-c)d}{cS}\geq \frac{1+\epsilon}{n}$.
\item $2/3\leq Q(c,d)\leq 1$ if $c|nd$, $d\leq -\frac{\epsilon S}{n}$: In this case, $\exists\, i$ such that $p_{i}=\frac{1}{n}-\frac{d}{S}\geq \frac{1+\epsilon}{n}$.
\end{itemize}
Therefore, we have
\begin{itemize}
\item $0\leq Q(1,d)\leq 1$ for $d\in\{-\big\lfloor\frac{S}{n(n-1)}\big\rfloor,\ldots,\frac{S}{n}\}$;
\item $0\leq Q(1,0)\leq 1/3$;
\item $2/3\leq Q(1,d)\leq 1$ for $d\in\{-\big\lfloor\frac{S}{n(n-1)}\big\rfloor,\ldots,-\big\lceil\frac{\epsilon S}{n}\big\rceil\}\cup\{\big\lceil\frac{\epsilon S}{n(n-1)}\big\rceil,\ldots,\frac{S}{n}\}$.
\end{itemize}
Using Paturi's lower bound \cite{paturi1992degree}, we have
\begin{align}
\deg_{d}Q(1,d)\geq\Omega\Bigg(\frac{\sqrt{\big\lfloor\frac{S}{n(n-1)}\big\rfloor\cdot\frac{S}{n}}}{\big\lceil\epsilon S/n(n-1)\big\rceil}\Bigg)=\Omega\Big(\frac{\sqrt{n}}{\epsilon}\Big).
\end{align}
Therefore, $\deg Q(c,d)\geq\deg_{d}Q(1,d)=\Omega(\sqrt{n}/\epsilon)$.
\end{proof}

\begin{proof}[Proof of \prop{lower-bound-2}]
The proof is similar to that of \prop{lower-bound-1}. Following the symmetrization technique, we still obtain a bivariate polynomial $Q(c,d)$ where such that the degree of $Q$ is at most two times the query complexity of min-entropy estimation, and $c\in\{1,\ldots,n-1\}, d\in\{-\big\lfloor\frac{Sc}{n(n-c)}\big\rfloor,\ldots,\frac{S}{n}\}$, $0\leq Q(c,d)\leq 1$ if $c|nd$. Furthermore, we consider the property testing problem of determining whether $\sum_{i\in\range{n}}p_{i}^{\alpha}\leq\frac{2}{n^{\alpha-1}}$ or $\sum_{i\in\range{n}}p_{i}^{\alpha}\geq\frac{2+2\epsilon}{n^{\alpha-1}}$, where the accept probability should be at most $1/3$ for the former case and at least $2/3$ for the latter case. We also assume $c=1$. On the one hand, when $0\leq d\leq\big\lfloor\frac{n^{1/\alpha}-1}{n-1}\cdot\frac{S}{n}\big\rfloor$, we have $\frac{1}{n}+\frac{(n-1)d}{S}\leq\frac{1}{n^{1-1/\alpha}}$, and
\begin{align}
\sum_{i\in\range{n}}p_{i}^{\alpha}\leq\Big(\frac{1}{n^{1-1/\alpha}}\Big)^{\alpha}+(n-1)\Big(\frac{1}{n-1}\Big(1-\frac{1}{n^{1-1/\alpha}}\Big)\Big)^{\alpha}\leq\frac{2}{n^{\alpha-1}}.
\end{align}
On the other hand, because $(1+m)^{\alpha}\approx 1+m\alpha$ when $m=o(1)$, we have
\begin{align}
&\Big(\frac{1}{n^{1-1/\alpha}}+\frac{3\epsilon}{\alpha n^{1-1/\alpha}}\Big)^{\alpha}+(n-1)\Big(\frac{1}{n-1}\Big(1-\frac{1}{n^{1-1/\alpha}}-\frac{3\epsilon}{\alpha n^{1-1/\alpha}}\Big)\Big)^{\alpha}-\frac{2+2\epsilon}{n^{\alpha-1}} \nonumber \\
&\quad= \frac{1}{n^{\alpha-1}}\Big(1+\frac{3\epsilon}{\alpha}\Big)^{\alpha}+\frac{1}{(n-1)^{\alpha-1}}\Big(1-\frac{1}{n^{1-1/\alpha}}-\frac{3\epsilon}{\alpha n^{1-1/\alpha}}\Big)^{\alpha}-\frac{2+2\epsilon}{n^{\alpha-1}} \\
&\quad\approx \frac{1}{n^{\alpha-1}}\Big(1+3\epsilon+1-\frac{\alpha}{n^{1-1/\alpha}}-\frac{3\epsilon\alpha}{\alpha n^{1-1/\alpha}}-(2+2\epsilon)\Big) \\
&\quad= \frac{1}{n^{\alpha-1}}\Big(\epsilon-\frac{\alpha+3\epsilon}{n^{1-1/\alpha}}\Big) \\
&\quad\geq 0
\end{align}
for large enough $n$. As a result, when $d\geq\big\lceil\frac{(1+3\epsilon/\alpha)n^{1/\alpha}-1}{n-1}\cdot\frac{S}{n}\big\rceil$, we have $\sum_{i\in\range{n}}p_{i}^{\alpha}\geq\frac{2+2\epsilon}{n^{\alpha-1}}$.

Therefore, we have
\begin{itemize}
\item $0\leq Q(1,d)\leq 1$ for $d\in\{0,\ldots,\frac{S}{n}\}$;
\item $0\leq Q(1,d)\leq 1/3$ for $d\in\{0,\ldots,\big\lfloor\frac{n^{1/\alpha}-1}{n-1}\cdot\frac{S}{n}\big\rfloor\}$;
\item $2/3\leq Q(1,d)\leq 1$ for $d\in\{\big\lceil\frac{(1+3\epsilon/\alpha)n^{1/\alpha}-1}{n-1}\cdot\frac{S}{n}\big\rceil,\ldots,\frac{S}{n}\}$.
\end{itemize}
Using Paturi's lower bound \cite{paturi1992degree}, we have
\begin{align}
\deg_{d}Q(1,d)\geq\Omega\Bigg(\frac{\sqrt{\big\lfloor\frac{n^{1/\alpha}-1}{n-1}\cdot\frac{S}{n}\big\rfloor\big(\frac{S}{n}-\big\lfloor\frac{n^{1/\alpha}-1}{n-1}\cdot\frac{S}{n}\big\rfloor\big)}} {\big\lceil\frac{(3\epsilon/\alpha)n^{1/\alpha}}{n-1}\cdot\frac{S}{n}\big\rceil}\Bigg)=\Omega\Big(\frac{\alpha n^{\frac{1}{2}-\frac{1}{2\alpha}}}{\epsilon}\Big).
\end{align}
Therefore, $\deg Q(c,d)\geq\deg_{d}Q(1,d)=\Omega(\alpha n^{\frac{1}{2}-\frac{1}{2\alpha}}/\epsilon)$.
\end{proof}

Technically, our proofs only focus on the degree in $d$ for $c=1$, but in general it is possible to prove a better lower bound when analyzing the degree of the polynomial in $c$ and $d$ together. We leave this as an open problem.


\subsection*{Acknowledgements}
We thank Andrew M. Childs for discussions that inspired the proof of \thm{Renyi-integer}, and general suggestions on our manuscript; we thank Yanjun Han for introducing us classical references related to Shannon and R{\'e}nyi entropy estimation, in particular his papers \cite{han2016minimax,jiao2014maximum,jiao2015minimax}. We also thank anonymous reviewers for helpful comments on an earlier version of this paper.  TL acknowledges support from NSF CCF-1526380.


\providecommand{\bysame}{\leavevmode\hbox to3em{\hrulefill}\thinspace}


\appendix
\section{\thm{Monte-Carlo-multiplicative}: Multiplicative quantum Chebyshev inequality}\label{append:Monte-Carlo-multiplicative}

The main technique that we use is Lemma 4 in \cite{montanaro2015quantum}, which approximates a random variable with an additive error as long as its second-moment is bounded:
\begin{lemma}[Lemma 4 in \cite{montanaro2015quantum}]\label{lem:bounded-L2-norm}
Assume $\mathcal{A}$ is a quantum algorithm that outputs a random variable $X$. Then for $\epsilon$ where $0<\epsilon<1/2$ (multiplicative error), by using $O((1/\epsilon)\log^{3/2}(1/\epsilon)\log\log(1/\epsilon))$ executions of $\mathcal{A}$ and $\mathcal{A}^{-1}$, Algorithm 2 in \cite{montanaro2015quantum} outputs an estimate $\tilde{\E}[X]$ of $\E[X]$ such that\footnote{The original error probability in \eqn{bounded-L2-norm} is $1/5$, but it can be improved to $1/50$ by rescaling the parameters in Lemma 4 in \cite{montanaro2015quantum} up to a constant.}
\begin{align}\label{eqn:bounded-L2-norm}
\Pr\big[\big|\tilde{\E}[X]-\E[X]\big|\geq\epsilon(\sqrt{\E[X^{2}]}+1)^{2}\big]\leq 1/50.
\end{align}
\end{lemma}
\noindent
Based on \lem{bounded-L2-norm} and inspired by Algorithm 3 and Theorem 5 in \cite{montanaro2015quantum}, we propose \algo{Monte-Carlo-multiplicative}.
\vspace{0mm}

\begin{algorithm}[H]
\SetKwFor{ForSubroutine}{Regard the following subroutine as $\mathcal{A}$:}{}{endfor}
Run the algorithm that gives $a,b$ such that $\E[X]\in[a,b]$\;
Set $\mathcal{A'}=\mathcal{A}/\sigma b$\;
Run $\mathcal{A'}$ once and denote $\widetilde{m}$ to be the output. Set $\mathcal{B}=\mathcal{A'}-\widetilde{m}$\;
Let $\mathcal{B}_{-}$ be the algorithm that calls $\mathcal{B}$ once; if $\mathcal{B}$ outputs $x\geq 0$ then $\mathcal{B}_{-}$ outputs 0, and if $\mathcal{B}$ outputs $x<0$ then $\mathcal{B}_{-}$ outputs $x$. Similarly, let $\mathcal{B}_{+}$ be the algorithm such that if $\mathcal{B}$ outputs $x<0$ then $\mathcal{B}_{+}$ outputs 0, and if $\mathcal{B}$ outputs $x\geq 0$ then $\mathcal{B}_{+}$ outputs $x$\;
Apply \lem{bounded-L2-norm} to $-\mathcal{B}_{-}/6$ and $\mathcal{B}_{+}/6$ with error $\frac{\epsilon a}{48\sigma b}$ and failure probability $1/50$, and obtain estimates $\widetilde{\mu}_{-}$ and $\widetilde{\mu}_{+}$, respectively\;
Output $\tilde{\E}[X]=\sigma b(\widetilde{m}-6\widetilde{\mu}_{-}+6\widetilde{\mu}_{+})$\;
\caption{Estimate $\E[X]$ within multiplicative error $\epsilon$.}
\label{algo:Monte-Carlo-multiplicative}
\end{algorithm}

\begin{proof}[Proof of \thm{Monte-Carlo-multiplicative}]
Because $\var[X]\leq\sigma^{2}\E[X]^{2}\leq\sigma^{2}b^{2}$, by Chebyshev's inequality we have
\begin{align}
\Pr\big[\big|\widetilde{m}-\E[X/\sigma b]\big|\geq 4\big]\leq 1/16.
\end{align}
Therefore, with probability at least $15/16$ we have $|\widetilde{m}-\E[X/\sigma b]|\leq 4$. Denote $X_{B}=\frac{X}{\sigma b}-\widetilde{m}$, which is the random variable output by $\mathcal{B}$; $X_{B,+}:=\max\{X_{B},0\}$ is then the output of $\mathcal{B}_{+}$ and $X_{B,-}:=\min\{X_{B},0\}$ is the output of $\mathcal{B}_{-}$. Assuming $|\widetilde{m}-\E[X/\sigma b]|\leq 4$, we have
\begin{align}
\E[X_{B}^{2}]&=\E\Big[\Big(\big(\frac{X}{\sigma b}-\E\big[\frac{X}{\sigma b}\big]\big)+\big(\E\big[\frac{X}{\sigma b}\big]-\widetilde{m}\big)\Big)^{2}\Big] \\
&\leq 2\E\Big[\Big(\frac{X}{\sigma b}-\E\big[\frac{X}{\sigma b}\big]\Big)^{2}\Big]+2\E\Big[\Big(\E\big[\frac{X}{\sigma b}\big]-\widetilde{m}\Big)^{2}\Big] \\
&\leq 2(1^{2}+4^{2})=34.
\end{align}
Therefore, $\E\big[(X_{B}/6)^{2}\big]\leq 34/36<1$, hence $\E\big[(X_{B,+}/6)^{2}\big]<1$ and $\E\big[(-X_{B,-}/6)^{2}\big]<1$. By \lem{bounded-L2-norm}, we have
\begin{align}
\big|\widetilde{\mu}_{-}-\E[-X_{B,-}/6]\big|\leq\frac{\epsilon a}{12\sigma b}\qquad\text{and}\qquad \big|\widetilde{\mu}_{+}-\E[X_{B,+}/6]\big|\leq\frac{\epsilon a}{12\sigma b}
\end{align}
both with failure probability at most $1/50$. Because
\begin{align}
\E[X]=\sigma b\big(\widetilde{m}+\E[X_{B}]\big)=\sigma b\big(\widetilde{m}+\E[X_{B,+}]-\E[-X_{B,-}]\big),
\end{align}
with probability at least $15/16\cdot (1-1/50)^{2}>9/10$, we have
\begin{align}
\big|\tilde{\E}[X]-\E[X]\big|&\leq \sigma b\cdot\big(6\big|\widetilde{\mu}_{-}-\E[-X_{B,-}/6]\big|+6\big|\widetilde{\mu}_{+}-\E[X_{B,+}/6]\big|\big) \\
&\leq \sigma b\cdot 2\cdot 6\cdot\frac{\epsilon a}{12\sigma b}=\epsilon a\leq\epsilon\E(X).
\end{align}
\end{proof}

\end{document}